\documentclass[11pt,a4paper]{snamsart}

\usepackage[utf8]{inputenc}
\usepackage[T1]{fontenc}

\usepackage{mathrsfs}
\usepackage{tikz}

\definecolor{purple}{cmyk}{.51,.91,0,.34}
\usepackage[bookmarksdepth=2,
  colorlinks=true,
  linkcolor=black,
  citecolor=purple,
  pdfauthor={Libor Barto, Jakub Bulin, Andrei Krokhin, and Jakub Oprsal},
  pdftitle={Algebraic approach to promise constraint satisfaction}]%
  {hyperref}

\usepackage{mathrsfs}

\newtheorem{theorem}{Theorem}[section]
\newtheorem{lemma}[theorem]{Lemma}
\newtheorem{proposition}[theorem]{Proposition}
\newtheorem{corollary}[theorem]{Corollary}

\theoremstyle{definition}
\newtheorem{definition}[theorem]{Definition}

\theoremstyle{remark}

\newtheorem{example}[theorem]{Example}

\newtheorem{remark}[theorem]{Remark}

\numberwithin{equation}{section}

\newcommand{\vdotsequals}{\mathrel{\setbox0=\hbox{$\equals$}\makebox[\the\wd0]{\vdots}}}
\newcommand{\textto}{\unskip${}\to{}$}

\let\rel\mathbf         %
\let\clo\mathscr        %
\let\lang\mathcal       %
\let\tup\mathbf         %
\let\equals\approx      %

\let\epsilon\varepsilon
\let\meet\wedge

\DeclareMathOperator{\E}{\mathsf{E}}
\DeclareMathOperator{\Ref}{\mathsf{R}}
\DeclareMathOperator{\Pow}{\mathsf{P}}
\DeclareMathOperator{\Pfin}{\mathsf{P_{\mathrm{fin}}}}
\DeclareMathOperator{\ER}{\mathsf{ER}}
\DeclareMathOperator{\EPfin}{\mathsf{EP_{\mathrm{fin}}}}
\DeclareMathOperator{\ERPfin}{\mathsf{ERP_{\mathrm{fin}}}}

\DeclareMathOperator{\CSP}{\text{CSP}}
\DeclareMathOperator{\PCSP}{\text{PCSP}}
\DeclareMathOperator{\Pol}{Pol}

\DeclareMathOperator{\Ham}{Ham}

\newcommand{\GLC}{\textsc{GLC}}
\newcommand{\LC}{\textsc{LC}}
\newcommand{\GLLC}{\textsc{GLLC}}
\newcommand{\LLC}{\textsc{LLC}}
\newcommand{\MC}{\textsc{MC}}
\newcommand{\PMC}{\textsc{PMC}}
\newcommand{\LMC}{\textsc{LMC}}
\newcommand{\PLMC}{\textsc{PLMC}}
\let\ALC\MC
\let\PLC\PMC

\newcommand{\vc}[1]{\mathbf{#1}}

\newcommand{\Sat}{\textsc{Sat}}
\newcommand{\nae}{\text{\rm NAE}}
\newcommand{\NAE}{\rel H}

\newcommand{\NP}{\textsf{NP}}
\newcommand{\Ptime}{\textsf{P}}
\newcommand{\NL}{\textsf{NL}}
\newcommand{\Lspace}{\textsf{L}}
\newcommand{\AC}{\textsf{AC}}

\newcommand{\OneInThree}{\rel T}
\newcommand{\ttt}[1]{\langle #1 \rangle}
\newcommand{\rrr}[1]{[#1]}
\newcommand{\abs}[1]{\mathopen|#1\mathclose|}

\newcommand{\Proj}{\clo P}
\newcommand{\proj}{\ensuremath p}
\DeclareMathOperator{\ar}{ar}

\newcommand{\yes}{{\scshape yes}}
\newcommand{\no}{{\scshape no}}

\DeclareMathOperator{\BLP}{BLP}
\DeclareMathOperator{\LP}{LP}
\DeclareMathOperator{\AIP}{AIP}
\DeclareMathOperator{\IP}{IP}
\newcommand{\conv}{{\text{conv.}}}
\newcommand{\aff}{{\text{aff.}}}
 
\begin{document}
\author{Libor Barto}
\address{
  Department of Algebra, Fac. Math. \& Phys., Charles University,
  Pra\-gue, Czechia
}
\email{libor.barto@gmail.com}

\author{Jakub Bulín}
\address{
  Department of Algebra, Fac. Math. \& Phys., Charles University,
  Pra\-gue, Czechia
}
\email{jakub.bulin@mff.cuni.cz}

\author{Andrei Krokhin}
\address{
  Department of Computer Science, Durham University, UK
}
\email{andrei.krokhin@durham.ac.uk}

\author{Jakub Opršal}
\address{
  Department of Computer Science, Durham University, UK
}
\email{jakub.oprsal@durham.ac.uk}

\title{Algebraic approach to promise constraint satisfaction}
\date{\today}

\keywords{constraint satisfaction, promise problem, approximation, graph colouring, polymorphism}

\thanks{
  Preliminary versions of parts of this paper were published in the proceedings of STOC 2019 and LICS 2019 \cite{BKO19,Bar19}. Libor Barto has received funding from the European Research Council (ERC) under the European Unions Horizon 2020 research and innovation programme (Grant Agreement No.\ 771005, CoCoSym).
  Jakub Bul\'in was supported by the Austrian Science Fund project P29931, the Czech Science Foundation project 18-20123S, Charles University Research Centre program UNCE/SCI/022 and PRIMUS/SCI/12.
  Andrei Krokhin and Jakub Opr\v sal were supported by the UK EPSRC grant EP/R034516/1.
  Jakub Opr\v sal has also received funding from the European Research Council (ERC) under the European Unions Horizon 2020 research and innovation programme (Grant Agreement No.\ 681988, CSP-Infinity).}

  \begin{abstract}
    The complexity and approximability of the constraint satisfaction problem (CSP) has been actively studied over the last 20 years. A new version of the CSP, the promise CSP (PCSP) has recently been proposed, motivated by open questions about the approximability of variants of satisfiability and graph colouring. The PCSP significantly extends the standard decision CSP. The complexity of CSPs with a~fixed constraint language on a~finite domain has recently been fully classified, greatly guided by the algebraic approach, which uses polymorphisms --- high-dimensional symmetries of solution spaces --- to analyse the complexity of problems. The corresponding classification for PCSPs is wide open and includes some long-standing open questions, such as the complexity of approximate graph colouring, as special cases. 
    
    The basic algebraic approach to PCSP was initiated by Brakensiek and Guruswami, and in this paper we significantly extend it and lift it from concrete properties of polymorphisms to their abstract properties. We introduce a~new class of problems that can be viewed as algebraic versions of the (Gap) Label Cover problem, and show that every PCSP with a~fixed constraint language is equivalent to a~problem of this form. This allows us to identify a~``measure of symmetry'' that is well suited for comparing and relating the complexity of different PCSPs via the algebraic approach. We demonstrate how our theory can be applied by giving both general and specific hardness/tractability results. Among other things, we improve the state-of-the-art in approximate graph colouring by showing that, for any $k\geq 3$, it is $\NP$-hard to find a~$(2k-1)$-colouring of a~given $k$-colourable graph.
  \end{abstract}

  \maketitle

  \setcounter{tocdepth}{1}
  \tableofcontents

\section{Introduction}

What kind of inherent mathematical structure makes a~computational problem tractable, i.e., polynomial time solvable (assuming $\Ptime\ne\NP$)?
Finding a~general answer to this question is one of the fundamental goals of theoretical computer science. The \emph{constraint satisfaction problem (CSP)} and its variants are extensively used towards this ambitious goal for two reasons: on the one hand, the CSP framework is very general and includes a~wide variety of computational problems, and on the other hand, this framework has very rich mathematical structure providing an excellent laboratory both for complexity classification methods and for algorithmic techniques.

The basic aim in a~CSP is to decide whether there is an assignment of values from some domain $A$ to a~given set of variables, subject to constraints on the combinations of values which can be assigned simultaneously to certain specified subsets of variables. Important variants of the CSP include counting and optimisation (both exact and approximate) and extensions of the basic framework, e.g.\ by using real-valued functions instead of relations/predicates (to specify valued constraints) or allowing global constraints (see surveys in \cite{KZ17}). 
Since the basic CSP is \NP-complete (and, for other variants, as hard as it can be) in full generality, a~major line of research in the CSP focuses on identifying tractable cases and understanding the mathematical structure enabling tractability (see \cite{KZ17}).

One particular family of CSPs that receives a~great amount of attention consists of the CSPs with a~fixed \emph{constraint language}~\cite{FV98,KZ17}, i.e., with a~restricted set of types of constraints. Since constraints are usually given by relations, a~constraint language is simply a~set $\Gamma$ of relations on a~domain $A$. The restricted CSP where only relations from $\Gamma$ can specify constraints is denoted by $\CSP(\Gamma)$. Many computational problems, including various versions of logical satisfiability, graph colouring, and systems of equations can be represented in this form~\cite{FV98,KZ17}. It is well-known \cite{FV98} that the basic CSP can be cast as a~homomorphism problem from one relational structure to another (the latter is often called a~\emph{template}), and we will use this view.
Problems $\CSP(\Gamma)$ correspond to the case when the template structure is fixed.
There is an active line of research into CSPs with infinite $A$ (see e.g. surveys \cite{Bod08,BM17,Pin15}), but throughout this paper, we assume that $A$ is finite (unless specified otherwise).

In \cite{FV98}, Feder and Vardi conjectured that for each finite constraint language $\Gamma$, the (decision) problem $\CSP(\Gamma)$ is either in \Ptime{} or \NP-complete. This conjecture inspired a~very active research programme in the last 20 years, which recently culminated in two independent proofs of the conjecture, one by Bulatov \cite{Bul17} and the other by Zhuk \cite{Zhu17} (along with similar classification results for other CSP variants, e.g.\ \cite{Bul13,BK16,KKR17,TZ16}).
All of these proofs heavily use the so-called algebraic approach to the CSP. On a~very high level, this approach uses multivariate functions that preserve relations in a~constraint language (and hence solution sets of problem instances), called \emph{polymorphisms}. Thus polymorphisms can be seen as high-dimensional ``symmetries'' of solution sets. Roughly, lack of such symmetries implies hardness of the corresponding problem, while presence of symmetries implies tractability.
This approach was started in a~series of papers by Jeavons et al., e.g.\ \cite{JCG97,Jea98}, where  the key role of polymorphisms was established.
It was then taken to a~more abstract level in \cite{BJ01,BJK05}, where an~abstract view on polymorphisms was used, through universal algebras and varieties formed by algebras --- this allowed a~powerful machinery of structural universal algebra to be applied to the CSP. Another important general methodological improvement was \cite{BOP18}, where it was shown that special equations of simple form satisfied by polymorphisms govern the complexity of CSPs. Even though \cite{BOP18} did not impact specifically on the resolution of the Feder-Vardi conjecture, it strongly influenced the present paper.

A~new extended version of the CSP, the so-called \emph{promise constraint satisfaction problem} (\emph{PCSP}), has recently been introduced \cite{AGH17,BG16a,BG18}, motivated by open problems about (in)approximability for variants of SAT and graph colouring.
Roughly, this line of research in approximability concerns finding an approximately good  solution to an instance of a~(typically hard) problem when a~good solution is guaranteed to exist (see discussion and references in \cite{BG17}). Approximation can be understood in terms of relaxing constraints, or in terms of counting satisfied/violated constraints --- in this paper, we use the former.
Specifically, in the PCSP, each constraint in an instance has two relations: a~`strict' one, and a~`relaxed' one, and one needs to distinguish between the case when an instance has a~solution subject to the strict constraints and the case when it has no solution even subject to the relaxed constraints. One example of such a~problem (beyond CSPs) is the case when the only available strict relation is the disequality on a~$k$-element set and the corresponding relaxed relation is the disequality on a~$c$-element set (with $c\ge k$) --- the problem is then to distinguish $k$-colourable graphs from those that are not even $c$-colourable. This problem (and hence the problem of colouring a~given $k$-colourable graph with $c$ colours) has been conjectured \NP-hard, but the question in full generality is still open after more than 40 years of research.
We give more examples later.
Note that if the strict form and the relaxed form for each constraint coincide, then one gets the standard CSP, so the PCSP framework greatly generalises the CSP.

The problem of systematically investigating the complexity of PCSPs (with a~fixed constraint language) was  suggested in~\cite{AGH17,BG16a,BG18}. We remark that, beyond CSPs, the current knowledge of the complexity landscape of PCSPs is quite limited, and we do not even have analogues of full classification results for graph homomorphisms \cite{HN90} and Boolean CSPs \cite{Sch78} --- which were the most important basic special cases of CSP complexity classifications that inspired the Feder-Vardi conjecture.
The quest of complexity classification of PCSPs is of great interest for a~number of reasons. It brings together two very advanced methodologies: analysing the complexity of CSPs via algebra and the approximability of CSPs via PCP-based methodology, hence the possibility of fruitful cross-fertilisation and influence beyond the broad CSP framework.
It is perfect for further exploring the thesis that (high-dimensional) symmetries of solution spaces are relevant for complexity --- which is certainly true for most CSP-related problems, but may be applicable in a~wider context.
Finally, this quest includes long-standing open problems as special cases.

\subsubsection*{Related work.}
An accessible exposition of the algebraic approach to the CSP can be found in \cite{BKW17}, where many ideas and results leading to (but not including) the resolution \cite{Bul17,Zhu17} of the Feder-Vardi conjecture are presented. The volume \cite{KZ17} contains surveys about many aspects of the complexity and approximability of CSPs.

The first link between the algebraic approach and PCSPs was found by Austrin, Håstad, and Guruswami \cite{AGH17}, and it was further developed by Brakensiek and Guruswami \cite{BG16,BG16a,BG18,BG18b}. They use a~notion of polymorphism suitable for PCSPs to prove several hardness and tractability results.  Roughly, the polymorphisms of a~PCSP (template) are multivariate functions from the domain of its `strict' relations to that of its `relaxed' relations that map each strict relation into the corresponding relaxed relation.  For example, the $n$-ary polymorphisms of the PCSP template corresponding to $k$ vs.\ $c$ graph colouring (we say polymorphisms from $\rel K_k$ to $\rel K_c$) are the homomorphisms from the $n$-th Cartesian power of $\rel K_k$ to $\rel K_c$, i.e., the $c$-colourings of $\rel K_k^n$.
It is shown in~\cite{BG18b} that the complexity of a~PCSP is fully determined by its polymorphisms --- in the sense that two PCSPs with the same set of polymorphisms have the same complexity. 

Much of the previous work on the complexity of PCSPs was focused on specific problems, especially on approximate graph and hypergraph colouring and their variants. We describe this in more detail in Examples~\ref{ex:2+eps}--\ref{ex:rainbow} in the next section.
Let us note here that, despite much effort, there is a huge gap between known algorithmic and \NP-hardness results for colouring 3-colourable graphs with $c$ colours: the best known \NP-hardness result (without additional assumptions) prior to this paper went only as far as $c=4$ \cite{KLS00,GK04},
while the best (in terms of $c$) known efficient algorithm uses roughly $O(n^{0.199})$ colours to colour an $n$-vertex 3-colourable graph~\cite{KT17}.
There are also hardness results concerning hypergraph colouring with a~super-constant number of colours, e.g.~\cite{ABP19,GHHSV17}, but problems like this do not fall directly into the framework that we consider in this paper.

We remark that appropriate versions of polymorphisms have been used extensively in many CSP complexity/approximability classifications: standard polymorphisms for decision and counting CSPs, for approximating Min CSPs and testing solutions (in the sense of property testing) \cite{BKW17,Bul13,Bul17,BK16,CVY16,DKKMMO17,DKM15,Zhu17}, fractional polymorphisms for exact optimisation problems \cite{KKR17,TZ16}, $\alpha$-approximate polymorphisms for approximating Max CSPs \cite{BR15}. In all cases, the presence of nice enough polymorphisms (of appropriate kind) leads directly to efficient algorithms, while their absence leads to hardness.
Interestingly, it was shown in \cite{BR15} that the Unique Games Conjecture is equivalent to the \NP-hardness of approximating Max CSPs beyond a~specific numerical parameter of their (nice enough) approximate polymorphisms.

\subsubsection*{Our contribution}
The main contribution of the present paper is a~new abstract algebraic theory for the PCSP. A~crucial property of polymorphisms for PCSPs is that, unlike in CSPs, they cannot be composed (as functions). The ability to compose polymorphisms to produce new polymorphisms was used extensively in the algebraic theory of CSPs. 
This could be viewed as a~serious limitation on the applicability of the algebraic approach to PCSPs. Alternatively, it might indicate that the ability to compose is not that essential, and that a~composition-free abstract algebraic theory for PCSPs (and hence for CSPs) can be developed. Our results suggest that the latter is in fact the case.

We show that certain abstract properties of polymorphisms, namely systems of minor identities (i.e., function equations of a~simple form) satisfied by polymorphisms, fully determine the complexity of a~PCSP. This shifts the focus from concrete properties of polymorphisms to their abstract properties. Systems of minor identities satisfied by polymorphisms provide a~useful measure of how much symmetry a~problem has. This measure gives a~new tool to compare/relate the complexity of PCSPs, far beyond what was available before. We envisage that our paper will bring a~step change in the study of PCSPs, similar to what \cite{BJK05, BJ01} did for the CSP. Let us explain this in some more detail.

To be slightly more technical, a~\emph{minor identity} is a~formal expression of the form
\[
  f(x_1,\dots,x_n) \equals
  g(x_{\pi(1)},\dots,x_{\pi(m)})
\]
where $f,g$ are function symbols (of arity $n$ and $m$, respectively), $x_1$, \dots, $x_n$ are variables, and $\pi\colon [m] \to [n]$ (we use notation $[n]=\{1,\ldots,n\}$ throughout). 
A~minor identity can be seen as an equation where the function symbols are the unknowns, and if some specific functions $f$ and $g$ satisfy such an identity then $f$ is called a~\emph{minor} of $g$. We use the symbol $\equals$ instead of $=$ to stress the difference between a~formal identity (i.e. equation involving function symbols) and equality of two specific functions. A~\emph{minor condition} is a~finite system of minor identities (where the same function symbol can appear in several identities). A~bipartite minor condition is one where sets of function symbols appearing on the left- and right-hand sides of the identities are disjoint. Such condition is said to be satisfied in a~set $\clo F$ of functions if it is possible to assign a~function from $\clo F$ of the corresponding arity to each of the function symbols in such a~way that all the identities are simultaneously satisfied (as equalities of functions, i.e., for all possible values of the $x_i$'s).

Informally, the main results of our new general theory state that
\begin{enumerate}
  \item[(A)] If every bipartite minor condition satisfied in the polymorphisms of (the template of) one PCSP $\Pi_1$ is also satisfied in the polymorphisms of another PCSP $\Pi_2$, then $\Pi_2$ is log-space reducible to $\Pi_1$ (see Theorem~\ref{thm:main} for a~formal statement).
  Moreover, we characterise the premise of the above claim in many equivalent ways (see Theorem~\ref{thm:minor-homomorphism-is-pp-constructibility} and Corollary~\ref{cor:pp-constructible-minor-homomorphism}).

  \item [(B)] Every PCSP $\Pi$ is log-space equivalent to the problem deciding whether a~given bipartite minor condition is satisfiable by projections/dictators or not satisfiable even by polymorphisms of $\Pi$ (see Theorem~\ref{thm:pcsp-and-lc}).
\end{enumerate}

The first of the above results establishes the key role of bipartite minor conditions satisfied in poly\-mor\-phisms --- in particular, the hardness or tractability of a~PCSP can always be explained on this abstract level, since any two PCSPs have the same complexity if their polymorphisms satisfy the same bipartite minor conditions.  Moreover, this abstract level allows one to compare any two PCSPs, even when it does not makes sense to compare their sets of polymorphisms inclusion-wise (say, because the functions involved are defined on different sets).
The second result establishes that every PCSP is equivalent to what can be viewed as an algebraic version of the Gap Label Cover problem, which is the most common starting point of PCP-based hardness proofs in the inapproximability context. Our result uses the fact the Label Cover can be naturally interpreted as the problem, which we call MC, of checking triviality of a~system of minor identities. The gap version of MC has an algebraic component in place of the quantitative gap of Gap Label Cover. In particular, result (B) can provide a~general approach to proving \NP-hardness of PCSPs --- via analysis of bipartite minor conditions satisfied by polymorphisms.

Our general algebraic reductions use constructions that resemble dictatorship tests, which are also present in many inapproximability proofs.  All known hardness results for PCSPs (possibly except \cite{Kho01,Hua13}) can be proved using our view. We show how our theory can be used to obtain general algebraic sufficient conditions for \NP-hardness of a~PCSP --- we give an example of this (Theorem~\ref{thm:no-epsilon-robust-is-hard}) that covers several cases considered in the literature before. 
Our theory can also be used to translate specific hardness results into general hardness results (which of course can then be applied in new specific cases). We demonstrate how this works and answer specific questions from \cite{BG16,BG16a} by showing, among other things, the following.

\begin{enumerate}
  \item[(C)] For any $k\ge 3$, it is \NP-hard to distinguish $k$-colourable graphs from those that are not $(2k-1)$-colourable.
  (See Theorem~\ref{thm:hardness-k-vs-2k-1}).  
\end{enumerate}

In particular, it follows that it is \NP-hard to 5-colour a 3-col\-our\-able graph.  This might seem a small step towards closing the big gap in our understanding of approximate graph colouring, but we believe that it is important methodologically and that further development of our general theory and further analysis of polymorphisms for graph colouring will eventually lead to a proof of \NP-hardness for any constant number of colours.

We prove that the approximate graph colouring problem in result (C) has less symmetry (in the sense of bipartite minor conditions) than approximate hypergraph colouring, which is known to be \NP-hard. Then our theory implies the required reduction. Our theory also allows one to rule out the existence of certain reductions --- for example, we can explain exactly how $k$ vs.\ $c=2k-1$ colouring differs from the cases $c=2k-2$ and $c=2k$, and hence why we are able to improve the result from $c=2k-2$ \cite{BG16} to $c=2k-1$ and why moving to larger $c$ requires further analysis of polymorphisms and bipartite minor conditions. 

On the tractability side, practically all conditions describing tractable cases of CSP (as well as conditions describing the power of specific important algorithms) have the form of minor identities. It follows from our results that efficient algorithms should be based on such conditions too.
We provide such characterisation of the power of three algorithms for PCSPs. The first two of these algorithms bring a~tractability result from a~CSP to PCSPs. The second and the third algorithm are related to recent tractability results from \cite{BG18b}. While the the algorithmic novelty of these results is limited, they demonstrate how the applicability of certain algorithms is characterised by minor conditions satisfied by polymorphisms.

Interestingly, all known tractability results for finite-domain PCSPs (such as those above or those from \cite{BG18,BG18b}) follow the same scheme: the PCSP is shown to be a subproblem (inclusion-wise) of a tractable CSP, possibly with an infinite domain.
We show that using infinite domain CSPs in this scheme can be necessary --- namely, a simple specific PCSP, that is known to be a~subproblem of several tractable infinite-domain CSPs \cite{BG18,BG18b}, is \emph{not} a~subproblem of any tractable finite-domain CSP (see Theorem \ref{thm:infinity-necessary}).

\subsubsection*{Subsequent work}

Two new recent results are based on the general theory presented in this paper. A dichotomy for symmetric Boolean PCSPs (generalising a classification from \cite{BG18}) was proved in \cite{FKOS19}.  It was shown in \cite{KO19} that, for any fixed 3-colourable non-bipartite graph $\rel H$, it is NP-hard to find a 3-colouring for given graph that admits a homomorphism to $\rel H$. The latter result appears in this paper for the case when $\rel H$ is a 5-cycle, but the methodology in \cite{KO19} uses the general theory from this paper, rather than being just an extension of the proof for the special case.

\subsubsection*{Further discussion}

Let us now discuss how the complexity classification quest for PCSPs compares with that for CSPs. As we said above, the gist of the algebraic approach is that lack or presence of high-dimensional symmetries determines the complexity. For (finite-domain) CSPs, there is a~sharp algebraic dichotomy: having only trivial symmetries (i.e., satisfying only those systems of minor identities that are satisfied in polymorphisms of every CSP) leads to \NP-hardness, while  any non-trivial symmetry  implies rather strong symmetry and thus leads to tractability. Moreover, the algorithms for tractable cases are (rather involved) combinations of only two basic algorithms --- one is based on local propagation \cite{BK14} and the other can be seen as a~very general form of Gaussian elimination \cite{IMMVW10}.
It is already clear that the situation is more complicated for PCSPs: there are hard PCSPs with non-trivial (but limited in some sense) symmetries, and tractable cases are more varied \cite{AGH17,BG18,BG18b,DRS05}.  This calls for more advanced methods, and we hope that our paper will provide the basis for such methods.
There is an obvious question whether PCSPs exhibit a~dichotomy as CSPs do, but there is not enough evidence yet to conjecture an answer. More specifically, it is not clear whether there is any PCSP whose polymorphisms are not limited enough (in terms of satisfying systems of minor identities) to give \NP-hardness, but also not strong enough to ensure tractability. Classifications for special cases such as Boolean PCSPs and graph homomorphisms would help to obtain more intuition about the general complexity landscape of PCSPs, but these special cases are currently open. A more detailed discussion of possible directions of further research can be found in Section~\ref{sec:conclusion}.

\subsubsection*{Organisation of this paper}

The following section contains formal definitions of the concepts introduced in the introduction as well as several concrete examples of PCSPs.

Sections~\ref{sec:algebraic-reductions} and~\ref{sec:constructions} are the core of our general theory.
Sections~\ref{sec:algebraic-reductions} introduces reductions via algebraic conditions and proves the main results of our general theory.
In Section~\ref{sec:constructions}, we describe several relational constructions behind the reduction from Section~\ref{sec:algebraic-reductions}, and provide many equivalent ways to characterise the applicability of our main reduction.

Sections~\ref{sec:hardness-from-lc} and \ref{sec:hardness-from-pcsp} are focused on the hardness results. In particular, Section~\ref{sec:hardness-from-lc} contains algebraic sufficient conditions  leading to hardness by reduction from several versions of Label Cover and  also a~sketch of a~simplified proof of \NP-hardness of approximate hypergraph colouring. Section~\ref{sec:hardness-from-pcsp} contains our result on \NP-hardness of approximate (3 vs. 5) graph colouring.

Section~\ref{sec:tractable} characterises the power of three polynomial-time algorithms for PCSPs by means of minor conditions. 

In Section~\ref{sec:fin_intract}, we use a~specific tractable PCSP to show that the only currently known approach for proving tractability of a~PCSP --- by a~natural reduction to a~tractable CSP --- must in general involve CSPs with an infinite domain.

Section~\ref{sec:wonderland} describes algebraic counterparts of the constructions introduced in Section~\ref{sec:constructions} and gives additional information that can be useful for further developing our general algebraic theory.

Section~\ref{sec:graph_coloring} further discusses the approximate graph colouring problem and provides technical results about distinguishing various special cases of this problem by means of minor conditions.

Finally, Section~\ref{sec:conclusion} summarises the results of the paper and discusses possible directions of further research. 

\begin{figure}
  \begin{tikzpicture}[every edge/.append style = {->}]
    \foreach \i in {2,...,10}
      \node (\i) at (\i,0) {\i};

    \draw (2) edge  (3);
    \draw (3) edge  (4);
    \draw (4) edge [dotted] (5);
    \draw (5) edge [dotted] (6);
    \draw (7) edge [dotted] (8);
    \draw (3) edge [bend right] (5);
    \draw (4) edge [bend left] (6);
    \draw (4) edge [bend left] (7);
    \draw (4) edge [bend left] (8);
    \draw (6) edge [bend right] (10);
    \draw (4) edge [bend left] (9);
  \end{tikzpicture}

  \caption{Graph of dependencies of sections. Dotted edges express non-essential references between sections.}
\end{figure}
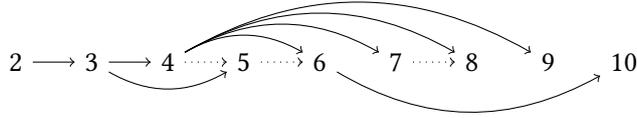   %
\section{Preliminaries}
  \label{sec:prelim}

This section contains formal definitions of the notions introduced above. For comparison, the algebraic theory behind fixed-template CSPs can be found in a~recent survey \cite{BKW17}.

\subsection{CSPs and PCSPs}

We use the notation $[n]=\{1,\ldots,n\}$ and $E_n=\{0,1,\ldots,n-1\}$ throughout the paper.

\begin{definition}
A \emph{constraint language} $\Gamma$ on a set $A$ is finite set of relations on $A$, possibly of different arity. Then $A$ is called the \emph{domain} of $\Gamma$.
\end{definition}

To work with several constraint languages (possibly on different domains), it is often convenient to fix an indexing of the relations in $\Gamma$.
This is formalised as follows.

\begin{definition}
A \emph{(relational) structure} is a tuple $\rel A =(A; R_1^{\rel A},\ldots,R_l^{\rel A})$ where each $R_i^{\rel A}\subseteq A^{\ar(R_i)}$ is a relation on $A$ of arity $\ar(R_i)\ge 1$. We say that $\rel A$ is finite if $A$ is finite. We will assume that all structures in this paper are finite unless specified otherwise.

Two structures $\rel A =(A; R_1^{\rel A},\ldots,R_l^{\rel A})$ and $\rel B =(B; R_1^{\rel B},\ldots,R_l^{\rel B})$ are called \emph{similar} if they have the same number of relations and $\ar(R_i^{\rel A})=\ar(R_i^{\rel B})$ for each~$i\in [l]$.
\end{definition}

For example, a~(directed) graph is relational structure with one binary relation. Any two graphs are similar structures.

We often use a~single letter instead of $R_i$ to denote a~relation of a~structure, e.g.\ $S^{\rel A}$ would denote a~relation of $\rel A$, the corresponding relation in a~similar structure $\rel B$, would be denoted by $S^{\rel B}$. Also, throughout the paper we denote the domains of structures $\rel A$, $\rel B$, $\rel K_n$ and so on by
$A$, $B$, $K_n$ etc., respectively.

\begin{definition}
For two similar relational structures $\rel A$ and $\rel B$, a~\emph{homomorphism} from $\rel A$ to $\rel B$ is a~map $h\colon A \rightarrow B$ such that, for each~$i$, 
\[
  \text{if } (a_1,\ldots,a_{\ar(R_i)})\in R_i^{\rel A}
  \text{ then }
  (h(a_1),\ldots,h(a_{\ar(R_i)}))\in R_i^{\rel B}.
\]
We write $h\colon \rel A\to \rel B$ to denote this, and simply $\rel A \rightarrow \rel B$ to denote that a homomorphism from $\rel A$ to $\rel B$ exists.  In the latter case, we also say that $\rel A$ maps homomorphically to $\rel B$.
\end{definition}

\begin{definition}
For a fixed structure $\rel B$, $\CSP(\rel B)$ is the problem of deciding whether a given input structure $\rel I$, similar to $\rel B$, admits a homomorphism to $\rel B$. In this case $\rel B$ is called the \emph{template} for $\CSP(\rel B)$.
\end{definition}

For example, when $\rel I$ and $\rel B$ are graphs and $\rel B=\rel K_k$ is a $k$-clique, a homomorphism from $\rel I$ to $\rel B$ is simply a $k$-colouring of $\rel I$. Then $\CSP(\rel B)$ is the standard $k$-colouring problem for graphs. 

To see how the above definition corresponds to the definition of CSP with variables and constraints, view the domain of the structure $\rel I$ as consisting of variables, relations in $\rel I$ specifying which tuples of variables the constraints should be applied to, and (the corresponding) relations in $\rel B$ as sets of allowed tuples of values.

\begin{definition}
A \emph{PCSP template} is a pair of similar structures $\rel A$ and $\rel B$ such that $\rel A \rightarrow \rel B$. The problem $\PCSP(\rel A,\rel B)$ is, given an input structure $\rel I$ similar to $\rel A$ and $\rel B$, output \yes{} if $\rel I \rightarrow \rel A$, and \no{} if $\rel I \not\rightarrow \rel B$.
\end{definition}

Note that $\PCSP(\rel A,\rel A)$ is simply $\CSP(\rel A)$.  The promise in the PCSP is that it is never the case that $\rel I \not\rightarrow \rel A$ and $\rel I \rightarrow \rel B$. Note also that the assumption $\rel A \rightarrow \rel B$ is necessary for the problem to make sense --- otherwise, the \yes- and \no-cases would not be disjoint. We define $\PCSP(\rel A,\rel B)$ as a~decision problem, but it can also be defined as a~search problem:

\begin{definition} Given two relational structures $\rel A$, $\rel B$ as above, the search version of $\PCSP(\rel A,\rel B)$ is, given an~input structure $\rel I$ that maps homomorphically to $\rel A$, \emph{find} a~homomorphism $h\colon \rel I\to \rel B$.
\end{definition}

There is an obvious reduction from the decision variant of every PCSP to its search variant.  Nevertheless, it is open whether these two problems are equivalent for all PCSP templates. Note that, for problems $\CSP(\rel A)$, these two versions are always equivalent \cite{BJK05}.
We would also like to note that all of our results are formulated and proved for the decision version. However, all of them can be generalised to the corresponding search versions of the problems.

Let us give some examples of the problems of the form $\PCSP(\rel A,\rel B)$ which are proper promise problems, i.e., not of the form $\CSP(\rel A)$. More examples (of both tractable and intractable PCSPs) can be found in \cite{BG16,BG18,BG18b}.
\begin{example}[$(2+\epsilon)$-\Sat] 
\label{ex:2+eps}
For a tuple $\mathbf t\in \{0,1\}^n$, let $\operatorname{Ham}(\mathbf t)$ be the Hamming weight of $\mathbf t$. Fix an integer $k\ge 1$. Let $\ne_2$ be the relation $\{(0,1),(1,0)\}$. Let 
\begin{align*}
\rel A &= (\{0,1\}; \{\mathbf t \in \{0,1\}^{2k+1}\mid \operatorname{Ham}(\mathbf t)\ge k\}, {\ne_2}),\\  
\rel B &= (\{0,1\}; \{\mathbf t \in \{0,1\}^{2k+1}\mid \operatorname{Ham}(\mathbf t)\ge 1\}, {\ne_2}).
\end{align*}
The problem $\PCSP(\rel A,\rel B)$ is then (equivalent to) the following variant of $(2k+1)$-\Sat: given an instance of $(2k+1)$-\Sat\ such that some assignment satisfies at least $k$ literals in each clause, find a~normal satisfying assignment.
This problem, called $(2+\epsilon)$-\Sat, was proved to be \NP-hard in \cite{AGH17}.
\end{example}

\begin{example}[1-in-3- vs. Not-All-Equal-\Sat] \label{example:1in3-nae}
Let 
\begin{align*}
\rel T &= (\{0,1\}; \{(1,0,0),(0,1,0),(0,0,1)\}),\\
\NAE_2 &= (\{0,1\}; \{0,1\}^3\setminus \{(0,0,0),(1,1,1)\}).
\end{align*}
Even though $\CSP(\rel T)$ and $\CSP(\NAE_2)$ are (well-known) $\NP$-hard problems, the problem $\PCSP(\rel T,\NAE_2)$ was shown to be in $\Ptime$ in \cite{BG18,BG18b}, along with a range of similar problems. All these results are obtained by using the same general scheme - a natural reduction to a tractable CSP, possibly with an infinite domain. We give several general results related to this scheme in Section~\ref{sec:tractable} and then show in
Section~\ref{sec:fin_intract} that any proof of tractability of $\PCSP(\rel T,\NAE_2)$ via this scheme must use a CSP with an infinite domain.
\end{example}

\begin{example}[Approximate Graph Colouring]
  \label{ex:approx-graph-col}
For $k\ge 2$, let ${\ne_k}$ denote the relation $\{(a,b)\in E_k^2\mid a\neq b\}$.  For $k \le c$, let
\begin{align*}
  \rel K_k &= (E_k; {\ne_k}),\\
  \rel K_c &= (E_c; {\ne_c}).
\end{align*}
Then $\PCSP(\rel K_k,\rel K_c)$ is the well-known {\em approximate graph colouring} problem: given a $k$-colourable graph, find a $c$-colouring. The decision version of this asks to distinguish between  $k$-colourable graphs and those that are not even $c$-colourable.
The complexity of this problem has been studied since 1976 \cite{GJ76} --- the problem has been conjectured to be \NP-hard for any fixed $3\le k\le c$, but this is still open in many cases (see \cite{BG16,GK04,Hua13,KLS00,Kho01}) if we assume only $\Ptime\ne\NP$. For $k=3$, the case $c=4$ was shown to be $\NP$-hard
\cite{KLS00,GK04}, but even the case $c=5$ was still open, and we settle it in this paper (see Section~\ref{sec:hardness-from-pcsp} for more details).
It was shown (using polymorphisms) in~\cite{BG16} that, for any $k\ge 3$, it is \NP-hard to distinguish $k$-colourable graphs from those that are not $(2k-2)$-colourable.
This gives the best known \NP-hardness results for small enough $k$, but we further improve this result in this paper. For large enough $k$, the best known $\NP$-hardness result is for $k$ vs.\ $2^{\Omega(k^{1/3})}$ colouring \cite{Hua13}.
By additionally assuming somewhat non-standard variants of the Unique Games Conjecture (with perfect completeness), \NP-hardness of all approximate graph colouring problems (with $k\ge 3$) was proved in \cite{DMR09}. 
\end{example}

\begin{example}[Approximate Graph Homomorphism]
\label{ex:approx-graph-hom}
The standard graph $k$-co\-lour\-ing problem has a~natural generalisation to graph homomorphisms. For a~fixed (undirected) graph $\rel H$, the problem $\CSP(\rel H)$ is, given a graph $\rel I$ (which can be assumed to be undirected), decide whether $\rel I\rightarrow \rel H$. This problem is often called $\rel H$-colouring. It has been extensively studied, see, e.g.~\cite{HN90,HN04,Lar17}. A~well-known result by Hell and Nešetřil~\cite{HN90} states that $\rel H$-colouring is solvable in polynomial time if $\rel H$ is bipartite or has a loop, and it is $\NP$-complete otherwise. A natural generalisation to the promise setting, i.e., problems $\PCSP(\rel G,\rel H)$, was suggested in \cite{BG18b}. It was conjectured there that the problem is \NP-hard for any non-bipartite loopless undirected $\rel G,\rel H$ with $\rel G \rightarrow\rel H$. Even the case $\rel G=\rel C_5$ (5-cycle), $\rel H=\rel K_3$ was mentioned in~\cite{BG18b} as an intriguing open problem, and we settle it in this paper.
See Section~\ref{sec:hardness-from-pcsp} for more details.
\end{example}

\begin{example}[Approximate Hypergraph Colouring] \label{ex:hypergraph-colouring}
This problem is similar to the previous one, but uses the ``not-all-equal'' relation 
$\nae_k=E_k^3\setminus \{(a,a,a)\mid a\in E_k\}$ instead of $\neq_k$, and similarly for $c$, i.e., we are talking about $\PCSP(\NAE_k,\NAE_c)$ where
\begin{align*}
  \NAE_k &= (E_k; {\nae_k}),\\
  \NAE_c & =(E_c; {\nae_c}).
\end{align*}
A colouring of a~hypergraph is an assignment of colours to its vertices that leaves no hyperedge monochromatic.
Thus, in (the search variant of) this problem one needs to find a $c$-colouring for a given $k$-colourable 3-uniform hypergraph.
This problem has been proved to be \NP-hard for any fixed $2\le k\le c$ \cite{DRS05}.
\end{example}

\begin{example}[Strong vs. Normal Hypergraph Colouring]
Let $k,c\ge 2$, and let
\begin{align*}
\rel A &= (E_{k+1}; \{(a_1,\ldots,a_k)\in E_{k+1}^k\mid a_i\ne a_j \mbox{ for all } i,j\}),\\
\rel B &= (E_c; \{(a_1,\ldots,a_k)\in E_c^k\mid a_i\ne a_j \mbox{ for some } i,j\}).
\end{align*}
Then $\PCSP(\rel A,\rel B)$ is the problem of distinguishing $k$-uniform hypergraphs
that admit a strong $(k+1)$-colouring (i.e., one without repetition of colours in any hyperedge) from those that do not admit a normal $c$-colouring. It was conjectured in \cite{BG16} that this problem is $\NP$-hard for all $(k,c)\ne (2,2)$, and some special cases ($k=3,4$ and $c=2$) were settled in that paper, but this conjecture remains wide open.
\end{example}

\begin{example}[Rainbow vs. Normal Hypergraph Colouring]
\label{ex:rainbow}
Let $k,q,c$ be positive integers with $k\ge q\ge 2$ and $c\ge 2$, and let
\begin{align*}
\rel R_{k,q} &= (E_q; \{(a_1,\ldots,a_k)\in E_q^k\mid \{a_1,\ldots,a_k\}=E_q\}),\\
\rel H_{k,c} &= (E_c; \{(a_1,\ldots,a_k)\in E_c^k\mid a_i\ne a_j \mbox{ for some } i,j\}).
\end{align*}
In $\PCSP(\rel R_{k,q},\rel H_{k,c})$, one is given a $k$-uniform hypergraph which has a~$q$-colouring such that all colours appear in each hyperedge, and one needs to find a normal $c$-colouring. This problem is known to be in \Ptime{} for $k=q$ and $c=2$; a~randomised algorithm for it can be found in~\cite{McD93}, and a deterministic algorithm due to Alon (unpublished) is mentioned in~\cite{BG18}.  $\PCSP(\rel R_{k,q},\rel H_{k,c})$ is \NP-hard if $2\le q\le\lfloor k/2\rfloor$ \cite{GL17} or if $2\le q\le k-2\lfloor\sqrt{k}\rfloor$ and $c=2$ \cite{ABP18}.  Further variations of such PCSPs were considered in \cite{ABP18,GL17}.
\end{example}

\subsection{Polymorphisms}

We now proceed to define polymorphisms, which are the main algebraic technical tools 
used in the analysis of CSPs.

\begin{definition}\label{def:nth-power}
The \emph{$n$-th power} of $\rel A$ is the structure $\rel A^n=(A^n; R_1^{\rel A^n},\ldots,R_l^{\rel A^n})$ whose relations are defined as follows: for every $\ar(R_i)\times n$ matrix $M$ such that all columns of $M$ are in $R_i^{\rel A}$, consider the rows of $M$ as elements of $A^n$ and put this tuple in $R_i^{\rel A^n}$.
\end{definition}

\begin{definition}
Given two similar relational structures $\rel A$ and $\rel B$, an $n$-ary \emph{polymorphism}\footnote{Called \emph{weak polymorphism} in \cite{AGH17,BG16a}.} \emph{from} $\rel A$ \emph{to} $\rel B$ is a~homomorphism from $\rel A^n$ to $\rel B$.
To spell this out, a polymorphism is a~mapping $f$ from $A^n$ to $B$ such that, for each $i\le l$ and all tuples $(a_{11},\dots,a_{\ar(R_i)1})$, \dots, $(a_{1n},\dots,a_{\ar(R_i)n}) \in R_i^{\rel A}$, we have
\[
  ( f( a_{11},\dots,a_{1n} ), \dots,
  f( a_{\ar(R_i)1},\dots,a_{\ar(R_i)n} ) ) \in R_i^{\rel B}
.\]
We denote the set of all polymorphisms from $\rel A$ to $\rel B$ by $\Pol(\rel A,\rel B)$, and we write simply $\Pol(\rel A)$ for $\Pol(\rel A,\rel A)$.
\end{definition}

For the case $\rel A=\rel B$, this definition coincides with the standard definition of a~polymorphism of a relational structure $\rel A$ (see, e.g., \cite[Section 4]{BKW17}).

\begin{definition}
A {\em projection} (also known as {\em dictator}\footnote{Projection is the standard name for these objects in universal algebra, while dictator is the standard name in approximability literature.}) on a set $A$ is an operation $p_i^{(n)}\colon A^n\rightarrow A$ of the form $p_i^{(n)}(x_1,\ldots,x_n)=x_i$.
\end{definition}

It is well-known and easy to see that every projection on $A$ is a~polymorphism of every relational structure on $A$. We will sometimes write simply $p_i$ when the arity $n$ of a projection is clear from the context.

\begin{example} \label{ex:poly-of-1in3-nae}
Consider the structures $\rel T$ and $\rel H_2$ from Example~\ref{example:1in3-nae}.
It is well-known and not hard to verify that 
\begin{itemize}
  \item $\Pol(\rel T)$ consists of all projections on $\{0,1\}$, and 
  \item $\Pol(\rel H_2)$ consists of all operations of the form $\pi(p_i^{(n)})$ where $\pi$ is a permutation on $\{0,1\}$ and $p_i^{(n)}$ is a projection.
\end{itemize}
However, $\Pol(\rel T,\rel H_2)$ contains many functions that are not like projections. For example, for any $k\ge 1$, define $f_k\colon\{0,1\}^{3k-1}\rightarrow \{0,1\}$ so that  $f_k(\mathbf t)=1$ if $\Ham(\mathbf t)\ge k$ and $f_k(\mathbf t)=0$ otherwise. It is easy to see that $f_k\in \Pol(\rel T,\rel H_2)$. Indeed, if $M$ is a~$3\times (3k-1)$ matrix whose columns come from the set $\{(1,0,0),(0,1,0),(0,0,1)\}$ then some row of $M$ contains strictly fewer than $k$ 1's and some other row contains at least $k$ 1's. Therefore, by applying $f_k$ to the rows of $M$, one obtains a tuple in $\{0,1\}^3\setminus \{(0,0,0),(1,1,1)\}$, as required.
\end{example}

\begin{example}
Recall Example~\ref{ex:approx-graph-col}. It is easy to check that, for any $n\ge 1$, the $n$-ary functions from $\Pol(\rel K_k,\rel K_c)$ are simply the $c$-colourings of $\rel K_k^n$, the $n$-th power of $\rel K_k$ (in the sense of Definition~\ref{def:nth-power}).
\end{example}

Unlike polymorphisms of a~single relational structure $\rel A$,  the set $\Pol(\rel A,\rel B)$ is not closed under composition --- for example, if $f(x,y,z)$ and $g(x,y)$ are in $\Pol(\rel A,\rel B)$, then $f(g(x,w),y,z)$ is not necessarily there. In general, the composition is not always well-defined, and even when it is (e.g.\ when $\rel A$ and $\rel B$ have the same domain), $f(g(x,w),y,z)$ may not be in $\Pol(\rel A,\rel B)$. 
However, it is always closed under taking minors.

\begin{definition} An~$n$-ary function $f\colon A^n\to B$ is called a~\emph{minor} of an $m$-ary function $g\colon A^m \to B$ given by a~map $\pi \colon [m] \to [n]$ if
\[
  f(x_1,\dots,x_n) = g(x_{\pi(1)},\dots,x_{\pi(m)})
\]
for all $x_1,\dots,x_n \in A$.
\end{definition}
Alternatively, one can say that $f$ is a~minor of $g$ if it is obtained from $g$ by identifying variables, permuting variables, and introducing dummy variables.

\begin{definition}
  Let $\clo O(A,B) = \{f\colon A^n\rightarrow B\mid n\ge 1\}$. A~\emph{(function) minion} 
  \footnote{What we define as a~function minion was called \emph{minor closed class} in \cite{Pip02}, and \emph{clonoid} in \cite{AM16}. Albeit clonoids are usually also required to be closed under composition with a~given clone from the outside.}
  $\clo M$ on a pair of sets $(A,B)$ is a~non-empty subset of $\clo O(A,B)$ that is closed under taking minors. For fixed $n\ge 1$, let $\clo M^{(n)}$ denote the set of $n$-ary functions from $\clo M$. Unless stated otherwise, we assume that all minions are defined on finite sets.
\end{definition}

We remark that clones have been used extensively in the algebraic theory of CSP --- a clone is simply a minion on $(A,A)$ that is closed under composition and contains all projections. For any structure $\rel A$, $\Pol(\rel A)$ is a clone. For more detail, see~\cite{BKW17}.

We now introduce one of the central notions of this paper; it generalises the notion of \emph{h1 clone homomorphisms} from \cite{BOP18}.

\begin{definition}\label{def:minor-hom} 
  Let $\clo M$ and $\clo N$ be two minions (not necessarily on the same pairs of sets). A~mapping $\xi\colon \clo M \to \clo N$ is called a~\emph{minion homomorphism} if 
  \begin{enumerate}
    \item it preserves arities, i.e., $\ar(g)=\ar(\xi(g))$ for all $g\in \clo M$, and 
    \item it preserves taking minors, i.e., for each $\pi\colon [m] \to [n]$ and each $g\in \clo M^{(m)}$ we have
    \[
      \xi (g)(x_{\pi(1)},\dots,x_{\pi(m)}) = \xi (g(x_{\pi(1)},\dots,x_{\pi(m)}))
    .\]
  \end{enumerate}
\end{definition}

Item (2) above can also be interpreted as `preserving satisfaction of minor identities', i.e.,  if 
\(
  f(x_1,\dots,x_n) =
  g(x_{\pi(1)},\dots,x_{\pi(m)})
\) for some $f \in \clo M^{(n)}$, $g\in \clo M^{(m)}$, and $\pi\colon [m] \to [n]$,
then
\[
  \xi(f)(x_1,\dots,x_n) =
  \xi(g)(x_{\pi(1)},\dots,x_{\pi(m)})
.\]

\begin{example} \label{ex:pol-3,4-maps-to-projections}
  We will construct a~minion homomorphism from $\Pol(\rel K_3,\rel K_4)$ to the minion $\clo P_2$ of all projections on a two-element set. 

  Our minion homomorphism is built on the following combinatorial statement proved in \cite[Lemma 3.4]{BG16}: for each $f\in \Pol^{(n)}(\rel K_3,\rel K_4)$, there exist $t\in K_4$ 
  (we will call any such $t$ a~\emph{trash colour}), a coordinate $i\in [n]$, and a map $\alpha\colon K_3 \to K_4$ such that
  \[
    f(a_1,\dots,a_n) \in \{ t, \alpha( a_i ) \}
  \]
  for all $a_1,\dots,a_n\in K_3$. In other words, if we erase the value $t$ from the table of $f$ then $f(x_1,\dots,x_n)$, which is now a partial function, depends only on $x_i$. Moreover, it is shown in \cite[Lemma 3.9]{BG16}, that while the trash colour $t$ is not necessarily unique, the coordinate $i$ is.%

  We define $\xi\colon \Pol(\rel K_3,\rel K_4) \to \Proj_2$ by mapping each $f$ to $p_i$ (preserving the arity) for the $i$ that satisfies the above. To prove that such $\xi$ is a~minion homomorphism, consider $f,g$ such that $f$ is a minor of $g$, i.e.,
  \[
    f(x_1,\dots,x_n) = g(x_{\pi(1)},\dots,x_{\pi(m)})
  \]
  for some $\pi\colon [m] \to [n]$.
  We claim that if $t$ is a~trash colour for $g$, then it is also a~trash colour for $f$. Indeed, if
 $g(a_1,\dots,a_m) \in \{ t,\alpha(a_i) \}$ for all $a_1,\dots,a_n\in K_3$, we get
  \[
    f(a_1,\dots,a_n) = g(a_{\pi(1)},\dots,a_{\pi(m)}) \in \{t,\alpha(a_{\pi(i)})\}.
  \]
  This also shows that the coordinate $j$ assigned to $f$ is $\pi(i)$, therefore $\xi(f) = p_{\pi(i)}$. We conclude that $\xi$ preserves taking minors. That is true since
  \[
    p_{\pi(i)}(x_1,\dots,x_n) = x_{\pi(i)} = p_i(x_{\pi(1)},\dots,x_{\pi(m)}).
  \]
\end{example}

\subsection{Primitive positive formulas}
\label{sec:pp-formulas}

Jeavons' proof \cite{Jea98} that polymorphisms capture the complexity of CSPs was built on a~Galois correspondence between relations and operations \cite{Gei68, BKKR69, BKKR69a}. A Galois correspondence appropriate for PCSP  was found in \cite{Pip02}, and Brakensiek and Guruswami \cite[Appendices D and E]{BG16a} used it to generalise Jeavons' result to PCSP. We build further on these results.

\emph{Primitive positive formulas (pp-formulas)} have been used extensively in the theory of CSPs --- see~\cite{BKW17} for many examples.

\begin{definition}
  For a~relational structure $\rel A = (A; R_1^{\rel A},\dots,R_l^{\rel A})$, a primitive positive formula (pp-formula) over $\rel A$ is an~existentially quantified conjunction of predicates of the form $(v_{j_1},\dots,v_{j_{k_{i}}}) \in R_i$ or $v_{j_1} = v_{j_2}$ where $v_i$ are variables and $k_i=\ar (R_i)$.
\end{definition}

Note that any instance of $\CSP(\rel A)$ can be interpreted as a conjunction of constraints, i.e., as a~pp-formula $\Psi$ without quantifiers and the equality predicate. Similarly, since an instance of $\CSP(\rel A)$ is also an instance of $\CSP(\rel B)$ for any two similar structures $\rel A$ and $\rel B$, and any pp-formula over $\rel A$ is also a~pp-formula over $\rel B$, we can talk about pp-formulas over a~PCSP template $(\rel A,\rel B)$.

To ease readability, we will write $\Psi^{\rel A}$ when we interpret each relational symbol $R_i$ appearing in $\Psi$ as the relation $R_i^{\rel A}$. Therefore, while a~pp-formula $\Psi$ contains formal expressions such as $(x_1,x_2)\in R_i$, $\Psi^{\rel A}$ contains expressions of the form $(x_1,x_2)\in R_i^{\rel A}$, and therefore satisfying assignments of $\Psi^{\rel A}$ form a~subset of $A^k$ where $k$ is the number of free variables of $\Psi$.

\begin{definition} \label{def:strict-relaxation}
  Let $(\rel A,\rel B)$ be a~PCSP template. A~template $(\rel A',\rel B')$ such that $A' = A$ and $B' = B$ is said to be
  \begin{enumerate}
\item \emph{primitive positive definable} (\emph{pp-definable}) in $(\rel A,\rel B)$ if for each relational symbol $R$ of $(\rel A',\rel B')$ there is a~pp-formula $\Psi_R$ over $(\rel A,\rel B)$ such that
\begin{align*}
  R^{\rel A'} &= \{ (a_1,\dots,a_{\ar (R)}) \in A^{\ar(R)} \mid \Psi_R^{\rel A} (a_1,\dots,a_{\ar (R)}) \}
\text{, and} \\
  R^{\rel B'} &= \{ (b_1,\dots,b_{\ar (R)}) \in B^{\ar(R)} \mid \Psi_R^{\rel B} (b_1,\dots,b_{\ar (R)}) \};
\end{align*}
\item a~\emph{strict relaxation} of $(\rel A,\rel B)$ if $\rel A, \rel A', \rel B, \rel B'$ are all similar, and for each relational symbol $R$, we have $R^{\rel A'} \subseteq R^{\rel A}$ and $R^{\rel B} \subseteq R^{\rel B'}$.
\end{enumerate}
We say that a~template is \emph{primitive positive promise definable} (\emph{ppp-definable}) in $(\rel A,\rel B)$, if it is it is obtained from $(\rel A,\rel B)$ by a~sequence of pp-definitions and strict relaxations.
\end{definition}

We remark that ppp-definability is closely related to polymorphisms and the reduction between PCSPs described in \cite{BG16a}. 

\begin{theorem}[\cite{Pip02,BG16a}] \label{thm:galois-corespondence}
  The following are equivalent for a~pair of PCSP templates $(\rel A,\rel B)$ and $(\rel A',\rel B')$ over the same pair of finite domains:
  \begin{enumerate}
    \item $(\rel A',\rel B')$ is ppp-definable in $(\rel A,\rel B)$.
    \item $\Pol(\rel A,\rel B) \subseteq \Pol(\rel A',\rel B')$.
  \end{enumerate}
Moreover, if these conditions hold, then $\PCSP(\rel A',\rel B')$ is polynomial-time reducible to $\PCSP(\rel A,\rel B)$.
\end{theorem}

\section{Algebraic reductions}
 \label{sec:algebraic-reductions}

In this section, we prove the two main results of our general theory. The first one can be formulated right away.

\begin{theorem} \label{thm:h1-reductions} \label{thm:main}
  Let $(\rel A_1, \rel B_1)$ and $(\rel A_2, \rel B_2)$ be two finite PCSP templates, and let $\clo M_i = \Pol(\rel A_i,\rel B_i)$ for $i=1,2$. If there is a~minion homomorphism $\xi\colon \clo M_1 \to \clo M_2$ then $\PCSP(\rel A_2,\rel B_2)$ is log-space reducible to $\PCSP(\rel A_1,\rel B_1)$.
\end{theorem}

The proof of this theorem is provided by the second of our results, but to formulate it, we need to first address a~few formalities about the problem of deciding minor conditions. 
We remark that  we characterise the existence of such a~minion homomorphism in many equivalent ways (see Theorem~\ref{thm:minor-homomorphism-is-pp-constructibility} and Corollary~\ref{cor:pp-constructible-minor-homomorphism}).

\subsection{Deciding satisfiability of bipartite minor conditions}
  \label{sec:MC}

The study of Maltsev conditions\footnote{Not to be confused with Maltsev identities $f(x,y,y) \equals f(y,y,x) \equals x$, which are a~specific case of a~Maltsev condition.} \cite{Tay73,Neu74} has a long history in universal algebra. We will use only a special case of (strong) Maltsev conditions called minor conditions and their restricted `bipartite' form.  

\begin{definition} \label{def:minor-condition-sat}
A~\emph{bipartite minor condition} is a~finite set $\Sigma$ of minor identities where the sets of function symbols used on the right- and left-hand sides are disjoint. More precisely, we say that a~minor condition $\Sigma$ is \emph{bipartite} over two disjoint sets of function symbols $\lang U$ and $\lang V$, if it contains only identities of the form
\[
  f(x_1,\dots,x_n) \equals
  g(x_{\pi(1)},\dots,x_{\pi(m)})
\]
where $f\in\lang U$ and $g\in \lang V$ are symbols of arity $n$ and $m$, respectively, $x_1$, \dots, $x_n$ are variables, and $\pi\colon [m] \to [n]$.

Such a~condition is said to be \emph{satisfied} in a~minion $\clo M$ on $(A,B)$ if there is an assignment $\zeta\colon \lang U \cup\lang V \to \clo M$ that assigns to each function symbol a function from $\clo M$ of the corresponding arity so that
\[
    \zeta(f)(a_1,\dots,a_n) = \zeta(g)(a_{\pi(1)},\dots,a_{\pi(m)})
\]
for each identity
\(
    f(x_1,\dots,x_n) \equals g(x_{\pi(1)},\dots,x_{\pi(m)})
\)
in $\Sigma$ and all $a_1,\dots,a_n\in A$. We say that a~minor condition is \emph{trivial} if it is satisfied in every minion, in particular, in the minion $\clo P_A$ consisting of all projections (dictators) on a~set $A$ that contains at least two elements.
\end{definition}

We note that as long as a~bipartite minor condition is satisfied in some $\clo P_A$ with $|A|\ge 2$, it is satisfied in every minion: Recall, that by definition every minion $\clo M$ is non-empty, and therefore it contains a~unary function $f$ (obtained by identifying all variables in a function from $\clo M$). Consequently, $\clo M$ contains functions defined by $(x_1,\dots,x_n) \mapsto f(x_i)$ for each $i$. These functions then behave similarly to projections in $\clo P_A$.

The symbols $f$, $g$, etc.\ in a~minor condition are abstract function symbols. Nevertheless, we sometimes (in particular, when working with specific simple minor conditions) use the same symbols to denote concrete functions that satisfy this condition.
When we want to stress the assignment of concrete functions to symbols, we use $\zeta(f)$ for the concrete function assigned to the abstract symbol~$f$.

\begin{example}\label{ex:cond-H2Hk}
Consider the following bipartite minor condition. We set $\lang U$ to contain a~single binary symbol $f$, and $\lang V$ a~single quaternary symbol $g$. The set $\Sigma$ then consists of identities:
\begin{align*}
  f(x,y) &\equals g(y,x,x,x) \\
  f(x,y) &\equals g(x,y,x,x) \\
  f(x,y) &\equals g(x,x,y,x) \\
  f(x,y) &\equals g(x,x,x,y).
\end{align*}
This condition is not trivial, since if $f = p_1$ and $g = p_i$ for some $i$ then the $i$-th identity is not satisfied, and if $f = p_2$ then the first identity forces that $g = p_1$ which contradicts any of the other identities.

The above bipartite minor condition is satisfied in the minion $\Pol(\rel H_2,\rel H_k)$ (recall Example~\ref{ex:hypergraph-colouring}) for each $k\geq 4$.  We define a~function $\zeta(g)$ by the following:
\[
  \zeta(g)(x,y,z,u) = \begin{cases}
    a & \text{if at least 3 arguments are equal to $a$;}\\
    x + 2 & \text{otherwise.}
  \end{cases}
\]
The function $\zeta(f)$ is defined by $\zeta(f)(x,y) = x$. Clearly $\zeta(f)$ and $\zeta(g)$ satisfy the required identities, also $\zeta(f)$ is in $\Pol(\rel H_2,\rel H_k)$. We now show that also $\zeta(g)$ is. 
Consider a $3\times 4$ matrix $M=(a_{ij})$ such that each column of $M$ is a triple in $\nae_2$. We need to show that applying $\zeta(g)$ to the rows of $M$ gives a triple in $\nae_k$. For contradiction, assume that we get a triple of the form $(a,a,a)$.
If $a$ is $0$ or $1$, then in each row of $M$ at least three entries are equal to $a$. In this case, one of the columns of $M$ is $(a,a,a)\not\in\nae_2$, a contradiction. Otherwise, $a$ is $2$ or $3$, which implies that the first column of $M$ is $(a-2,a-2,a-2)$, a contradiction again. 
\end{example}

\begin{example}\label{ex:cond-K3K5}
The condition $\Sigma$ from the previous example is also satisfied in the minion $\Pol(\rel K_3,\rel K_5)$ (recall Example~\ref{ex:approx-graph-col}). To see this, define $\zeta(f)(x,y)=x$ again and  define $\zeta(g)$ as follows:
\[
  \zeta(g)(x,y,z,u) = \begin{cases}
    a & \text{if at least three arguments are equal to $a$, else}\\
    0 & \text{if $x=0$ and at least one of $y,z,u$ is $0$, else}\\
    1 & \text{if $x=0$ and at least two of $y,z,u$ are $1$, else}\\
    2 & \text{if $x=0$ and at least two of $y,z,u$ are $2$, else}\\
    x + 2 & \text{otherwise.}
  \end{cases}
\]
It is straightforward to check that $\zeta(g)$ is a polymorphism and that $\zeta(f)$ and $\zeta(g)$ satisfy $\Sigma$.
\end{example}

\begin{example}\label{ex:nocond-H2Hk}
We now present a simple minor condition that is {\em not} satisfied
in $\Pol(\rel H_2, \rel H_k)$ for any $k\ge 2$:
\begin{align*}
  f(x,y) &\equals g( x,x,y,y,y,x ), \\
  f(x,y) &\equals g( x,y,x,y,x,y ), \\
  f(x,y) &\equals g( y,x,x,x,y,y ).
\end{align*}
Note that the columns of $x$'s and $y$'s on the right above correspond to the triples in $\nae_2$. Now assume that this condition is satisfied by some functions $\zeta(f), \zeta(g)$ in $\Pol(\rel H_2, \rel H_k)$, so the identities above become equalities.
Then substitute 0 for $x$ and 1 for $y$ in these equalities. The triple (column) on the right-hand side of the system is obtained by applying $\zeta(g)$
to the six triples in $\nae_2$, so it must be in $\nae_k$.
On the other hand, this triple is equal to $(b,b,b)$ where $b=\zeta(f)(0,1)$, which is not in $\nae_k$.

In Section~\ref{sec:graph_coloring} we prove that, for any $k\ge 3$, this minor condition is not satisfied in $\Pol(\rel K_k, \rel K_{2k-1})$ either --- this is the key part in our proof that $\PCSP(\rel K_k, \rel K_{2k-1})$ is \NP-hard.
\end{example}

\begin{remark} An identity of \emph{height 1} is an expression of the form
\[
    f(x_{\pi(1)},\dots,x_{\pi(n)}) \equals g(x_{\sigma(1)},\dots,x_{\sigma(m)})
\]
where $f,g$ are function symbols, $\pi\colon [n] \to [k]$ and $\sigma\colon [m] \to [k]$. Systems of such identities were considered in \cite{BOP18} (see also \cite{BKW17}). Clearly, any minor identity has height 1. Moreover, any height 1 condition, i.e. a finite set of height 1 identities, can be turned into a (bipartite) minor condition by replacing each height 1 identity by two minor identities, e.g.\ the identity above would be replaced by
\begin{align*}
    e(x_1,\dots,x_k) &\equals f(x_{\pi(1)},\dots,x_{\pi(n)})\\
    e(x_1,\dots,x_k) &\equals g(x_{\sigma(1)},\dots,x_{\sigma(m)})
\end{align*}
where $e$ is a~newly introduced symbol. It is obvious that any minion $\clo M$ satisfies the obtained bipartite minor condition if and only if it satisfies the original height~1 condition.
\end{remark}

\begin{definition}
  We define the problem $\ALC(N)$ (\emph{triviality of a~bipartite minor condition}) as the~problem where the input is a~triple $(\Sigma,\lang U,\lang V)$, with $\Sigma$ a~bipartite minor condition over $\lang U$ and $\lang V$ that involves function symbols of maximal arity $N$, and the goal is to decide whether the condition $\Sigma$ is trivial.
\end{definition}

Deciding triviality of bipartite minor conditions is essentially just a~different interpretation of the Label Cover problem that was introduced in \cite{ABSS97}. To compare these two problems, we use a~formulation of Label Cover that is closer to the one that appeared in e.g.~\cite{AGH17} and \cite{BG16}. In addition, we bound the size of the label sets by a~constant $N$. Some bounded version often appears in the literature as it is well-known that if $N\ge 3$ then it is a~NP-complete problem (see e.g.\ \cite[Lemma 4.2]{BG16}).

\begin{definition}[Label cover] \label{def:label-cover}
Fix a positive integer $N$. We define $\LC(N)$ as the following decision problem: The input is a~tuple $(U,V,E,l,r,\Pi)$ where 
\begin{itemize}
\item $G=(U,V;E)$ is a bipartite graph,
\item $l,r\leq N$ are positive integers, and 
\item $\Pi$ is a~family of maps $\pi_{e}\colon [r] \to [l]$, one for each $e\in E$.
\end{itemize} 
The goal is to decide whether there is a labelling of vertices from $U$ and $V$ with labels from $[r]$ and $[l]$, respectively, such that if $(u,v)\in E$ then the label of $v$ is mapped by $\pi_{(u,v)}$ to the label of $u$.
\end{definition}

\begin{remark}
Note that $\LC(N)$ is a CSP where vertices are variables and edges correspond to constraints: a constraint corresponding to an edge $e\in E$ is given by $\pi_e$.
\end{remark}

We interpret a~label cover instance $(U,V,E,l,r,\Pi)$ with $l,r\leq N$ as a~bipartite minor condition  $\Sigma$, an input to $\MC(N)$, as follows.
\begin{itemize}
  \item Each vertex $u \in U$ is interpreted as an~$l$-ary function symbol $f_u$, and each vertex $v\in V$ as an~$r$-ary function symbol~$g_v$.
  \item For each edge $e = (u,v)$ we add to $\Sigma$ the~identity
    \begin{equation} \tag{$\varheartsuit$}\label{id:lc}
      f_u( x_{1}, \dots, x_{l}) \equals g_v ( x_{\pi_e(1)},\dots,x_{\pi_e(r)} )
    .\end{equation}
\end{itemize}
Observe that $\Sigma$ is indeed a~bipartite minor condition over $\lang U = \{ f_u \mid u\in U \}$ and $\lang V = \{ g_v \mid v\in V \}$. We claim that the minor condition obtained in this way is trivial if and only if the original Label Cover instance has a~solution. The main difference is that a~solution to Label Cover is a~labelling while a~solution (a~witness) to triviality of minor conditions is an assignment of projections to the function symbols. Nevertheless, there is a~clear bijection between the labels and the projections: label $i$ corresponds to projection $p_i$. Clearly, a~constraint $((u,v),\pi)$ of the Label Cover is satisfied by a pair of labels $(i,j)$ if and only if assigning $p_i$ and $p_j$ to $f_u$ and $g_v$, respectively, satisfies (\ref{id:lc}).

The \emph{long code} is an~error-correcting code that can be defined as the longest code over the Boolean alphabet that does not repeat bits. Precisely, it encodes a~value $i \in [n]$ as the string of bits of length $2^n$ corresponding to the table of the function $p_i \in \clo P_2$.

Therefore, it is possible to see the problem \MC{} as just a~conjunction of Label Cover with the Long code. This conjunction has often been used before, e.g. \cite{BGS98,Hst01}. Our insight is that satisfaction of a~constraint can be extended to functions that are not projections (words that are not code words of the Long code), we can simply say that the constraint corresponding to the edge $(u,v)$ is satisfied if $f_u$ and $g_v$ satisfy (\ref{id:lc}). This approach circumvents some combinatorial difficulties of using Label Cover and the Long code, and is essential for our reduction to work.

The second (and main) advantage of using identities instead of Label Cover is that it allows us to define the following promise version. 
We will use this promise problem as an intermediate problem for our reduction, and it is interesting in its own right.

\begin{definition}\label{def:plc}
  Fix a~minion $\clo M$ and a~positive integer $N$. We define $\PLC_{\clo M}(N)$ (\emph{promise satisfaction of a~bipartite minor condition}) as the~promise problem in which, given a~bipartite minor condition $\Sigma$ that involves symbols of arity at most $N$,
one needs to output \yes{} if $\Sigma$ is trivial and \no{}
if $\Sigma$ is not satisfiable in $\clo M$.
\end{definition}

The promise in the above problem is that it is never the case that $\Sigma$ is non-trivial, but satisfied in $\clo M$.

\begin{remark}\label{rem:minor-hom-stronger-promise}
Let $\clo M_1,\clo M_2$ be two minions such that there is a minion homomorphism $\xi\colon\clo M_1 \rightarrow \clo M_2$. Then, for any $N$, $\PLC_{\clo M_2}(N)$ is obtained from $\PLC_{\clo M_1}(N)$ simply by strengthening the promise. 
To see this, observe that if some $\Sigma$ is not satisfied in $\clo M_2$ then it cannot be satisfied in $\clo M_1$. Indeed, suppose the contrary, say that some $f_i$'s and $g_j$'s from $\clo M_1$ satisfy $\Sigma$. Since $\xi$ is a~minion homomorphism from $\clo M_1$ to $\clo M_2$, it follows that $\xi(f_i)$'s and $\xi(g_j)$'s satisfy $\Sigma$ in $\clo M_2$.
\end{remark}

We can finally formulate our second main result.

\begin{theorem} \label{thm:pcsp-and-lc}
Let $(\rel A,\rel B)$ be a~template and let $\clo M$ denote its polymorphism minion.
\begin{enumerate}
  \item If $N$ is an upper bound on the size of any relation $R^{\rel A}$ of $\rel A$ as well as on $|A|$, then $\PCSP(\rel A,\rel B)$ can be reduced to $\PLC_{\clo M}(N)$ in log-space.
  \item For each $N > 0$, $\PLC_{\clo M}(N)$ can be reduced to $\PCSP(\rel A, \rel B)$ in log-space.
\end{enumerate}
\end{theorem}

Before we get to the proof, let us comment on how Theorem~\ref{thm:h1-reductions} follows from this result.

\begin{proof}[Proof of Theorem \ref{thm:h1-reductions} given Theorem \ref{thm:pcsp-and-lc}]
  We recall that we have two PCSP templates $(\rel A_1, \rel B_1)$ and $(\rel A_2, \rel B_2)$, and the corresponding polymorphism minions $\clo M_i=\Pol(\rel A_i,\rel B_i)$,  $i=1,2$. Our goal is to find a~log-space reduction from $\PCSP(\rel A_2,\rel B_2)$ to $\PCSP(\rel A_1,\rel B_1)$ given that there is a~minion homomorphism $\xi\colon \clo M_1 \to \clo M_2$. 

By Remark~\ref{rem:minor-hom-stronger-promise}, we have 
that, for any $N$, $\PLC_{\clo M_2}(N)$ is obtained from $\PLC_{\clo M_1}(N)$ by strengthening the promise. 
  Clearly, this gives a (trivial) log-space reduction from $\PLC_{\clo M_2}(N)$ to $\PLC_{\clo M_1}(N)$. 

  To conclude the proof, we connect this reduction with the two reductions from Theorem~\ref{thm:pcsp-and-lc}. Starting with $\PCSP(\rel A_2,\rel B_2)$, we reduce it to $\PLC_{\clo M_2}(N)$ where $N$ is given by the first item of Theorem~\ref{thm:pcsp-and-lc}. The above paragraph then gives us a~reduction to $\PLC_{\clo M_1}(N)$. Finally, the second item of Theorem~\ref{thm:pcsp-and-lc} provides a reduction to $\PCSP(\rel A_1,\rel B_1)$.
\end{proof}

The proof of Theorem~\ref{thm:pcsp-and-lc} is given in the following two subsections.

\subsection{From PCSP to minor conditions} \label{sec:two-prover-protocol}

We now prove Theorem~\ref{thm:pcsp-and-lc}(1). For that we need to provide a~reduction from $\PCSP$ to $\PMC$ for a given PCSP template and its polymorphism minion. This reduction follows a~standard way of proving hardness of Label Cover \cite{ABSS97}. It is built on a~two-prover protocol introduced by \cite{BGKW88}.
Our presentation of this reduction is a~generalisation of \cite[Lemma 4.2]{BG16}.

Even though we start with a~PCSP with template $(\rel A,\rel B)$, we only use the structure $\rel A$ for the construction of a bipartite minor condition from a given
instance $\rel I$ of $\PCSP(\rel A,\rel B)$.
The structure $\rel B$ will influence soundness of the reduction.
Fix an enumeration $A = \{a_1,\dots,a_n\}$ of the domain of $\rel A$, and consider an~instance $\rel I$, i.e., a~structure similar to $\rel A$. We construct a~bipartite minor condition $\Sigma=\Sigma(\rel A,\rel I)$ over $\lang U$ and $\lang V$ in the following way.
\begin{enumerate}
  \item Define $\lang U$ to be the set of symbols $f_v$ for $v\in I$, each of arity $n = |A|$.
  \item For each relation $R$ do the following: let $k=\ar(R)$, $m = |R^{\rel A}|$, and let
  	\(
      \{(a_{\pi_1(1)},\dots,a_{\pi_k(1)})\), \dots, \( (a_{\pi_1(m)},\dots,a_{\pi_k(m)})\}
  	\)
    be the list of all tuples from $R^{\rel A}$;
    for each constraint $C = ((v_1,\dots,v_k),R)$, i.e., each tuple $(v_1,\dots,v_k) \in R^{\rel I}$, introduce into $\lang V$ a~new symbol $g_{C}$ of arity $m$ and add to $\Sigma$ the following identities:
  \begin{align*}
    f_{v_1} (x_1,\dots,x_n) &\equals g_C ( x_{\pi_1(1)}, \dots, x_{\pi_1(m)} ) \\
      &\vdotsequals \\
    f_{v_k} (x_1,\dots,x_n) &\equals g_C ( x_{\pi_k(1)}, \dots, x_{\pi_k(m)} ).
  \end{align*}
\end{enumerate}
This assigns, to each $\rel I$ and $\rel A$, an instance $(\Sigma,\lang U,\lang V)$ of $\PMC(N)$. The bound $N$ is the larger of $|A|$ and the maximum of $|R^{\rel A}|$, for all relations of $\rel A$. Clearly, if $\rel A$ is fixed, the condition $\Sigma$ is computable from $\rel I$ in log-space.

\begin{example}
\label{ex:nae-sat}
We show a~reduction from NAE-\Sat{} to $\MC(6)$. NAE-\Sat{} is the same as $\CSP(\rel H_2)$ (see Example~\ref{example:1in3-nae}). Starting with an instance $\rel I$ of NAE-\Sat, for each variable $v\in I$ we add a~binary function symbol $f_v$, and for each constraint $C$ involving (not necessarily different) $v_1,v_2,v_3$ we add a~6-ary symbol $g_C$ and the following identities:
\begin{align*}
  f_{v_1} (x,y) &\equals g_C( x,x,y,y,y,x ), \\
  f_{v_2} (x,y) &\equals g_C( x,y,x,y,x,y ), \\
  f_{v_3} (x,y) &\equals g_C( y,x,x,x,y,y ).
\end{align*}
The arity of $f_v$ was chosen to be two so that the assignment of a~projection would correspond to an assignment of a~value to $v$. In particular, each variable in the above identities corresponds to one of the elements from the domain. The function $g_C$ has arity 6, since each constraint in NAE-\Sat{} has exactly $6$ satisfying assignments; the columns of variables on the right hand side of these identities correspond to these satisfying assignments (three $x$'s or three $y$'s never align). 
\end{example}

To conclude the proof of Theorem~\ref{thm:pcsp-and-lc}(1), it is enough to prove the following lemma.

\begin{lemma} \label{lem:pcsp-to-lc}
  Let $\rel A$, $\rel B$, and $\rel I$ be similar relational structures, and let $\Sigma=\Sigma(\rel A,\rel I)$ be constructed as above. Then
  \begin{enumerate}
    \item if there is a~homomorphism $h\colon \rel I \to \rel A$, then $\Sigma$ is trivial; and
    \item if $\Pol(\rel A,\rel B)$ satisfies $\Sigma$, then there is a~homomorphism from $\rel I$ to $\rel B$.
  \end{enumerate}
\end{lemma}

\begin{proof}
  To prove item (1), assume that $h\colon \rel I \to \rel A$ is a~homomorphism, and define $\zeta \colon \lang U \to \Proj_{A}$ by
  \(
    \zeta(f_v) = \proj_{i}
  \)
  where $i$ is chosen so that $h(v) = a_i$. We extend this assignment to symbols from $\lang V$: Let $C = ((v_1,\dots,v_k),R)$ be a~constraint of $\rel I$. Since $(h(v_1),\dots,h(v_k)) \in R^{\rel A}$, we can find a unique $j$ such that $(a_{\pi_1(j)}, \dots, a_{\pi_k(j)})$ is equal to  $(h(v_1),\dots,h(v_k))$. We set $\zeta(g_C) = \proj_{j}$. Clearly, this assignment satisfies all identities of $\Sigma$ involving $g_C$. Thus we found an~assignment from $\lang U \cup \lang V$ to projections that satisfies $\Sigma$, proving that $\Sigma$ is trivial.

  For item (2), let us first suppose that $\Sigma$ is satisfiable in projections and fix a satisfying assignment (of projections to symbols in $\lang U \cup \lang V$). In that case we can define a~map from $I$, equivalently $\lang U$, into $A$ by assigning to $v$ the $a_i$ corresponding to the projection assigned to $f_v$. One easy way to identify the projection is to interpret it as a~projection on the set $A$, i.e., suppose that $\zeta\colon \lang U \to \Proj_A$ is the assignment to projections, and define
  \begin{equation} \label{eq:proj-to-hom} \tag{$\clubsuit$}
    h(v) = \zeta(f_v)(a_1,\dots,a_n).
  \end{equation}
  This would give a~homomorphism to $\rel A$ by reversing the above argument. We only need a~homomorphism to $\rel B$, but we also only have $\zeta\colon \lang U \cup \lang V \to \Pol(\rel A,\rel B)$. Still, we define $h\colon I\to B$ by the rule (\ref{eq:proj-to-hom}). Clearly, this is a~well-defined assignment; we only need to prove that $h$ is a~homomorphism, i.e., that for each constraint $C = ((v_1,\dots,v_k),R)$, we have $(h(v_1),\dots,h(v_k))\in R^{\rel B}$. Consider the symbol $g_C$ and its image under $\zeta$. We know that $\zeta(g_C)$ is a~polymorphism from $\rel A$ to $\rel B$ that satisfies the corresponding identities in $\Sigma$. Let us therefore substitute $a_i$ for $x_i$, for $i\in[n]$, in those identities. We obtain the following:
  \begin{align*}
    \zeta(f_{v_1}) (a_1,\dots,a_n) &= \zeta(g_C) ( a_{\pi_1(1)}, \dots, a_{\pi_1(m)} ) \\
    &\vdotsequals \\
    \zeta(f_{v_k}) (a_1,\dots,a_n) &= \zeta(g_C) ( a_{\pi_k(1)}, \dots, a_{\pi_k(m)} ).
  \end{align*}
  But since all the tuples $(a_{\pi_1(j)},\dots,a_{\pi_k(j)})$, for $j\in[m]$, are in $R^{\rel A}$ and $\zeta(g_C)$ is a~polymorphism from $\rel A$ to $\rel B$, we get that the resulting tuple $(h(v_1),\dots,h(v_k))$ is in $R^{\rel B}$.
\end{proof}

\begin{remark}
Note that if we take $\rel A = \rel B$ in the previous lemma, we obtain that $\Sigma(\rel A,\rel I)$ is trivial if and only if $\Sigma(\rel A,\rel I)$ is satisfied in $\Pol (\rel A)$ if and only if $\rel I$ maps homomorphically to $\rel A$.
\end{remark}

\subsection{From minor conditions to PCSP} 
  \label{sec:reduction-from-lc-to-pcsp} \label{sec:lc-to-pcsp}

Our proof of Theorem \ref{thm:pcsp-and-lc}(2) follows another standard reduction in approximation: a~conjunction of Long code testing and Label Cover \cite{BGS98}. If both $\rel A$ and $\rel B$ are Boolean ($A = B = \{0,1\}$), the construction can be viewed as a~certain Long code test. Such analogy fails when $A$ and $B$ have a~different size.
Nevertheless, projections on the set $A$ (as opposed to $\{0,1\}$ for long codes) still play a~crucial role in the completeness of this reduction.
It has also appeared many times in the presented algebraic form and seems to be folklore (see e.g.\ \cite[Lemma 3.8]{CL17}).

The key idea of the construction is that the question `Is this bipartite minor condition satisfied by polymorphisms of $\rel A$?'\ can be interpreted as an~instance of $\CSP(\rel A)$. The main ingredient is that a~polymorphism is a~homomorphism from $\rel A^n$; this gives an instance whose solutions are exactly $n$-ary polymorphisms. 
Such instances are sometimes called ``indicator instances'' \cite{JCG97}. 
By considering the union of several such instances (one for each function symbol appearing in $\Sigma$) and then introducing equality constraints that reflect the identities, we get that a solution to the obtained instance corresponds to polymorphisms satisfying the identities. In detail, let us fix a~template $(\rel A,\rel B)$ and a~bound on arity $N$. We start with a~bipartite minor condition $(\Sigma,\lang U,\lang V)$ with arity bounded by $N$, and construct an instance $\rel I=\rel I_{\Sigma}(\rel A)$ of $\PCSP(\rel A,\rel B)$ in three steps:
\begin{enumerate}
  \item for each $n$-ary symbol $f$ in $\lang U \cup \lang V$, introduce into $\rel I$ a~fresh copy of $\rel A^n$ where each element $(a_1,\ldots,a_n)\in A^n$ is renamed to $v_{f(a_1,\dots,a_n)}$. To spell this out, for each relation $R^{\rel A}$, say of arity $k$, the relation $R^{\rel I}$ contains a~tuple $(v_{f(a_{11},\dots,a_{1n})},\ldots,v_{f(a_{k1},\dots,a_{kn})})$ if and only if, for each $1\le i\le n$, the tuple $(a_{1i},\ldots,a_{ki})$ is in $R^{\rel A}$.
  \item for each identity
    \(
      f( x_{1}, \dots, x_{l}) \equals g ( x_{\pi(1)},\dots,x_{\pi(r)} )
    \)
    in $\Sigma$, and all $a_1, \dots, a_l\in A$, add an~equality constraint ensuring
    \(
      v_{f( a_{1}, \dots, a_{l})} = v_{ g ( a_{\pi(1)},\dots,a_{\pi(r)} ) }
    \);
  \item identify all pairs of variables connected by (a path of) equality constraints and then remove the equality constraints.
\end{enumerate}

Clearly, the first and the second step can be done in log-space (note that the arity $n$ is bounded by the~constant $N$). The third step can be done in log-space by \cite{Rei08}.

\begin{lemma} \label{lem:long-code-test-from-pcsp}
Let $(\rel A,\rel B)$ be a~template, $N>0$, and let $\clo M = \Pol(\rel A,\rel B)$. The above construction gives a~log-space reduction from $\PLC_{\clo M}(N)$ to $\PCSP(\rel A,\rel B)$.
\end{lemma}

\begin{proof}
  To prove completeness, suppose that a bipartite minor condition $(\Sigma,\lang U,\lang V)$ is trivial and let $\zeta\colon \lang U \cup \lang V \to \Proj_A$ be a witness to that. We define a~homomorphism from $\rel I$ to $\rel A$ by setting
  \[
    h( v_{f(a_1,\dots,a_n)} ) = \zeta(f) (a_1,\dots,a_n)
  .\]
  Note that $h$ is well-defined due to step (2) in the above construction, i.e. because $\zeta$-images of function symbols satisfy identities in $\Sigma$.
  Further, since every projection is a~polymorphism of $\rel A$, $h$ is a~homomorphism.

  To prove soundness, assume that there is a~homomorphism $h\colon \rel I \to\rel B$. We reverse the above argument and define $\zeta$ by 
  \(
    \zeta(f) (a_1,\dots,a_n) = h( v_{f(a_1,\dots,a_n)} )
  \)
  for all $a_1,\dots,a_n\in A$.  Now, $\zeta(f) \in \clo M$ follows from the first step of the construction of $\rel I$, and satisfaction of the identities from $\Sigma$ follows from the second and the third steps of the construction.
\end{proof}

This concludes the proof of Theorem~\ref{thm:pcsp-and-lc}.

\section{Relational constructions} \label{sec:pp-constructions} \label{sec:constructions}

In the previous section, we have shown that the existence of a~minion homomorphism between polymorphism minions of two templates gives us a~log-space reduction between the corresponding PCSPs. The presented proof is both self-contained and succinct. In fact, we proved a~generalisation of \cite[Theorem~D.1]{BG18}, without referencing this special case, or even the Galois correspondence that underlines it \cite{Pip02}.
This section describes several concepts that give understanding of how our result relates to these two results and the theory developed for CSPs in its current form \cite{BOP18,BKW17}.

Our goal is to present a~deeper understanding of when and how our reduction works. Namely, Theorem~\ref{thm:h1-reductions} shows that $\PCSP(\rel A_2,\rel B_2)$ can be reduced to $\PCSP(\rel A_1,\rel B_1)$ under the assumption
that there is a~minion homomorphism from $\Pol(\rel A_1,\rel B_1)$ to $\Pol(\rel A_2,\rel B_2)$. In this section we show that this assumption is satisfied if and only if $(\rel A_2,\rel B_2)$ can be obtained from $(\rel A_1,\rel B_1)$ using certain relational constructions that we call \emph{pp-constructions}, and that happens if and only if there is a~homomorphism from a~certain relational structure constructed using $\Pol(\rel A_1,\rel B_1)$ and $\rel A_2$, which we call the \emph{free structure}, to~$\rel B_2$.

\subsection{Free structures} \label{sec:free}

It follows from Theorem~\ref{thm:main} that if we have an NP-hard $\PCSP(\rel A_2,\rel B_2)$, then, for any PCSP template $(\rel A_1,\rel B_1)$ such that there is a~minion homomorphism $\xi$ from $\Pol(\rel A_1,\rel B_1)$ to $\Pol(\rel A_2,\rel B_2)$, the problem $\PCSP(\rel A_1,\rel B_1)$ is also NP-hard.
The existence of such a~minion homomorphism seems to be a~global property of a~minion $\Pol(\rel A_1,\rel B_1)$. Nevertheless, given that the target template is finite, it turns out that it is not the case, as we show in this section. This allows us to obtain a reduction from knowing only the structure of polymorphisms of arity bounded by a~fixed parameter that depends on the target template --- we will use this to obtain new hardness results in Section~\ref{sec:hardness-from-pcsp}.

One way to get a~bound on the arity of polymorphisms is to compose the two reductions used to obtain Theorem~\ref{thm:main} (with the trivial reduction between the corresponding \PLC-problems) and analyse the soundness of the composite reduction.
Here we will describe another way of getting such bound, generalising \cite[Section 7]{BOP18}.
Our proof is based on two facts: First, a~minion homomorphism from a~minion $\clo M$ to a~minion $\clo A$ such that $\clo A \subseteq \clo O(A,B)$ is fully determined by its image on functions of arity $|A|$.
Second, for a~minion $\clo M$ and a~relational structure $\rel A$ there is a~most general structure $\rel F$ such that there is a~minion homomorphism from $\clo M$ to $\Pol(\rel A,\rel F)$. Structure $\rel F$ is most general in the sense that it maps homomorphically to any other structure $\rel B$ similar to $\rel A$ for which there is a~minion homomorphism from $\clo M$ to $\Pol(\rel A,\rel B)$.%
\footnote{In the case $\clo M$ is a~clone, the structure $\rel F$ is obtained by considering the free algebra generated by $\rel A$ in the variety of actions of $\clo M$ (see \cite[Section 3.2]{Opr17} for more detailed description of this case).}

In order to simplify the notation, we assume that $A = [n]$. The general case, as well as the case when $\rel A$ is an infinite structure, can be dealt with as in \cite{BOP18,Opr17}.

\begin{definition}
  Let $\rel A$ be a~finite relational structure on the set $A = [n]$, and $\clo M$ a~minion (not necessarily related to $\rel A$). The \emph{free structure of $\clo M$ generated by $\rel A$} is a~relational structure $\rel F_{\clo M}(\rel A)$ similar to $\rel A$. Its universe consists of $n$-ary functions of $\clo M$, i.e., $F_{\clo M}(A) = \clo M^{(n)}$.
  For any relation of $\rel A$, say $R^{\rel A} = \{ \tup r_1, \dots, \tup r_m \} \subseteq A^k$, the relation $R^{\rel F_{\!\clo M}(\rel A)}$ is defined as the set of all $k$-tuples $(f_1,\dots,f_k) \in F_{\clo M}(A)$ such that there exists an~$m$-ary (note that $m=|R^{\rel A}|$) function $g\in \clo M$ that satisfies
  \(
    f_i( x_1, \dots, x_n ) = g( x_{\tup r_1(i)},\dots, x_{\tup r_m(i)} )
  \)
  for each $i = 1,\dots,k$.
\end{definition}

There is a~natural minion homomorphism $\phi$ from $\clo M$ to $\Pol (\rel A,\rel F_{\clo M}(\rel A))$. It is defined by
\[
  \phi(g)(a_1,\dots,a_n) = f(x_1,\ldots,x_n)
,\]
where $f(x_1,\ldots,x_n) = g(x_{a_1},\dots,x_{a_n})$.
The relation $R^{\rel F_{\!\clo M}(\rel A)}$ is then the smallest relation $S$ such that each function from $\phi(\clo M)$ is a~polymorphism from $(A; R^{\rel A})$ to $(F_{\clo M}(A);S)$.

\begin{example} \label{ex:C(A)}
  As an example let us describe the free structures of the minion $\clo H$ of all Boolean functions of the form $x_{i_1}\meet \dots \meet x_{i_k}$.  In fact, $\clo H$ is the set of polymorphisms of the CSP template for \textsc{Horn 3-Sat}, i.e., the structure
  \[
    \rel H = (\{0,1\}; x\wedge y\to z, x\wedge y\to \neg z, \{0\},\{1\})
  \]
  (see \cite[Example 5]{BKW17}).  Note that the $n$-ary functions $f\in \clo H$ can be identified with non-empty subsets $[n]$: we identify $\emptyset \neq I\subseteq [n]$ with $f_I = \bigwedge_{i\in I} x_i$. Clearly any function from $\clo H$ can be expressed in this way. Also $f_I = f_J$ if and only if $I = J$.

  Now, fix a~relational structure $\rel A$ and assume $A = [n]$. The free structure $\rel F=\rel F_{\clo H}(\rel A)$ of $\clo H$ generated by $\rel A$ is then defined as follows. The elements of $\rel F$ are the $n$-ary functions from $\clo H$, i.e., the non-empty subsets of $A$. For a~relation $R^{\rel A} = \{ \tup r_1,\dots,\tup r_m \}$ of arity $k$, the relation $R^{\rel F}$ consists of all $k$-tuples $(f_{I_1},\dots,f_{I_k})$ for which there exists a~function $g_J \in \clo H^{(m)}$ such that
  \[
    f_{I_i}(x_1,\dots,x_n) = g_J ( x_{\tup r_1(i)},\dots, x_{\tup r_m(i)}) 
  \]
  for all $i$. Note that this identity is satisfied if and only if $I_i = \{ \tup r_j(i) \mid j\in J \}$ since the left-hand side is $\bigwedge_{a\in I_i} x_a$ and similarly, the right-hand side is $\bigwedge_{j\in J} x_{ \tup r_j(i)}$.  In other words, the elements of $R^{\rel F}$ can be viewed as those $k$-tuples $(I_1,\dots,I_k)$ of non-empty subsets of $A$ for which there exists a~subset $\{\tup r_j \mid j\in J\} \subseteq R^{\rel A}$ such that
  \(
    I_i = \{ \tup r_j(i) \mid j\in J \}
  \)
  for all $i = 1,\dots,k$.
  The resulting structure is isomorphic to the structure introduced in \cite[Section 6.1.1]{FV98} (called $U$ there) to characterise the so-called CSPs of width 1 --- now it is often referred to as the ``power structure'' of $\rel A$.
\end{example}

One of useful properties of the free structure is the following lemma that connects it to the condition $\Sigma(\rel A,\rel I)$ constructed in Section~\ref{sec:two-prover-protocol}.

\begin{lemma} \label{lem:sigma-conditions}
  Let $\clo M$ be a~minion and $\rel A$ a~relational structure. Then $\clo M$ satisfies the condition $\Sigma(\rel A,\rel F_{\clo M}(\rel A))$.
\end{lemma}

\begin{proof}
  Let $A = [n]$ and $\rel F = \rel F_{\clo M}(\rel A)$.
  We claim that $\Sigma(\rel A,\rel F)$ is satisfied in $\clo M$. Note that $\lang U$ consists of symbols $f_v$ where $v\in F = \clo M^{(n)}$. Therefore, we can define $\zeta\colon \lang U \to \clo M$ by $\zeta(f_v) = v$. To extend this map to $\lang V$, pick a~constraint $C$ corresponding to $(v_1,\dots,v_k)\in R^{\rel F}$ and let $R^{\rel A} = \{ \tup r_1,\dots,\tup r_m \}$. We need to find $\zeta(g_C)\in \clo M^{(m)}$ such that
  \[
    \zeta(f_{v_i})( x_1,\dots,x_n ) \equals \zeta(g_C)( x_{\tup r_1(i)},\dots, x_{\tup r_m(i)} )
  \]
  for all $i = 1,\dots,k$. Existence of such $\zeta(g_C)$ is guaranteed by the definition of $R^{\rel F}$ and the fact that the left-hand side is equal to $v_i(x_1,\dots,x_n)$.
\end{proof}

The following lemma relates free structures with minion homomorphisms between minions. This gives a~key correspondence between homomorphisms of certain relational structures and minion homomorphisms.

\begin{lemma} \label{lem:adjunction}
  Let $\clo M$ be a~minion, and $(\rel A,\rel B)$ a~PCSP-template.  There is a~1-to-1 correspondence between homomorphisms from $\rel F_{\clo M}(\rel A)$ to $\rel B$ and minion homomorphisms from $\clo M$ to $\Pol (\rel A,\rel B)$.
\end{lemma}

\begin{proof}
  Assume that $A = [n]$. A~homomorphism $c\colon \rel F_{\clo M}(\rel A)\to \rel B$ is then a~mapping from $\clo M^{(n)}$ to $B$. We would like to connect it to the restriction of a~minion homomorphism $\xi$ from $\clo M$ to $\Pol(\rel A,\rel B)$ on the $n$-ary functions. The only obstacle is that an~image of an~$n$-ary function under $\xi$ is a~mapping $A^n \to B$, not an element of $B$. Following the decoding used in Lemma~\ref{lem:pcsp-to-lc}, we identify such a~mapping with its image of the tuple $(1,\dots,n)$.
  Formally, to a minion homomorphism $\xi\colon \clo M \to \Pol(\rel A,\rel B)$, we assign the map  $c_\xi\colon f \mapsto \xi(f)(1,\dots,n)$. The fact that this map is a~homomorphism from $\rel F_{\clo M}(\rel A)$ to $\rel B$ follows by an argument similar to the proof of Lemma~\ref{lem:pcsp-to-lc}(2).

  It remains to prove that any such restriction can be extended in a~unique way.  Let $\phi$ denote the natural minion homomorphism from $\clo M$ to $\Pol(\rel A,\rel F_{\clo M}(\rel A))$, i.e., for $f\in \clo M$, say of arity $m$, $\phi(f)$ is the $m$-ary function from $\Pol(\rel A,\rel F_{\clo M}(\rel A))$ defined as follows: $\phi(f)(a_1,\dots,a_m)$ is the element $g \in \clo M^{(n)}$ such that 
  \begin{equation}
    f(x_{a_1},\dots,x_{a_m}) = g(x_1,\dots,x_n).
    \tag{$\vardiamondsuit$}\label{eq:col-minor}
  \end{equation}  
  Given $c\colon \rel F_{\clo M}(\rel A) \to \rel B$, we define $\xi_c\colon\clo M \to \Pol(\rel A,\rel B)$ by
  \[
    \xi_c(f)\colon (a_1,\dots,a_m) \mapsto c(\phi(f)(a_1,\dots,a_m)).
  \]
  It is easy to see that the mapping $\xi_c$ preserves taking minors. We need to prove that $\xi_c(f) \in \Pol(\rel A,\rel B)$ for each $f$, but that follows from the fact that $\phi(f) \in \Pol(\rel A,\rel F_{\clo M}(\rel A))$, i.e., it is a~homomorphism from $\rel A^m$ to $\rel F_{\clo M}(\rel A)$ and $c$ is a~homomorphism from $\rel F_{\clo M}(\rel A)$ to $\rel B$.

  For the uniqueness, suppose that $\chi\colon \clo M\to \Pol(\rel A,\rel B)$ is a~minion homomorphism such that $c_{\chi} = c$, i.e., $\chi(g)(1,\dots,n) = c(g)$ for all $g\in \clo M^{(n)}$. Since $\chi$ preserves minors, we get that for any $(a_1,\dots,a_m)\in A^m$ and $f$ satisfying (\ref{eq:col-minor}), we have $\chi(f)(a_1,\dots,a_m) = \chi(g)(1,\dots,n) = c(g)$ concluding $\chi = \xi_c$.
\end{proof}

We remark in passing that the correspondence in the above lemma is \emph{natural} in the categorical sense. More precisely, the mapping ${-}\mapsto \rel F_{-}(\rel A)$, assigning to a~minion the free structure generated by $\rel A$ is an~adjoint to the functor $\Pol(\rel A,{-})$.

\begin{remark} The previous lemma holds also for minions $\clo M$ over infinite sets since the proof applies to this case.
\end{remark}

\subsection{Pp-constructions}
\label{sec:pp-constr}

We introduce two relational constructions: \emph{pp-powers} that generalise the analogous notion for CSP templates \cite{BOP18}, and \emph{homomorphic relaxations} which generalise both strict relaxations (Definition \ref{def:strict-relaxation}) and homomorphic equivalence for CSP templates. Relaxation is in fact a~very natural notion for promise problems: the idea is that any problem that has a~stronger promise has to be at least as easy as the original problem.

\begin{definition} \label{def:relaxation}
  Assume that $(\rel A,\rel B)$ and $(\rel A',\rel B')$ are similar PCSP templates. We say that $(\rel A',\rel B')$ is a~\emph{homomorphic relaxation}\footnote{Homomorphic relaxations of CSP templates (i.e., $\rel A=\rel B$, but $\rel A'$ and $\rel B'$ can be different) have been considered in \cite{BG18b} where they are called \emph{homomorphic sandwiches}.} of $(\rel A,\rel B)$ if there are homomorphisms $h_A\colon\rel A' \to \rel A$ and $h_B\colon\rel B \to \rel B'$.
\end{definition}

All the relaxations that appear in this paper are homomorphic relaxations, therefore we will usually omit the word `homomorphic'. In particular, the strict relaxation defined in Definition \ref{def:strict-relaxation} is a~special case of homomorphic relaxation, namely one where both $h_A$ and $h_B$ are identity maps on the corresponding domains.

Clearly, if $(\rel A',\rel B')$ is a~relaxation of $(\rel A,\rel B)$, then the trivial reduction (which does not change the input) is a~reduction from $\PCSP(\rel A,\rel B)$ to $\PCSP(\rel A',\rel B')$, since it is a~strengthening of the promise, as mentioned above.

The following is a generalisation of the definition of a pp-power for CSP templates (see {\cite[Definition 14]{BKW17}}). 

\begin{definition}
  Let $(\rel A,\rel B)$ and $(\rel A',\rel B')$ be two PCSP templates. We say that $(\rel A',\rel B')$ is an ($n$-th)~\emph{pp-power} of $(\rel A,\rel B)$ if $A' = A^n$, $B' = B^n$, and, if we view $k$-ary relations on $\rel A'$ and $\rel B'$ as $kn$-ary relations on $A$ and $B$, respectively, then $(\rel A',\rel B')$ is pp-definable in $(\rel A,\rel B)$ in the sense of Definition~\ref{def:strict-relaxation}.
\end{definition}

\begin{lemma} \label{lem:pp-gives-minor}
  Let $(\rel A_1,\rel B_1)$ and $(\rel A_2,\rel B_2)$ be two PCSP templates. If 
  \begin{enumerate}
    \item $(\rel A_2,\rel B_2)$ is a~relaxation of $(\rel A_1,\rel B_1)$, or
    \item $(\rel A_2,\rel B_2)$ is a~pp-power of $(\rel A_1,\rel B_1)$,
  \end{enumerate}
  then there is a~minion homomorphism $\xi \colon \Pol(\rel A_1,\rel B_1) \to \Pol(\rel A_2,\rel B_2)$.
\end{lemma}

\begin{proof}
  For $i=1,2$, let $\clo M_i=\Pol(\rel A_i,\rel B_i)$.
  \begin{enumerate}
    \item Assume $h_A\colon \rel A_2\to \rel A_1$ and $h_B\colon \rel B_1 \to \rel B_2$ are homomorphisms, and define $\xi\colon \clo M_1 \to \clo M_2$ by $\xi(f)\colon (a_1,\dots,a_n) \mapsto h_B(f(h_A(a_1),\dots,h_A(a_n)))$. It is easy to see that $\xi(f)$ is a~polymorphism of $(\rel A_2,\rel B_2)$ (as it is a~composition of homomorphisms), and also that $\xi$ preserves minors.

    \item Assume that $(\rel A_2,\rel B_2)$ is $n$-th pp-power of $(\rel A_1,\rel B_1)$. We define $\xi$ to map each $f$ to its component-wise action on $A_2 = A_1^n$, i.e.,
  \[
    \xi(f)\colon (\tup a_1,\dots,\tup a_m) \mapsto
      \bigl(f( \tup a_1(1),\dots, \tup a_m(1) ), \dots, f( \tup a_1(n),\dots, \tup a_m(n) )\bigr)
  .\]
  Clearly, $\xi$ is minor preserving. Also $\xi(f)$ is a~polymorphism, since $f$ is a~polymorphism from $\rel A_1$ to $\rel B_1$, and therefore it preserves all pp-definable relations.
  \qedhere
  \end{enumerate}
\end{proof}

The above together with Theorem~\ref{thm:h1-reductions} proves that both the constructions yield a~log-space reduction between the corresponding PCSPs. We remark that the same can also be proven directly. Since the both constructions yield a~reduction, we can combine them and still retain log-space reductions. This motivates the next definition, which again is similar to a notion for CSPs (see {\cite[Section 3.4]{BKW17}}). 

\begin{definition} \label{def:pp-constructible}
We say that $(\rel A',\rel B')$ is \emph{pp-constructible} from $(\rel A,\rel B)$ if there exists a~sequence
\[
  (\rel A,\rel B) = (\rel A_1,\rel B_1) , \dots, (\rel A_k,\rel B_k) = (\rel A',\rel B')
\]
of templates where each $(\rel A_{i+1},\rel B_{i+1})$ is a~pp-power or a~homomorphic relaxation of $(\rel A_{i},\rel B_{i})$.
We say that $(\rel A',\rel B')$ is pp-constructible from $\rel A$ to mean that it is pp-constructible from $(\rel A,\rel A)$.
\end{definition}

\begin{corollary} \label{cor:pp-constructible}
  If $(\rel A',\rel B')$ is pp-constructible from $(\rel A,\rel B)$ then there is a~minion homomorphism from $\Pol(\rel A,\rel B)$ to $\Pol(\rel A',\rel B')$.
\end{corollary}

\begin{proof} Assume that we have a~sequence $(\rel A_1,\rel B_1)$, \dots, $(\rel A_k,\rel B_k)$ as in the definition of pp-constructibility, and $\clo M_i = \Pol(\rel A_i,\rel B_i)$ for $i\in [k]$. That means, that by Lemma~\ref{lem:pp-gives-minor}, we have minion homomorphisms $\xi_i\colon \clo M_i\to \clo M_{i+1}$ for all $i<k$. The minion homomorphism from $\Pol (\rel A,\rel B) = \clo M_1$ to $\Pol (\rel A',\rel B') = \clo M_k$ is obtained by composing all~$\xi_i$'s.
\end{proof}

An example of a~template that can be pp-constructed is a~template obtained using the free structure.

\begin{lemma} \label{lem:free-is-pp}
  Let $(\rel A_1,\rel B_1)$ be a~PCSP template, $\clo M = \Pol (\rel A_1,\rel B_1)$, and $\rel A_2$ be a~relational structure. Then the template $(\rel A_2,\rel F_{\clo M}(\rel A_2))$ is a~relaxation of a~pp-power of $(\rel A_1,\rel B_1)$.
\end{lemma}

\begin{proof}
  Let us first comment on some ideas underlying the proof. We will argue that the whole reduction from $\PCSP(\rel A_2,\rel F)$, where $\rel F = \rel F_{\clo M}(\rel A_2)$, to $\PCSP(\rel A_1,\rel B_1)$ according to the proof of Theorem~\ref{thm:h1-reductions} is basically a~pp-construction of $(\rel A_2,\rel F)$ from $(\rel A_1,\rel B_1)$, and in particular it is a~relaxation of a~pp-power. Note that Theorem~\ref{thm:h1-reductions} applies, since we have a~minion homomorphism $\phi\colon\clo M\to \Pol(\rel A_2,\rel F)$, as defined in the proof of Lemma~\ref{lem:adjunction}.

  For the formal proof, assume that $A_2 = [n]$, and let $N = |A_1|^n$. We first describe an $N$-th pp-power of $(\rel A_1,\rel B_1)$, and then argue that $(\rel A_2,\rel F)$ is a~homomorphic relaxation of this power. Let $R^{\rel A_2} = \{ \tup r_1,\dots, \tup r_m \}$ be a~relation of $\rel A_2$ of arity $k$.
  Recall that the relation $R^{\rel F}$ is defined as the set of all $k$-tuples $(f_1,\dots,f_k)$ of functions from $A_1^n$ to $B_1$ for which there exists a~polymorphism $g\colon \rel A_1^m \to \rel B_1$ such that
  \begin{equation} \tag{$\spadesuit$} \label{eq:free-is-pp}
    f_i(x_1,\dots,x_n) = g(x_{\tup r_1(i)},\dots, x_{\tup r_m(i)})
  \end{equation}
  for all $i \in [k]$. Note that we do not need to require that $f_i$'s are polymorphisms since the mentioned identities enforce that all $f_i$'s are minors of the polymorphism $g$. We argue that the set of all such tuples is pp-definable in $\rel B_1$ as a~$kN$-ary relation. In fact, since this is just a~bipartite minor condition, we can use the construction from Section~\ref{sec:reduction-from-lc-to-pcsp} and adapt it to provide a~pp-formula: the free variables of the pp-formula are labelled by $f_i(a_1,\dots,a_n)$ where $i\in [k]$ and $a_1,\dots,a_n \in A_1$, the quantified variables are labelled by $g(a_1,\dots,a_m)$ where $a_1,\dots,a_m \in A_1$. Further, let $\lang R^{\rel A_1}$ denote the set of all relational symbols of $\rel A_1$. We use the following pp-formulas to define the pp-power (see Definition~\ref{def:strict-relaxation}):
  \begin{multline*}
    \Psi_R (v_{f_i(a_1,\dots,a_n)}, \dots) = \exists_{a_1,\dots,a_m\in A_1} v_{g(a_1,\dots,a_m)}\\
      \bigwedge_{S\in \lang R^{\rel A_1}}\bigwedge_{\tup s_1,\dots,\tup s_m \in S^{\rel A_1}}
        (v_{g(s_1(1),\dots,s_m(1))},\dots, v_{g(s_1(\ar(S)),\dots,s_m(\ar (S)))}) \in S
      \meet {}\\
      \bigwedge_{i\in [k]}\bigwedge_{a_1,\dots,a_n \in A_1}
        v_{f_i(a_1,\dots,a_n)} = v_{g(a_{\tup r_1(i)},\dots,a_{\tup r_m(i)})}.
  \end{multline*}
  The first conjunction ensures that the values assigned to $v_{g(a_1,\dots,a_m)}$'s give a~valid polymorphism $g$, the second conjunction then ensures that this polymorphism will satisfy (\ref{eq:free-is-pp}).
  Let us denote the resulting pp-power by $(\rel A_1',\rel B_1')$.

  We still need to find homomorphisms from $\rel A_2$ to $\rel A_1'$ and from $\rel B_1'$ to $\rel F$.
  The first one can be constructed following an argument from Lemma~\ref{lem:long-code-test-from-pcsp}: We define $h_A\colon \rel A_2 \to \rel A_1'$ by $h_A(a) = p_a$ where $p_a$ denotes the projection on the $a$-th coordinate (note that $\rel A_1'$ is an $|A_1|^{|A_2|}$-th pp-power of $\rel A_1$). Further, if $(a_1,\dots,a_k)\in R^{\rel A_2}$, i.e., $(a_1,\dots,a_k) = \tup r_i$ for some $i$, then choosing the value $a_i$ for $v_{g(a_1,\dots,a_m)}$ will give a~satisfying assignment of $\Psi_R^{\rel A_1}$.
  The homomorphism from $\rel B_1'$ to $\rel F$ is easier: since $F\subseteq B_1'$, and $\Psi_R^{\rel B_1}(f_1,\dots,f_k)$ if and only if $(f_1,\dots,f_k) \in R^{\rel F}$ (this was the motivation behind the definition of $\Psi_R$), we can define $h_B\colon B_1' \to F$ as any extension of the identity mapping on $F$.
\end{proof}

Finally, we are ready to formulate and prove the main result of this section.

\begin{theorem} \label{thm:minor-homomorphism-is-pp-constructibility} \label{thm:minionhom-and-pp}
  Let $(\rel A_i,\rel B_i)$ for $i=1,2$ be PCSP templates and $\clo M_i = \Pol(\rel A_i,\rel B_i)$.
  The following are equivalent:
  \begin{enumerate}
    \item There exists a~minion homomorphism $\xi \colon \clo M_1 \to \clo M_2$.
      \label{it3:h1pp}
    \item $\clo M_2$ satisfies all bipartite minor conditions satisfied in $\clo M_1$.
      \label{it4:h1pp}
    \item $\clo M_2$ satisfies the condition $\Sigma(\rel A_2,\rel F_{\!\clo M_1}(\rel A_2))$.
      \label{itnew:h1pp}
    \item There exists a~homomorphism from $\rel F_{\!\clo M_1}(\rel A_2)$ to $\rel B_2$.
      \label{it5:h1pp}
    \item $(\rel A_2,\rel B_2)$ is a~homomorphic relaxation of a~pp-power of $(\rel A_1,\rel B_1)$.
      \label{it1:h1pp}
    \item $(\rel A_2,\rel B_2)$ is pp-constructible from $(\rel A_1,\rel B_1)$.
      \label{it2:h1pp}
  \end{enumerate}
\end{theorem}

\begin{proof}
  The implication (\ref{it1:h1pp})\textto(\ref{it2:h1pp}) is trivial, (\ref{it2:h1pp})\textto(\ref{it3:h1pp}) follows from Lemma~\ref{lem:pp-gives-minor}, (\ref{it3:h1pp})\textto(\ref{it4:h1pp}) follows directly from the definition (see Definition~\ref{def:minor-hom} and the comment after it), (\ref{it4:h1pp})\textto(\ref{itnew:h1pp}) follows from Lemma~\ref{lem:sigma-conditions} for $\clo M = \clo M_1$ and $\rel A = \rel A_2$, and (\ref{itnew:h1pp})\textto(\ref{it5:h1pp}) follows directly from Lemma~\ref{lem:pcsp-to-lc}(2).
  
  Finally, we prove (\ref{it5:h1pp})\textto(\ref{it1:h1pp}). Lemma~\ref{lem:free-is-pp} gives that $(\rel A_2,\rel F_{\!\clo M_1}(\rel A_2))$ is a~relaxation of a~pp-power of $(\rel A_1,\rel B_1)$, the homomorphism from $\rel F_{\!\clo M_1}(\rel A_2)$ to $\rel B_2$ then proves that $(\rel A_2,\rel B_2)$ is a~relaxation of $(\rel A_2,\rel F_{\!\clo M_1}(\rel A_2))$. We obtain the desired claim  by composing the two relaxations into one.
\end{proof}

\section{Hardness from the PCP theorem}
  \label{sec:hardness-from-lc}

The celebrated PCP theorem \cite{ALMSS98,AS98} is a~starting point for many proofs of inapproximability of many problems. As an example, we refer to the work of H{\aa}stad \cite{Hst01} that derives inapproximability of several CSPs from the PCP theorem. Also note that the PCP theorem itself can be formulated as a~result on inapproximability of the CSP (see \cite[Theorem 1.3]{Din07}). Concrete results on hardness of many PCSPs rely on the PCP theorem (e.g.\ \cite{BG16,AGH17,Kho01,Hua13,DRS05}). A~common approach for using the PCP theorem is to first derive hardness of some approximation version of Label Cover, or some of its variants, and then using gadgets reduce from Label Cover to PCSP.

In the scope of this paper, we described a~reduction from $\MC$, which is essentially Label Cover, to $\PCSP$ (see Section~\ref{sec:MC}). The present section then derives some algebraic conditions for applicability of this reduction when applied to Label Cover itself.
We note that there are many approximation versions of Label Cover and many variants of PCPs. Most can be used as a~starting point for a~reduction, and we do not provide an exhaustive description. We focus on a~few versions of Label Cover including a~plain approximation version thereof and one that was used in \cite{DRS05} to obtain \NP-hardness of approximate hypergraph colouring.

As mentioned above, the general approach for a~reduction from (some version of) Label Cover to some $\PCSP(\rel A,\rel B)$ is to interpret an instance of Label Cover as a~minor condition (as described in the beginning of Section~\ref{sec:MC}), and then reduce from $\PMC_{\clo M}$, where $\clo M=\Pol(\rel A,\rel B)$, to $\PCSP(\rel A,\rel B)$ using Theorem~\ref{thm:pcsp-and-lc}. To ensure that this composite reduction works, it is enough to relate the \yes- and \no-instances of the LC, or its variant, and the corresponding PMC. Proving completeness (i.e., that \yes-instances are preserved) is usually straight-forward, while proving soundness (i.e., preserving \no-instances) require some extra work.

\begin{remark}
Not all known \NP-hardness proofs for PCSPs can be easily adapted for our approach. One such case is the hardness proof for $\PCSP(\rel K_k,\rel K_c)$ where $c=2^{\Omega(k^{1/3})}$ and $k$ is large enough \cite{Hua13}. The key difference is that in our approach the completeness part of reductions is trivial and the soundness is the hard part, but in the proof from \cite{Hua13} the situation is opposite. 
\end{remark}

Before, we get to more general cases, let us briefly focus on a~reduction from the plain (exact) Label Cover.

\subsection{Reduction from Label Cover}

As noted before, exact Label Cover is essentially the same as deciding non-triviality of bipartite minor conditions. This also means that we can use Theorem~\ref{thm:pcsp-and-lc}(1) to give an immediate proof that \MC{} is \NP-hard. A~similar reduction from some \NP-hard CSP is a~commonly used argument for \NP-hardness of Label Cover.

\begin{theorem}
  \label{thm:mc-is-np-hard}
  $\MC(N)$ is \NP-hard for each $N\geq 3$.
\end{theorem}

\begin{proof}
We reduce from 1-in-3-\Sat{} using Theorem \ref{thm:pcsp-and-lc}(1): Let $\rel T$ denote the CSP template of 1-in-3-\Sat{} and $R^{\rel T}$ its ternary relation (see Example~\ref{example:1in3-nae}).
It is well-known that every polymorphism of $\rel T$ is a~projection, i.e., $\Pol(\rel T)=\clo P_2$.
The mentioned theorem then gives a~reduction from $\CSP(\rel T)$ to $\PMC_{\clo P_2}(N)$ for each $N$ that is larger then both the size of the domain of $\rel T$, which is 2, and the number of tuples in $R^{\rel T}$, which is 3. Therefore, $\PMC_{\clo P_2}(N)$ is \NP-hard for each $N\geq 3$. Since $\PMC_{\clo P_2}(N)$ is the same as $\MC(N)$, we obtain the desired hardness.
\end{proof}

To align \no-instances of $\MC(N)$ and $\PMC_{\clo M}(N)$, we need that $\clo A$ does not satisfy any non-trivial bipartite minor conditions involving symbols of arity at most $N$. Since the theorem above gives \NP-hardness of $\MC(N)$ for any $N = 3$, this means that arity at most 3 is enough, as we state in the following direct corollary of Theorem~\ref{thm:pcsp-and-lc}(2) and Theorem~\ref{thm:mc-is-np-hard}.

\begin{corollary} \label{prop:ternaries-are-enough-for-hardness}
  If $\Pol(\rel A,\rel B)$ does not satisfy any non-trivial bipartite minor condition of arity at most three, $\PCSP(\rel A,\rel B)$ is NP-hard.
  \qed
\end{corollary}

The assumption of the above corollary can be also satisfied by constructing a~minor preserving map $\xi\colon \clo M^{(3)} \to \clo P_2^{(3)}$. The search for such a~minion homomorphism can be easily automated, which, with a~suitable implementation, can be useful for small enough structures. We also remark that since it is enough to work with binary and ternary functions, this might simplify some combinatorial arguments.

We now show that a~minion that does not satisfy any non-trivial bipartite minor condition of small arities cannot satisfy such a~condition of large arity. 
This has been also observed in \cite[Section 5.3]{BP18}.

\begin{proposition} \label{prop:3-is-enough}
  The following are equivalent for every minion $\clo M$.
  \begin{enumerate}
    \item $\clo M$ does not satisfy any non-trivial bipartite minor condition of arity at most three.
    \item There exists a~minion homomorphism from $\clo M$ to $\clo P_2$.
  \end{enumerate}
\end{proposition}

\begin{proof}
  The implication (2)\textto(1) is obvious, let us prove (1)\textto(2).
  We use the notation from the proof of Theorem~\ref{thm:mc-is-np-hard} and, follow the proof of Theorem~\ref{thm:minor-homomorphism-is-pp-constructibility}(\ref{itnew:h1pp})\textto(\ref{it3:h1pp}). The key observation is that the condition $\Sigma(\rel T,\rel F_{\clo M}(\rel T))$ is of arity at most three, and therefore it is trivial by the assumption and Lemma~\ref{lem:sigma-conditions}. The bound on arity is clear from the construction: the condition is composed of identities of the form
  \begin{align*}
    f_{v_1}(x,y) &\equals g_{(v_1,v_2,v_3),R}(x,x,y)\\
    f_{v_2}(x,y) &\equals g_{(v_1,v_2,v_3),R}(x,y,x)\\
    f_{v_3}(x,y) &\equals g_{(v_1,v_2,v_3),R}(y,x,x)
  \end{align*}
  where $(v_1,v_2,v_3) \in R^{\rel F_{\!\clo M}(\rel T)}$. Thus, by Lemma~\ref{lem:pcsp-to-lc}(2), we get that $\rel F_{\clo M}(\rel T) \to \rel T$ which implies that there is a~minor homomorphism from $\clo M$ to $\Pol(\rel T) = \clo P_2$ by (e.g.) Lemma~\ref{lem:adjunction}.
\end{proof}

\begin{remark}
  The above proposition can be easily generalised for minions on infinite sets: the finiteness was used only to ensure that $\clo M^{(2)}$, and consequently also the free structure $\rel F_{\clo M}(\rel T)$, is finite. This can be circumvented by a~standard compactness argument. See e.g.\ \cite[Lemma III.5]{BMOOPW19}.
\end{remark}

The significance of the above proposition and the considerations of this subsection is that it implies that a~reduction from $\LC(N)$ using polymorphism gadgets (following Section~\ref{sec:lc-to-pcsp}) works for some $N$ if and only if the same reduction works for $N=3$.

\begin{example}[See also Example~\ref{ex:pol-3,4-maps-to-projections}]
  \label{ex:pol-3,4-maps-to-projections-2}
  In \cite{BG16}, the authors prove that $\PCSP(\rel K_3,\rel K_4)$ is \NP-hard by a~reduction from $\LC$. This in fact means that there is a~simple gadget reduction from e.g.\ 1-in-3-\Sat, or equivalently, from $\CSP(\rel T)$. This reduction then uses only properties of binary and ternary polymorphisms from $\rel K_3$ to $\rel K_4$.

  We use Section~\ref{sec:algebraic-reductions} for the reduction: Given an instance $\rel I$ of 1-in-3-\Sat, we first obtain a~bipartite condition $\Sigma = \Sigma(\rel I,\rel T)$. This condition (similarly as in the proof of Proposition~\ref{prop:3-is-enough} involves only symbols of arity two and three. From $\Sigma$, we then obtain a~graph $\rel G = \rel I_{\Sigma}(\rel K_3)$. The goal is to show that if $\rel I \to \rel T$ then $\rel G$ is 3-colourable (which is the easier part, and we refer to Section~\ref{sec:algebraic-reductions} for the proof), and that if $\rel G$ is 4-colourable then $\rel I \to \rel T$. The latter can be observed by constructing a~homomorphism from the free structure $\rel F_{\clo M}(\rel T)$ to $\rel T$ where $\rel M = \Pol(\rel K_3,\rel K_4)$. Let us sketch one such homomorphism $h\colon \rel F_{\clo M}(\rel T) \to \rel T$.

  The elements of $\rel F_{\clo M}(\rel T)$ are binary polymorphisms of $(\rel K_3,\rel K_4)$. Therefore, to construct a~homomorphism, we need to choose a~Boolean value for each of the colourings of $\rel K_3^2$ with $4$ colours. Such colourings are easy to describe; in fact the `trash colour lemma' that we mentioned in Example~\ref{ex:pol-3,4-maps-to-projections} is much easier to prove for binary functions. Therefore, we can map a~colouring $f$ to $0$ if, after removing the trash colour, $f$ depends on the first variable, and $1$ if it depends on the second. This defines the mapping $h$. The only hard part is to prove that $h$ is a~homomorphism, for that it is necessary to look at ternary polymorphisms of $(\rel K_3,\rel K_4)$ since they determine the relational structure $\rel F_{\clo M}(\rel T)$.
\end{example}

\subsection{Reduction from Gap Label Cover}
  \label{sec:hardness-from-gap-lc}

Let us continue with a~more general reduction from Gap Label Cover. The framework that we present generalises the approach of e.g.\ \cite{AGH17,BG16}, where the authors proved \NP-hardness of various fixed-template PCSPs using polymorphisms by reducing from Gap Label Cover.

\begin{definition}
The \emph{Gap Label Cover} problem with parameters $\delta$~(\emph{completeness}), $\epsilon$~(\emph{soundness}), and $N$, denoted by $\GLC_{\delta,\epsilon}(N)$, is a~promise problem which, given an instance of $\LC(N)$,
\begin{itemize}
  \item accepts if there is an~assignment that satisfies at least $\delta$-fraction of the given constraints, or
  \item rejects if no assignment satisfies more than $\epsilon$-fraction of the given constraints.
\end{itemize}
\end{definition}

The hardness of Gap Label Cover with perfect completeness (i.e., $\delta = 1$) and some soundness $\epsilon < 1$ can be directly obtained from the PCP theorem of \cite{ALMSS98,AS98}. The soundness parameter can then be brought down to arbitrarily small $\epsilon > 0$ using the parallel repetition theorem of Raz \cite{Raz98} at the cost of increasing~$N$.

\begin{theorem}[\cite{ALMSS98,AS98,Raz98}] \label{thm:glc}
  There exist constant $K_1,K_2 > 0$ such that for every $\epsilon > 0$ and every $N \geq K_1\epsilon^{-K_2}$, $\GLC_{1,\epsilon}(N)$ is \NP-hard.
\end{theorem}

Fix some $\PCSP(\rel A,\rel B)$ and let $\clo M=\Pol(\rel A,\rel B)$.
As before, we can reduce Gap Label Cover to $\PCSP(\rel A,\rel B)$ via Theorem~\ref{thm:pcsp-and-lc}(2) if we can ensure that the~standard transformation from $\GLC_{1,\epsilon}(N)$ to $\PLC_{\clo M}(N)$ (i.e. interpreting a~Label Cover instance as a~bipartite minor condition) is a valid reduction. The completeness is immediate. To prove that \no-answers are preserved, one typically uses some limitations of the structure (which may be non-trivial to obtain) of bipartite minor conditions satisfied in $\clo M$ and shows the contrapositive.
More specifically, one proves that if a~bipartite minor condition $\Sigma$ is satisfied in $\clo M$ then some part of $\Sigma$ containing more than $\epsilon$-fraction of identities can be satisfied in projections, and hence the corresponding Label Cover instance has an assignment satisfying more than $\epsilon$-fraction of constraints.
To formulate a~theorem characterising this approach, we use the following definition that captures \no-instances of Gap Label Cover in an algebraic language.

\begin{definition} Let $\epsilon > 0$. We say that a~bipartite minor condition $\Sigma$ is \emph{$\epsilon$-robust} if no $\epsilon$-fraction of identities from $\Sigma$ is trivial.
\end{definition}

\begin{theorem} \label{thm:no-epsilon-robust-is-hard}
  There exist constants $K_1,K_2 > 0$ such that the following holds. If there exists an $\epsilon > 0$ and $N\geq K_1\epsilon^{-K_2}$ such that $\Pol(\rel A,\rel B)$ does not satisfy any $\epsilon$-robust bipartite minor condition involving symbols of arity at most $N$, then $\PCSP(\rel A,\rel B)$ is \NP-hard.
\end{theorem}

\begin{proof}
  Let $K_1$ and $K_2$ be the same as in Theorem~\ref{thm:glc}, so $\GLC_{1,\epsilon}(N)$ is \NP-hard for any $\epsilon$ and $N$ satisfying the assumptions, and let $\clo M = \Pol(\rel A,\rel B)$.  We transform an instance of $\GLC_{1,\epsilon}(N)$ to an instance of $\PLC_{\clo M}(N)$ in the usual way. Since label cover constraints are in 1-to-1 correspondence with identities in the bipartite minor condition, any \no-instance of $\GLC_{1,\epsilon}(N)$ is transformed into a~bipartite minor condition that is $\epsilon$-robust, and therefore fails in $\clo M$, i.e., this condition is a~\no{} instance of $\PLC_{\clo M}(N)$. This shows soundness of the reduction, and completeness is obvious. The statement now follows from Theorem~\ref{thm:pcsp-and-lc}(2).
\end{proof}

Since for every $\epsilon > 0$ there exists a~suitable $N$ in the above theorem, we can formulate the following useful weaker version of this theorem.

\begin{corollary} \label{cor:disregard-N}
  Let $\epsilon > 0$, if $\Pol(\rel A,\rel B)$ does not satisfy any $\epsilon$-robust bipartite minor condition, then $\PCSP(\rel A,\rel B)$ is \NP-hard.
\end{corollary}

The only general approach (that we are currently aware of) to verify that a~minion $\clo M$ satisfies no $\epsilon$-robust bipartite condition is a~probabilistic method applied as follows.
Suppose we can find a probability distribution on the set of pairs of arity-preserving mappings $\clo M \rightarrow \clo P_2$ so that, for each minor identity
  \[
    f(x_1,\dots,x_n) \equals g( x_{\pi(1)},\dots, x_{\pi(m)} ),
  \]
the following happens. If the identity is satisfied by functions $\zeta(f)$ and $\zeta(g)$ in $\clo M$, and we select $\tau, \tau'\colon \clo M \rightarrow \clo P_2$ according to the probability distribution, then the probability that the identity is satisfied by the projections $p_i = \tau(\zeta(f))$ and $p_j = \tau'(\zeta(g))$ (i.e., that $i = \pi(j)$) is at least $\epsilon$.  In this case, $\clo M$ satisfies no $\epsilon$-robust bipartite minor condition. Indeed, if a~bipartite minor condition $(\Sigma,\lang U,\lang V)$ is satisfied in $\clo M$ through a~map $\zeta\colon \lang U\cup \lang V\to \clo M$ and we randomly assign a~projection to each function $h$ in the image of $\zeta$ as above, then the probability that a~single minor identity in $\Sigma$ is satisfied is at least $\epsilon$ and it follows that the expected fraction of satisfied identities in $\Sigma$ is at least $\epsilon$. Consequently, some  $\epsilon$-fraction of identities in $\Sigma$ is trivial.

The random procedure of assigning projections to function in $\clo M$ can be thought of as a~generalised minion homomorphisms from $\clo M$ to $\clo P_2$. We will not study this concept in full generality since, in applications, a~very simple version has been sufficient so far. Namely, projections $\tau(f)$ as well as $\tau'(f)$ are chosen independently for each $f \in \clo M$ and uniformly from a~subset of projections (which corresponds to a~subset of coordinates). In this case, the probabilistic argument gives us the following lemma.

\begin{lemma} \label{lem:small-sets}
Let $\clo M$ be a~minion and let $C\colon \mathbb N \to \mathbb N$. Assume that there exists a~mapping $I$ assigning to each $h \in \clo M^{(n)}$ a~subset of $[n]$ of size at most $C(n)$ such that for each $\pi\colon [m] \to [n]$ and each $g\in \clo M^{(m)}$ we have
\[
\pi(I(g(x_1, \dots, x_m)) \cap I(g(x_{\pi(1)},\dots,x_{\pi(m)})) \neq \emptyset
.\]
Then $\clo M$ satisfies no $(1/C(N)^2)$-robust bipartite minor condition involving symbols of arity at most $N$.
\end{lemma}

\begin{proof}
  Let $(\Sigma,\lang U,\lang V)$ be a~bipartite minor condition that involves symbols of arity at most $N$ and is satisfied in $\clo M$ which is witnessed by $\zeta\colon \lang U\cup \lang V \to \clo M$. To ease readability, we write $f^{\clo M} = \zeta(f)$ for each $f$. Let $\epsilon = 1/C(N)^2$. We want to show that $\Sigma$ is not $\epsilon$-robust, i.e., we want to find an assignment $\rho\colon \lang U \cup \lang V \to \clo P_2$ that satisfies at least $\epsilon$-fraction of identities in $\Sigma$.

  Let us choose such an assignment by choosing a~coordinate $i\in T(f^{\clo M})$ uniformly at random, and setting $\rho(f) = p_i$. We claim that the probability that any single identity in $\Sigma$ is satisfied is at least $\epsilon$. Indeed, if we consider an identity
  \begin{equation} \label{eq:5a-Sigma}
    f(x_1,\dots,x_n) \equals g( x_{\pi(1)},\dots, x_{\pi(m)} ),
  \end{equation}
  in $\Sigma$, we get that $I(f^{\clo M}) \cap \pi(I(g^{\clo M})) \neq \emptyset$, and consequently, there is a~choice of $i\in I(f^{\clo M})$ and $j\in I(g^{\clo M})$ such that $i = \pi(j)$, i.e., $f = p_i$ and $g = p_j$ satisfies (\ref{eq:5a-Sigma}). This means that the probability is at least
  \[
    1/ \bigl(\abs{I(f^{\clo M})}\cdot \abs{I(g^{\clo M})}\bigr) \geq 1 / C(n)C(m) \geq 1/ C(N)^2
  .\]
  It follows that the expected fraction of identities that gets satisfied by $\rho$ is at least $\epsilon$, which means that there is an assignment of projections that satisfies at least $\epsilon$-fraction of identities in~$\Sigma$.
\end{proof}

Directly from the above and Corollary~\ref{cor:disregard-N}, we get the following.

\begin{corollary} \label{cor:small-sets}
Let $\clo M = \Pol(\rel A,\rel B)$ and let $C$ be a~constant such that there exists a~mapping $I$ assigning to each $h \in \clo M^{(n)}$ a~subset of $[n]$ of size at most $C$ such that for each $\pi\colon [m] \to [n]$ and each $g\in \clo M^{(m)}$ we have
\[
\pi(I(g(x_1, \dots, x_m)) \cap I(g(x_{\pi(1)},\dots,x_{\pi(m)})) \neq \emptyset
.\]
Then $\PCSP(\rel A,\rel B)$ is \NP-hard.
\end{corollary}

We will now explain some known hardness results that can be obtained using the above corollary. Possibly the simplest way to choose the set $I(f)$ of coordinates for a~function $f$ is to choose only coordinates that can influence the value of $f$. This is formalised in the following definition.

\begin{definition}
Let $f\colon A^n \to B$, a~coordinate $i\in [n]$ is called \emph{essential} if there exist $a_1,\ldots,a_n$ and $b_i$ in $A$ such that
\[
  f(a_1,\ldots,a_{i-1},a_i,a_{i+1},\ldots, a_n)\ne f(a_1,\ldots,a_{i-1},b_i,a_{i+1},\ldots, a_n). 
\]
A~minion $\clo N$ on $(A,B)$, where $A$ and/or $B$ can be infinite, is said to have \emph{essential arity at most $k$}, if each function $f\in \clo N$ has at most $k$ essential variables.  We say that $\clo N$ has bounded essential arity if it has essential arity at most $k$ for some~$k$.
\end{definition}

The following is a~generalisation of \cite[Theorem~4.7]{AGH17}.

\begin{proposition}
  \label{prop:bounded-arity}
  Let $(\rel A,\rel B)$ be a~PCSP template and let $\clo M = \Pol(\rel A,\rel B)$. Assume that there exists a~minion homomorphism $\xi\colon \clo M \rightarrow \clo N$ some minion $\clo N$, possibly on infinite sets, which has bounded essential arity and does not contain a~constant function (i.e., a~function without essential variables). Then $\PCSP(\rel A,\rel B)$ is \NP-hard.
\end{proposition}

\begin{proof}
  Let $C$ be the bound on the essential arity of $\clo N$. For $f\in \clo M$, we set $I(f)$ to be the set of all essential coordinates of $\xi(f)$. Clearly, $1\le \abs{I(f)} \leq C$. In order to apply Corollary~\ref{cor:small-sets}, we only need to show that if
  \[
    f(x_1,\dots,x_n) \equals g(x_{\pi(1)},\dots,x_{\pi(m)})
  \]
  then $I(f) \cap \pi(I(g)) \neq \emptyset$. Indeed, if $i\in [m]$ is an~essential coordinate of $\xi(f)$, there are two tuples $a_1,\dots,a_n$ and $b_1,\dots,b_n$ such that $a_i \neq b_i$, $a_{i'} = b_{i'}$ for all $i\neq i'$ and $\xi(f)(a_1,\dots,a_n) \neq \xi(f)(b_1,\dots,b_n)$. The last disequality together with the above identity (which is preserved by $\xi$) gives that
  \[
    \xi(g)(a_{\pi(1)},\dots,a_{\pi(n)}) \neq \xi(g)(b_{\pi(1)},\dots,b_{\pi(n)})
  .\]
  The two tuples of arguments differ only on coordinates $j'$ with $\pi(j')=i$, therefore $\xi(g)$ has to depend essentially on at least one coordinate from $\pi^{-1}(i)$. This shows that $i\in \pi(I(g))$, and therefore $I(f) \subseteq \pi(I(g))$. We get the claim since $I(f)\neq \emptyset$ by assumption.
\end{proof}

The following technical notion is a~slight strengthening of one that was used in \cite{BG18} as a~sufficient condition for \NP-hardness of some Boolean PCSPs.

\begin{definition} \label{def:C-fixing}
Let $C>0$ be a~constant. A~minion $\clo M \subseteq \clo O(\{0,1\})$ is said to be \emph{strongly $C$-fixing} if for each $f\in \clo M$ there exists a~(\emph{fixing}) set $I_f\subseteq[\ar(f)]$, $|I_f| \le C$ such that $f(x_1,\dots,x_{\ar(f)}) = 0$ whenever $x_i = 0$ for all $i\in I$, and similarly, $f(x_1,\dots,x_{\ar(f)}) = 1$ whenever $x_i = 1$ for all $i\in I$.
\end{definition}

We remark that a~strongly C-fixing minion does not need to have bounded essential arity.
The following proposition is a~generalisation of \cite[Theorem~5.1]{BG18}.

\begin{proposition} \label{corollary:C-fixing}
Let $(\rel A,\rel B)$ be a~PCSP template and let $\clo M=\Pol(\rel A,\rel B)$.
Assume that there exists a~minion homomorphism $\xi\colon \clo M\to \clo N$ for some minion $\clo N\subseteq \clo O(\{0,1\})$ which is strongly $C$-fixing for some $C>0$.
Then $\PCSP(\rel A,\rel B)$ is \NP-hard.
\end{proposition}

\begin{proof}
  We set $I(f)$ to be some fixing set of $\xi(f)$ of size at most~$C$. Observe that no function $f\in \clo N$ can have two disjoint fixing sets: if $I$ and $J$ would be disjoint and fixing, we would get that for a~tuple $(x_1,\dots,x_n)$ such that $x_i = 0$ for all $i\in I$ and $x_j = 1$ for all $j\in J$, we would get that $f(x_1,\dots,x_n)$ is both $0$ and $1$. On the other hand, a~$\pi$-image of a~fixing set $I$ of $g$ is a~fixing set for the minor $g(x_{\pi(1)},\dots,x_{\pi(m)})$. This shows that if $f$ and $g$ satisfy
  \[
    f(x_1,\dots,x_n) \equals g(x_{\pi(1)},\dots,x_{\pi(m)})
  ,\]
  and consequently their $\xi$-images satisfy the same identity, then $I(f) \cap \pi(I(g)) \neq \emptyset$. Hence Corollary~\ref{cor:small-sets} applies.
\end{proof}

\subsection{Reduction from Multilayered Label Cover}
  \label{sec:hardness-from-lglc}

\NP-hardness of Gap Label Cover is just one of the consequences of the PCP theorem, although arguably the most prevalent one. Other variants of GLC can also be used in proving \NP-hardness of some PCSP. In this section, we describe in algebraic terms a~reduction from the so-called \emph{Layered Label Cover} (or \emph{Multilayered PCP}) that has been used in \cite{DRS05} to prove hardness of approximate hypegraph colouring, i.e., $\PCSP(\rel H_2,\rel H_k)$. We also briefly comment on how Layered Label Cover is applied in this case. Gap Layered Label Cover was introduced in \cite{DGKR05}.

The Layered Label Cover ($\LLC{}$) is a generalisation of $\LC{}$ from bipartite graphs to $L$-partite graphs. The partite sets are referred to as \emph{layers}.

\begin{definition}[Layered Label Cover]
Fix positive integers $N$ and $L \geq 2$. We define $\LLC(L,N)$ as the following decision problem. The input is a~tuple 
\[
( (V_i, r_i)_{1 \leq i \leq L},  (E_{ij},\Pi_{ij})_{1 \leq i < j \leq L})
\]
where 
\begin{itemize}
  \item Each $G_{ij}=(V_i,V_j;E_{ij})$ is a~bipartite graph,
  \item $r_i\leq N$ are positive integers, and 
  \item each $\Pi_{ij}$ is a~family of maps $\pi_{ij,e}\colon [r_j] \to [r_i]$, one for each $e\in E_{ij}$.
\end{itemize} 
The goal is to decide whether there is a~labelling of vertices from $V_1, \ldots, V_L$ with labels from $[r_1], \ldots, [r_L]$, respectively, such that if $(u,v) \in E_{ij}$ then the label of $v$ is mapped by $\pi_{ij,(u,v)}$ to the label of $u$.
\end{definition}

Note that the LLC is a~CSP, where edges $(u,v) \in E_{ij}$ correspond to constraints between the layers $i$ and $j$.

Just like $\LC(N) = \LLC(2,N)$ is essentially the same problem as $\MC(N)$ (recall Subsection~\ref{sec:MC}), the $L$-layered version $\LLC(L,N)$ is essentially the same problem as $\LMC(L,N)$, the problem of deciding triviality of \emph{$L$-layered minor conditions}.
To be precise, for pairwise disjoint sets $\lang V_1,\ldots,\lang V_L$ of function symbols, an $L$-layered minor condition is a~tuple $(\Sigma,\lang V_1,\ldots,\lang V_L)$ where $\lang V_i$ are disjoint sets of function symbols, and $\Sigma$ is a~set of identities of the form $f(x_1,\ldots,x_{r_i})=g(x_{\pi(1)},\ldots,x_{\pi(r_j)})$ where $i<j$, $f\in \lang V_i$ and $g\in \lang V_j$.

For a minion $\clo M$, one can naturally define $\PLMC_{\clo M}(L,N)$, the promise version of $\LMC(L,N)$, in the same way as $\PMC_{\clo M}(N)$ is obtained from $\MC(N)$, i.e., the \yes-instances of $\LMC(L,N)$ are all the trivial $L$-layered minor conditions $\Sigma$ involving symbols of arity at most $N$, and \no-instances those $L$-layered minor conditions $\Sigma$ involving symbols of arity at most $N$ that are not satisfied in $\clo M$.
Moreover, if $\clo M = \Pol(\rel A,\rel B)$, then, for any fixed $L$ and $N$, $\PLMC_{\clo M}(L,N)$ can be reduced to $\PCSP(\rel A,\rel B)$ in log-space in the same way as in the proof of Theorem~\ref{thm:pcsp-and-lc}(2) (see also Subsection~\ref{sec:reduction-from-lc-to-pcsp}).

We will exploit these observations in a~similar way as in the previous subsection --- by exploring when a~known NP-hard gap version of $\LLC{}$ can be naturally reduced to $\PLMC_{\clo M}(L,N)$. An important gain of using more layers is a certain kind of density of $\LLC{}$ instances that can be required while preserving hardness.

\begin{definition}
An instance of $\LLC(L,N)$ is called \emph{weakly dense} if for any $1 < m <L$, any $m$ layers $i_1 < \cdots < i_m$, and any sets $S_j \subseteq V_{i_j}$ such that $\abs{S_j} \geq 2\abs{V_{i_j}}/m$ for every $1 \leq j \leq m$, there exist $1 \leq j < j' \leq m$ such that $\abs{E_{i_ji_{j'}} \cap (S_j \times S_{j'})} \geq \abs{E_{i_ji_{j'}}}/m^2$.
\end{definition}

We are ready to state a~gap version of Layered Label Cover from~\cite{DGKR05}.

\begin{definition}
The \emph{Gap Layered Label Cover} problem with parameters $\epsilon$, $L$, and $N$ denoted by $\GLLC_{\epsilon}(L,N)$, is a~promise problem in which, given a weakly dense instance of $\LLC(L,N)$, one needs to
\begin{itemize}
  \item accept if there is an~assignment that satisfies all the constraints, or
  \item reject if for every $1 \leq i < j \leq L$ no assignment satisfies more than $\epsilon$-fraction of the constraints between the layers $i$ and $j$.
\end{itemize}
\end{definition}

\begin{theorem} [{\cite[Theorem 4.2]{DGKR05}}] \label{thm:gllc}
  There exist constants $K_1, K_2 > 0$ such that for every $1 \geq \epsilon >0$, every $L \geq 2$, and every $N \geq K_1 \epsilon^{-K_2L}$ the problem $\GLLC_{\epsilon}(L,N)$ is NP-hard.
\end{theorem}

An analogue of Theorem~\ref{thm:no-epsilon-robust-is-hard}, which we state now, allows us to break a~polymorphism minion into finitely many sets that do not need to be minions and check the `no $\epsilon$-robust minor condition' property for each of the pieces separately (note that Definition~\ref{def:minor-condition-sat} of satisfaction of a~minor condition in $\clo M$ makes sense for an arbitrary subset of $\clo O(A,B)$).

\begin{theorem} \label{thm:no-piece_epsilon-robust-is-hard}
  Let $(\rel A,\rel B)$ be a PCSP template and let $\epsilon: \mathbb{N} \to \mathbb{N}$ be a function such that $\epsilon(N) \in \Omega(N^{-K})$ for each $K > 0$. Assume that $\clo M = \Pol(\rel A,\rel B)$ is a union of finitely many sets $\clo M_1, \ldots, \clo M_{k}$, none of which satisfies (for any $N$) any $\epsilon(N)$-robust bipartite minor condition involving symbols of arity at most $N$. Then $\PCSP(\rel A,\rel B)$ is \NP-hard.
\end{theorem}

\begin{proof}
  We show that, with an appropriate choice of $\epsilon', L$ and $N$, the usual transformation of a~$\GLLC_{\epsilon'}(L,N)$ to a~$\PLMC_{\clo M}(L,N)$ is a~reduction. Since $\GLLC_{\epsilon'}(L,N)$ is NP-hard by Theorem~\ref{thm:no-piece_epsilon-robust-is-hard} (for a~large enough $N$) and  $\PLMC_{\clo M}(L,N)$ reduces to $\PCSP(\rel A,\rel B)$, the claim will follow.

  It is straightforward that the \yes-instances are mapped to \yes-instances with any choice of parameters. The non-trivial part is to show that, with a suitable choice of the parameters, \no-instances are mapped to \no-instances. We set $L = 2k^2$ and choose $N,\epsilon'$ so that
  \[
    \epsilon' = \epsilon(N)/(4k^2) \mbox{ and } N \geq K_1 (\epsilon')^{-K_2L}
  ,\]
  which is possible as $\epsilon(N) = \Omega(N^{-K})$ for $K < 1/K_2L$.

  We verify the contrapositive. Consider a~weakly dense instance $I = ( (V_i, r_i)_{i \leq L},$  $(E_{ij},\Pi_{ij})_{i < j})$ of $\GLLC_{\epsilon'}(L,N)$ that is not mapped to a~\no-instance $\Sigma$ (over the sets of symbols $(\lang V_i)_{1 \leq i \leq L}$) of $\PLMC_{\clo M}(L,N)$, i.e., there exists a mapping $\zeta$ from $\cup_{i=1}^L{\lang V_i}$ to $\clo M$ which witnesses that $\clo M$ satisfies~$\Sigma$. 

  We colour each symbol $f$ in $\lang V_i$ by some index $i$ such that $\zeta(f) \in \clo M_i$.  Next we colour each layer $\lang V_i$ by the most popular colour among its members. Since the number of colours is at most $k$ and the number of layers is $2k^2$, at least $m = 2k$ layers $i_1, \dots, i_m$ received the same colour, say $c$. Finally, let $S_j$ (where $j=1, \dots, m$) denote the set of all elements of $\lang V_{i_j}$ with colour $c$.

  We have $\abs{S_j} \geq \abs{V_{i_j}}/k = 2\abs{V_{i_j}}/m$, therefore, as $I$ is weakly dense, there exist $j < j'$ such that at least $1/m^2$-fraction of identities in $\Sigma$, which are between $\lang V_{j}$ and $\lang V_{j'}$, is between $S_j$ and $S_{j'}$. The assignment $\zeta$ witnesses that the system of all identities between $S_j$ and $S_{j'}$ is satisfied in $\clo M_c$. Since $\clo M_c$ does not satisfy any $\epsilon(N)$-robust condition with symbols of arity at most $N$, some $\epsilon(N)$-fraction of identities between $S_j$ and $S_{j'}$ is trivial. Since $\epsilon' = \epsilon(N)/(4k^2) = \epsilon(N)/m^2$, it follows that some $\epsilon'$-fraction of identities between $\lang V_{i_j}$ and $\lang V_{i_{j'}}$ is trivial, which for the original instance $I$ means that at least $\epsilon'$-fraction of the constraints between $V_{i_j}$ and $V_{i_{j'}}$ is satisfied, so $I$ is not a~\no-instance of $\GLLC_{\epsilon'}(L,N)$.
\end{proof}

The assumption that $\clo M_i$ does not satisfy any $\epsilon(N)$-robust bipartite minor condition may be  verified by a probabilistic argument as in the previous section. We remark that Lemma~\ref{lem:small-sets} applies to sets $\clo M$ of functions that are not minions, e.g., sets $\clo M_i$ from the previous theorem.

\subsubsection*{Hardness of approximate hypergraph colouring}

In the rest of this subsection we will sketch a~proof of \NP-hardness of approximate graph colouring (see Example~\ref{ex:hypergraph-colouring}). This was proved in \cite{DRS05}, and in our sketch we will reuse some of the combinatorial arguments from that paper. The theorem can be stated as follows.

\begin{theorem}[\cite{DRS05}] \label{thm:DRS05}
  For any $k\geq 2$, $\PCSP(\NAE_2, \NAE_k)$ is \NP-hard.
\end{theorem}

Our proof loosely follows the one in \cite{DRS05} with the main difference being the usage of our general theory. In particular, we use the above universal reduction from LLC to PLMC instead of an ad-hoc reduction from LLC to hypergraph colouring used in \cite{DRS05}.

The combinatorial core of the proof is a~strengthening Lov\'asz's theorem on the chromatic number of Kneser graphs. To state the theorem in the algebraic language, recall that $\Ham(\tup{t})$ denotes the Hamming weight of a~tuple $\tup t \in \{0,1\}^n$. We call two tuples $\tup{u},\tup{v} \in \{0,1\}^n$ \emph{disjoint} if $u_i = 0$ or $v_i = 0$ (or both) for every $i\in [n]$.

\begin{theorem}[\cite{Lov78}] \label{thm:lovasz}
 Let $f \colon E_2^N \to E_k$ and $s = \lfloor (N+1-k)/2 \rfloor$. Then there exist disjoint tuples $\vc{u}^f,\vc{v}^f \in E_2^N$ such that $\Ham(\vc{u}^f) = \Ham(\vc{v}^f) = s$ and $f(\vc{u}^f) = f(\vc{v}^f)$.
\end{theorem}

This theorem has an immediate consequence for polymorphisms of $(\NAE_2,\NAE_k)$. Namely, it implies that, for every polymorphism $f\colon \NAE_2^N \to \NAE_k$, there exists $c \in E_k$ and a~\emph{$c$-avoiding set $A$} of coordinates of size at most $k$, by which we mean that  $f(\vc{w}) \neq c$ for every tuple $\vc{w} \in E_2^N$ such that $w_i = 1$ for every $i \in A$. Indeed, take $c = f(\vc{u}^f) = f(\vc{v}^f)$ and
\[
  A = \{ i \in [N] \mid u^f_i = v^f_i = 0 \}.
\]
Since $\vc{w}$ is such that $w_i = 1$ for every $i \in A$, the matrix with rows $\vc{u}^f$, $\vc{v}^f$, $\vc{w}$ has all columns in $\nae_2$, and therefore $f$ applied to the rows of this matrix gives a tuple in $\nae_k$. As $c = f(\vc{u}^f) = f(\vc{v}^f)$, we get $f(\vc{w}) \neq c$ for every such a tuple. The size of $A$ is indeed at most $\abs{N - 2\lfloor (N+1-k)/2\rfloor} \leq k$.

These $c$-avoiding sets have similar properties as fixing sets (see~Definition~\ref{def:C-fixing}): the $\pi$-image of a $c$-avoiding set for a function is clearly $c$-avoiding for the corresponding minor. However, a~function can have two disjoint $c$-avoiding sets (indeed, every set can be $c$-avoiding if this colour is never used by $f$), so exactly the same argument does not work. 

This issue can be resolved by strengthening Theorem~\ref{thm:lovasz} from Kneser graphs to Schrijver graphs~\cite{Schri78}. A~consequent counting argument then gives the following. 

\begin{lemma} [{\cite[Lemma 2.2]{DRS05}}]
Let $f \colon E_2^N \to E_k$ and $s = \lfloor (N+1-k)/2 \rfloor$. There exists $c^f \in E_k$ such that
 \begin{itemize}
 \item there exist disjoint tuples $\vc{u}^f,\vc{v}^f \in E_2^N$ such that $\Ham(\vc{u}^f) = \Ham(\vc{v}^f) = s$ and $f(\vc{u}^f) = f(\vc{v}^f) = c^f$ and
 \item the fraction of elements of $\{\vc{u} \colon \Ham(\vc{u}) = s\}$ such that $f(\vc{u}) = c^f$ is $\Omega(N^{-k-1})$ (where the constant hidden in $\Omega$ depends only on $k$).
 \end{itemize}
\end{lemma}

The first item gives us a~small (of size at most $k$) $c^f$-avoiding set for every function $f$. The second item together with another counting argument (see \cite[Claim 4.5]{DRS05}) implies that there is no collection of pairwise disjoint small $c^f$-avoiding sets of size bigger than $K_3\log N$, where the constant $K_3$ depends only on $k$. Now we can finish the proof using Theorem~\ref{thm:no-piece_epsilon-robust-is-hard} and Lemma~\ref{lem:small-sets}.
For $c=1, \dots, k$ we define $\clo M_c = \{f \in \clo M \mid c^f = c\}$ and, for $f \in \clo M_c$, we define $I(f)$ as the union of a~maximal collection of pairwise disjoint small $c^f$-avoiding sets. For any $\pi\colon [m] \to [n]$ and an $m$-ary polymorphism $g$ such that both $g(x_1, \dots, x_m)$ and $f(x_1, \dots, x_n) = g(x_{\pi(1)}, \dots, x_{\pi(m)})$ are in $\clo M_c$ we have that $\pi(I(g)) \cap I(f) \neq \emptyset$ since the $\pi$-image of a $c$-avoiding set (and hence of $I(g)$) is a $c$-avoiding set for $f$. Moreover, $\abs{I(h)} \leq K_3\log N$ if $\ar(h) \leq N$. By Lemma~\ref{lem:small-sets}, no $\clo{M}_c$ satisfies any $(1/K_3^2\log^2 N)$-robust bipartite minor condition, and then $\PCSP(\NAE_2,\NAE_k)$ is NP-hard by Theorem~\ref{thm:no-piece_epsilon-robust-is-hard}.
\section{Hardness from other PCSPs}
    \label{sec:hardness-from-pcsp}

In this section we derive hardness results by reductions from other PCSPs by directly applying results from Sections \ref{sec:algebraic-reductions} and \ref{sec:pp-constructions}.
We give a~concise algebraic characterisation of PCSPs that admit a~reduction from approximate hypergraph colouring, i.e., from $\PCSP(\rel H_2,\rel H_k)$ for some $k$ (which is \NP-hard by Theorem~\ref{thm:DRS05}), and we apply this characterisation to special cases of approximate graph colouring and graph homomorphism. We give similar characterisation results for the existence of a~reduction from approximate graph colouring or graph homomorphism problems (which are currently not known to be \NP-hard in full generality).

\subsection{Hardness from approximate hypergraph colouring}

As a~starting point for our reductions, we can use Theorem~\ref{thm:DRS05}. Following our approach, the first step in a reduction from $\PCSP(\NAE_2,\NAE_k)$ to $\PCSP(\rel A,\rel B)$ is a~reduction from $\PCSP(\NAE_2,\NAE_k)$ to $\PMC_{\clo H_k}(6)$ where $\clo H_k = \Pol(\NAE_2,\NAE_k)$, given by Theorem~\ref{thm:pcsp-and-lc}(2). To analyse when we can continue further, it is useful to understand which bipartite minor conditions are not satisfied in $\clo H_k$.

We explained in Example \ref{ex:nocond-H2Hk} that the following condition is not satisfied in $\clo H_k$ for any $k\ge 2$:
\begin{align*}
  t(x,y) &\equals o(x,x,y,y,y,x)\\
  t(x,y) &\equals o(x,y,x,y,x,y)\\
  t(x,y) &\equals o(y,x,x,x,y,y),
\end{align*}
We prove below that having polymorphisms satisfying this condition is the only obstacle for a~reduction from approximate hypergraph colouring to a~given PCSP template.

\begin{definition}
An \emph{Olšák function} is a~6-ary function $o$ that satisfies
\[
  o(x,x,y,y,y,x) \equals o(x,y,x,y,x,y) \equals o(y,x,x,x,y,y).
\]
\end{definition}

The above identities appeared in Olšák's paper \cite{Ols17}. The algebraic significance of these identities is that they give the weakest non-trivial Maltsev condition for all idempotent algebras (also infinite ones), though we will not use this fact in our paper.

\begin{theorem} \label{thm:no-olsak}
  Let $\clo M$ be a~minion. The following are equivalent.
  \begin{enumerate}
    \item There exists $K\geq 2$ and a~minion homomorphism $\xi\colon \clo M \to \clo H_K$;
    \item $\clo M$ does not contain an~Olšák function.
  \end{enumerate}
\end{theorem}

\begin{proof}
  Clearly, if $\clo M$ contains an~Olšák function, then its image  under a minion homomorphism would also be an~Olšák function. Since $\clo H_K$ does not contain such a~function, this proves the implication (1)\textto(2).  

  For the other implication assume that $\clo M$ does not contain an~Olšák function. Consider the 3-uniform hypergraph $\rel F = \rel F_{\clo M}(\NAE_2)$ obtained as the free structure of $\clo M$ generated by $\NAE_2$. Note that the vertices of $\rel F$ are binary functions in $\clo M$, and three such vertices $f,g,h$ are connected by a~hyperedge if there is a~$6$-ary function $o'$ such that
  \begin{align*}
    f(x,y) &\equals o'(x,x,y,y,y,x)\\
    g(x,y) &\equals o'(x,y,x,y,x,y)\\
    h(x,y) &\equals o'(y,x,x,x,y,y).
  \end{align*}
  Since $\clo M$ does not contain an Olšák function, we get that $\rel F$ does not contain hyperedges $(f,g,h)$ with $f=g=h$, which is enough to guarantee that it is colourable with $K = |F| = |\clo M^{(2)}|$ colours (note that $|\clo M^{(2)}|$ is finite). In other words, there is a homomorphism from $\rel F$ to $\rel H_K$, and we have a~minion homomorphism from $\clo M$ to $\Pol(\NAE_2,\NAE_K) = \clo H_K$ by Lemma~\ref{lem:adjunction}.
\end{proof}

\begin{corollary}  \label{cor:no-olsak}
  For every finite template $(\rel A,\rel B)$ that does not have an~Olšák polymorphism, $\PCSP(\rel A,\rel B)$ is \NP-hard.
\end{corollary}

\begin{proof}
  Let $\clo M = \Pol(\rel A,\rel B)$. From the previous theorem, we know that there is a~minion homomorphism from $\clo M$ to $\clo H_K$ for some $K$. Therefore, the claim follows from Theorem~\ref{thm:main} and Theorem~\ref{thm:DRS05}.
\end{proof}

We give two applications of the above corollary in the next subsection. 
Note that the absence of Olšák polymorphism is a useful sufficient condition for hardness of PCSPs, but by no means a universal tool for this, as witnessed by, e.g., the hardness results of \cite{KO19} (see also Proposition~\ref{prop:no-olsak-in-d-2d}).

\subsection{Hardness of approximate graph colouring and homomorphism}

In this section we settle some of the open problems mentioned in Examples~\ref{ex:approx-graph-col} and~\ref{ex:approx-graph-hom}. We start by considering approximate graph colouring $\PCSP(\rel K_k,\rel K_c)$. Recall that the strongest known \NP-hardness results (without additional assumptions) for this problem are the cases when $k\ge 3$ is arbitrary and  $c\le 2k-2$~\cite{BG16} and when $c\le 2^{\Omega(k^{1/3})}$ and $k$ is large enough~\cite{Hua13}.

We prove that distinguishing between $k$-colourable graphs and those not even $(2k-1)$-colourable is \NP-hard for all $k\geq 3$. In particular, we prove that colouring a $3$-colourable graph with $5$ colours is \NP-hard. This improves the results of \cite{KLS00,BG16} which are the best known bounds for small $k$.

By Corollary~\ref{cor:no-olsak}, it is enough to show that no polymorphism from $\rel K_k$ to $\rel K_{2k-1}$ is an~Olšák function. 

\begin{lemma} \label{lem:k,2k-1}
  $\Pol(\rel K_k,\rel K_{2k-1})$ does not contain an Olšák function.
\end{lemma}

\begin{proof}
  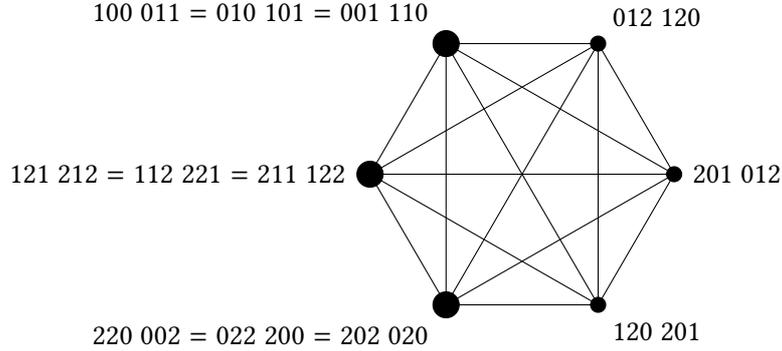
\begin{figure}
    \begin{tikzpicture}[scale = 2]
      \foreach \angle/\label/\size in {%
        120/{100 011${}={}$010 101${}={}$001 110}/3.46pt,
        180/{121 212${}={}$112 221${}={}$211 122}/3.46pt,
        240/{220 002${}={}$022 200${}={}$202 020}/3.46pt,
        60/{012 120}/2pt,
        300/{120 201}/2pt,
        0/{201 012}/2pt%
      }
      \node[circle,draw,fill,label={\angle:\label},inner sep = \size] at (\angle:1) {};

      \foreach \a in {0,60,...,300} {
          \draw (\a:1) -- (\a+60:1);
          \draw (\a:1) -- (\a+120:1);
        }
      \foreach \a in {0,60,120}
          \draw (\a:1) -- (\a+180:1);
    \end{tikzpicture}
    \caption{A 6-clique in the graph $\rel G$ for $k=3$. The vertices on the left were involved in the gluing that produced $\rel G$.}
    \label{fig:no-olsak-in-3,5}
  \end{figure}
  To show that $(\rel K_k,\rel K_{2k-1})$ does not have an Olšák polymorphism, we consider the indicator construction for the corresponding identities (see Section \ref{sec:two-prover-protocol}). Such a polymorphism can be viewed as $(2k-1)$-colouring of $\rel K_k^6$ which, for all $x,y\in K_k$, colours the vertices $(x,x,y,y,y,x)$, $(x,y,x,y,x,y)$, and $(y,x,x,x,y,y)$ by the same colour (which may depend on $x,y$). Since such vertices must be coloured the same, we can identify them.  So we construct a~graph $\rel G$ by considering the sixth power of $\rel K_k$, and then gluing together every triple of vertices of the form $(x,x,y,y,y,x)$, $(x,y,x,y,x,y)$, and $(y,x,x,x,y,y)$.

  We claim that $\rel G$ contains a~$2k$-clique. Namely, for all $i \in K_k$, we consider:
  \[
    a_i = (i,i+1,i+2,i+1,i+2,i) \text{ and }
    b_i = (i+1,i,i,i,i+1,i+1)
  \]
  where addition is modulo $k$. Note that the vertices $b_i$ are involved in the gluing that produced the graph, therefore they can be also represented by tuples $(i,i+1,i,i+1,i,i+1)$ and $(i,i,i+1,i+1,i+1,i)$. Fig.~\ref{fig:no-olsak-in-3,5} depicts this clique in $\rel G$ for $k=3$. We claim that all pairs of these vertices are connected by an edge. Clearly for any $i\neq j$, there is an edge between $a_i$ and $a_j$, as well as between $b_i$ and $b_j$.

  For edges between $a_i$ and $b_j$, we consider the following cases:
  \begin{itemize}
    \item $j \notin \{i,i+1\}$. Then also $j+1 \notin \{i+1,i+2\}$, so there is an edge between $b_j$ and $a_i$ since $b_j$ is represented by the tuple $(j,j,j+1,j+1,j+1,j)$.
    \item $j = i$. There is an edge between $a_i$ and $b_j$ since $b_j$ is represented by the tuple $(j+1,j,j,j,j+1,j+1)$.
    \item $j = i + 1$. There is an edge between $a_i$ and $b_j$ since $b_j$ is represented by the tuple $(j,j+1,j,j+1,j,j+1)$.
  \end{itemize}
  Altogether, we get that $\{a_i, b_i \mid i \in K_k \}$ is a~clique of size $2k$, and therefore $\rel K_k^6$ has no $(2k-1)$-colouring
which is an Olšák function. 
\end{proof}

As a~direct corollary of the above lemma and Corollary~\ref{cor:no-olsak}, we get the following.

\begin{theorem}\label{thm:hardness-k-vs-2k-1}
  Deciding whether a~given graph is $k$-colourable or not even $(2k-1)$-colourable is \NP-hard for any $k\geq 3$.
  \qed
\end{theorem}

We remark that \cite{BKO19} contains a more direct proof of Theorem~\ref{thm:hardness-k-vs-2k-1}, which mimics the proofs from Section~\ref{sec:algebraic-reductions}, without using the theory that we developed in Section~\ref{sec:pp-constructions}. 

We also remark that the presented method does not work for $c\geq 2k$ because $\Pol(\rel K_k,\rel K_{2k})$, and therefore also $\Pol(\rel K_k,\rel K_c)$, contains an~Ol\v{s}\'ak polymorphism. For a~detailed discussion, see~Proposition \ref{prop:no-olsak-in-d-2d}.

\medskip

We now consider the problem of distinguishing whether a~given graph homomorphically maps to $\rel C_5$ (a~5-cycle), or it is not even $3$-colourable. This problem (mentioned above in Example~\ref{ex:approx-graph-hom}) was suggested as an intriguing open case in~\cite{BG18}. This problem can be understood as the `other' relaxation of 3-colourability of graphs. For example, in approximate graph colouring, we may be interested whether a~graph is 3-colourable, or not 4-colourable, i.e., we relax $(\rel K_3,\rel K_3)$ to $(\rel K_3,\rel K_4)$. If we want a~different relaxation, we change the first $\rel K_3$ into some graph $\rel G$ that is $3$-colourable. The odd cycles are then a~natural choice, since any graph that is not $2$-colourable contains an~odd cycle. Therefore, if we could prove that $\PCSP(\rel C_{2k+1},\rel K_3)$ is \NP-hard, we would have a~complete picture for problems of the form $\PCSP(\rel G,\rel K_3)$. We further discuss problems of the form $\PCSP(\rel G,\rel H)$ for two graphs $\rel G$ and $\rel H$ in Section~\ref{sec:promise-graph-colouring}.

\begin{lemma} \label{C_5,K_3}
  $\Pol(\rel C_5,\rel K_3)$ does not contain an~Olšák function.
\end{lemma}
\begin{proof}
Suppose that $\Pol(\rel C_5,\rel K_3)$ contains an Olšák function $o$, which can be viewed as a 3-colouring of $\rel C_5^6$.
Assume that the vertices of $\rel C_5$ are $0,1,2,3,4$, in this cyclic order.
Consider the 6-tuples $b_i= (i+1,i,i,i,i+1,i+1), 0\le i \le 4,$ where the addition is modulo 5. 
These tuples form a 5-cycle in $\rel C_5^6$. Since there is only one 3-colouring of $\rel C_5$ up to automorphisms of $\rel C_5$ and $\rel K_3$, we can assume without loss of generality that
$o(b_0)=1, o(b_1)=0, o(b_2)=1, o(b_3)=2, o(b_4)=0.$

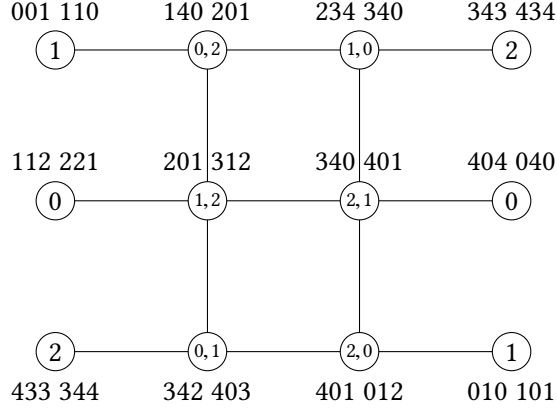
\begin{figure}
  \begin{tikzpicture}[every node/.style={draw, fill = white, inner sep = 0, minimum size = 1.2pc}, scale = 2 ]
    \newcommand{\ch}[2]{\scriptsize$#1,\!#2$}
    \draw (0,0) -- (3,0);
    \draw (0,1) -- (3,1);
    \draw (0,2) -- (3,2);
    \draw (1,0) -- (1,2);
    \draw (2,0) -- (2,2);
    \node [circle] at (0,2) [label = above:{001 110}] {1};
    \node [circle] at (0,1) [label = above:{112 221}] {0};
    \node [circle] at (0,0) [label = below:{433 344}] {2};

    \node [circle, fill = white] at (1,2) [label = above:{140 201}] {\ch02};
    \node [circle, fill = white] at (1,1) [label = 90:{201 312}] {\ch12};
    \node [circle] at (1,0) [label = below:{342 403}] {\ch01};

    \node [circle] at (2,2) [label = above:{234 340}] {\ch10};
    \node [circle] at (2,1) [label = 90:{340 401}] {\ch21};
    \node [circle] at (2,0) [label = below:{401 012}] {\ch20};

    \node [circle] at (3,2) [label = above:{343 434}] {2};
    \node [circle] at (3,1) [label = above:{404 040}] {0};
    \node [circle] at (3,0) [label = below:{010 101}] {1};
  \end{tikzpicture}
    \caption{No Olšák polymorphism for $(\rel C_5,\rel K_3)$.} 
    \label{fig:no-Olsak-C5-K3}
\end{figure}

It is easy to check that the graph in Fig.~\ref{fig:no-Olsak-C5-K3} is a subgraph of $\rel C_5^6$, with each node assigned a set of colours, i.e., elements of $\rel K_3$. Unique colours (in columns 1 and 4) correspond to the values of the assumed operation $o$ on these tuples; note that these colours are correct because $o$ is assumed to be an Olšák function.
The two-element lists of colours simply indicate the possible values for the corresponding tuples on the basis that each such tuple has a neighbour whose colour is known. Now, the two middle vertices in the middle row must be assigned different colours by $o$, but neither option extends to a full proper colouring of this graph.
\end{proof}

\begin{theorem}
  $\PCSP(\rel C_5,\rel K_3)$ is \NP-hard. \qed
\end{theorem}

\subsection{Implications of hardness of approximate graph colouring}
  \label{sec:promise-graph-colouring}

We proved above that the hardness of approximate hypergraph colouring implies hardness of any PCSPs satisfying an algebraic condition (namely, no Olšák polymorphism). In fact, this is one case of a more general pattern for such implications. To illustrate this, we now describe similar implications for the (yet unproved) hardness of approximate graph colouring and graph homomorphism problems.
Recall Examples~\ref{ex:approx-graph-col} and~\ref{ex:approx-graph-hom} and the conjectures mentioned there.

For the case of approximate graph colouring, the so-called Siggers  functions play a role similar to that played by Olšák functions in relation to hypergraphs. The original Siggers functions appeared in~\cite{Sig10}, but now there are several related versions of such functions (which are usually all called Siggers functions). We will use the one using 3-variable identities which appeared, e.g.\ in~\cite{BKOPP17}. 

\begin{definition}
\emph A {Siggers function} is a~6-ary function $s$ that satisfies
\[
  s(x,y,x,z,y,z) \equals s(y,x,z,x,z,y).
\]
\end{definition}

Just like the definition of an Olšák function relates to the six tuples in $\rel H_2$, the definition of a~Siggers polymorphism relates to the six edges of~$\rel K_3$ (viewed as a~directed graph).

We remark that, for problems $\CSP(\rel A)$ with finite $\rel A$, the different versions of a~Siggers function are equivalent in the sense that if $\Pol(\rel A)$ has one of them then it has all of them.  Moreover, for CSPs, this is also equivalent to having an Olšák polymorphism~\cite{Ols17} and this property exactly characterises tractable CSPs \cite{Bul17,Zhu17} (under $\Ptime\neq\NP$). However, it can be shown that such equivalences do not hold for PCSPs.

\begin{theorem} \label{thm:siggers}
  The following are equivalent.
  \begin{enumerate}
    \item For every finite template $(\rel A,\rel B)$ without a~Siggers polymorphism, $\PCSP(\rel A,\rel B)$ is \NP-hard.
    \item $\PCSP(\rel K_k,\rel K_c)$ is \NP-hard for all $c\geq k\geq 3$.
    \item $\PCSP(\rel K_3,\rel K_c)$ is \NP-hard for each $c\geq 3$.
  \end{enumerate}
\end{theorem}

\begin{proof}
  To prove (1)\textto(2), it is enough to show that $(\rel K_k,\rel K_c)$ has no Siggers polymorphisms for any $k\geq 3$. That is true, since any such polymorphism would force a~loop in $\rel K_c$ on the vertex $s(0,1,0,2,1,2) = s(1,0,2,0,2,1)$.

  The implication (2)\textto(3) is trivial.

  For (3)${}\to{}$(1), the proof is similar to that of Theorem~\ref{thm:no-olsak}.
  Assume that a template $(\rel A,\rel B)$ has no Siggers polymorphism, and let $\clo M = \Pol(\rel A,\rel B)$.  Consider the free structure $\rel F=\rel F_{\clo M}(\rel K_3)$.  Then $\rel F$ is a graph whose vertices are the ternary polymorphisms of $(\rel A,\rel B)$, and $(f,g)$ is an edge if there is a~6-ary $e\in \clo M$ such that
  \begin{align*}
    f(x,y,z) &= e(x,y,x,z,y,z) \\
    g(x,y,z) &= e(y,x,z,x,z,y)
  \end{align*}
  for all $x,y,z \in A$. Clearly, this graph is loopless, since a~loop would correspond to a~Siggers polymorphism. It is then $c$-colourable where $c=|F|$. So we have a homomorphism from $\rel F$ to $\rel K_c$, which by Lemma~\ref{lem:adjunction} implies that there is a minion homomorphism from $\clo M$ to $\Pol(\rel K_3,\rel K_c)$.  The result now follows from Theorem~\ref{thm:main}.
\end{proof}

We now generalise the above to the case of approximate graph homomorphism.

\begin{definition}[\cite{Ols17a}] \label{def:loop-condition}
  Fix a~loopless graph $\rel G$, with vertices $v_1,\dots,v_n$ and edges $(a_1,b_1),\dots,(a_m,b_m)$ where each edge $(u,v)$ is listed as both $(u,v)$ and $(v,u)$. The \emph{$\rel G$-loop condition} is the bipartite minor condition:
  \begin{align*}
    f(x_{v_1},\dots,x_{v_n}) &\equals e(x_{a_1},\dots,x_{a_m}) \\
    f(x_{v_1},\dots,x_{v_n}) &\equals e(x_{b_1},\dots,x_{b_m}).
  \end{align*}
\end{definition}

We remark that the $\rel G$-loop condition can be also constructed as $\Sigma(\rel G,\rel L)$ (recall Section~\ref{sec:two-prover-protocol}) where the graph $\rel L$ is `the loop', i.e., the graph with a~single vertex with a~loop.

\begin{example}
 Consider the $\rel K_3$-loop condition with edges of $\rel K_3$ listed as $(0,1)$, $(1,0)$, $(0,2)$, $(2,0)$, $(1,2)$, $(2,1)$.
Then the functions $e$ satisfying this condition (with some $f$) are exactly the Siggers functions.
\end{example}

\begin{theorem}
  The following are equivalent:
  \begin{enumerate}
    \item $\PCSP(\rel C_{2k+1},\rel K_c)$ is \NP-hard for all $c\geq 3$, $k \geq 1$.
    \item $\PCSP(\rel G,\rel H)$ is \NP-hard for any two non-bipartite loopless graphs $\rel G$, $\rel H$ with $\rel G\rightarrow \rel H$.
    \item For any finite template $(\rel A,\rel B)$ that does not satisfy the $\rel G$-loop condition for some non-bipartite loopless graph $\rel G$, $\PCSP(\rel A,\rel B)$ is \NP-hard.
  \end{enumerate}
\end{theorem}

\begin{proof}
  (2)\textto(3) is proven similarly to Theorem~\ref{thm:siggers}, and (3)\textto(1) is given by the fact that $(\rel C_{2k+1},\rel K_c)$ does not satisfy the $\rel C_{2k+1}$-loop condition.

  (1)\textto(2).  Assume that $\PCSP(\rel C_{2k+1},\rel K_c)$ is \NP-hard, and let $\rel G$, $\rel H$ be loopless non-bipartite graphs. Since $\rel G$ is not bipartite it contains an odd cycle, i.e., $\rel C_{2k+1} \to \rel G$ for some $k$, and since $\rel H$ is loopless, it is colourable by some finite number of colours, i.e., $\rel H\to\rel K_c$ for some $c$. Therefore, $(\rel C_{2k+1},\rel K_c)$ is a~relaxation of $(\rel G,\rel H)$, and the claim follows, e.g. by Lemma~\ref{lem:pp-gives-minor} and Theorem~\ref{thm:main}.  
\end{proof}
\section{Tractability from some CSPs}\label{sec:tractable}

In the previous section, we described some limitations on bipartite minor conditions satisfied in the polymorphism minion (such as the absence of an Olšák function) that imply hardness of the corresponding PCSP. The goal of this section is to describe some minor conditions that imply tractability of the corresponding PCSPs. Observe first that, after excluding the trivial cases that imply a~constant operation (the condition $f(x)\equals f(y)$), no other single bipartite minor condition can imply tractability, since any such condition involves functions of bounded arity and hence is satisfied in the minion of functions of some bounded essential arity whose corresponding PCSP is NP-hard (see Proposition~\ref{prop:bounded-arity}). Therefore, the conditions implying tractability always involve an infinite set of bipartite minor conditions.

We do not provide any new concrete tractability results here, our contribution is an algebraic characterisation of the power of several known algorithms. 
Each of the conditions that we present describes applicability of a~certain algorithm. Some of these algorithms are based on known algorithms for CSPs, e.g.~\cite{DP99,KOTYZ12}, and some on recent tractability results for PCSPs~\cite{BG18b}. In general, these algorithms are obtained from algorithms for a~fixed tractable CSP with template $\rel D$ (possibly with infinite domain), and we apply them to all PCSP templates that are pp-constructible from $\rel D$ (recall Definition~\ref{def:pp-constructible}). This can be viewed as a~(slight) generalisation of the `homomorphic sandwiching' method described by Brakensiek and Guruswami \cite{BG18b}.
We show in the next section that there exist a tractable finite fixed-template PCSP whose tractability cannot be explained by pp-constructibility from (or sandwiching) a~finite CSP template.

We consider three algorithms. The first one is based on a special case of local consistency, the second on the basic LP relaxation, and the last one on an affine integer relaxation. The last two types are closely related to the algorithms considered in \cite{BG18b}, and we refer to that paper for the latest algorithmic results for PCSPs that are based on relaxations over various numerical domains.
We remark that the complexity of (infinite-domain) CSP with numerical domains is receiving a good amount of attention, see~\cite{BM17}, but the connection with PCSPs has not been well-studied yet.

\subsection{Local consistency}
  \label{sec:local-consistency}

A~large family of algorithms for solving CSPs is based on some type of local propagation. Such an algorithm runs a propagation procedure which either refutes a given instance or transforms it to a locally consistent instance, without changing the set of solutions.
This algorithm is said to solve $\CSP(\rel A)$ if every locally consistent instance has a~solution. Such algorithms naturally generalise to $\PCSP(\rel A,\rel B)$: the algorithm interprets the input $\rel I$ as an instance of $\CSP(\rel A)$ and runs the appropriate consistency checking. A~negative answer means that there is no homomorphism from $\rel I$ to $\rel A$, and, for correctness, we only require that every instance that is consistent as an instance of~$\CSP(\rel A)$ admits a homomorphism to $\rel B$.
The most common type of such algorithms, the \emph{bounded width} algorithm~\cite{BK14}, works by inferring as much information about a solution as possible from considering fixed-size subsets of the instance, one at a time. 

\begin{definition}
For a structure $\rel I$ and a subset $X\subseteq I$, let $\rel I[X]$ denote the structure \emph{induced by} $X$ in $\rel I$ --- its domain is $X$, and each relation $R_i^{\rel I}$ is replaced by $R_i^{\rel I}\cap X^{\ar(R_i)}$. A \emph{partial homomorphism} from $\rel I$ to $\rel A$ with domain $X$ is any homo\-mor\-phism $\rel I[X]\rightarrow \rel A$.
\end{definition}

Given an instance $\rel I$ of $\CSP(\rel A)$ and $k\leq l$, the \emph{$(k,l)$-consistency algorithm} constructs the largest family $\mathcal F$ of partial homomorphisms from $\rel I$ to $\rel A$ with at most $l$-element domains satisfying the following two conditions:
\begin{itemize}
\item for any $f\in\mathcal F$, all restrictions of $f$ to smaller domains are also in $\mathcal F$, and
\item for any $f\in\mathcal F$ with at most $k$-element domain, there is an extension $g\in\mathcal F$ of $f$ to any $l$-element domain containing the domain of $f$.
\end{itemize}
On input $\rel I$, the algorithm starts with the set of all mappings from at most $l$-element subsets of $I$ to $A$ and repeatedly removes, until stable, mappings from $\mathcal F$ which are not partial homomorphisms or violate (at least) one of the two conditions above. 
We say that an instance $\rel I$ is a~\emph{$(k,l)$-consistent instance} of $\CSP(\rel A)$ if the output of the $(k,l$)-consistency algorithm is a~non-empty family of partial homomorphisms. 

\begin{definition}\label{def:bounded-width}
  A PCSP template $(\rel A, \rel B)$ has \emph{width $(k,l)$} if every instance $\rel I$ which is $(k,l)$-consistent as an instance of $\CSP(\rel A)$ maps homomorphically to $\rel B$. A~template has \emph{bounded width} if it has width $(k,l)$ for some $k\leq l$, and \emph{width~1} if it has width $(1,l)$ for some $l$. 
\end{definition}

Note that the $(k,l)$-consistency can be tested efficiently which provides a~polynomial time algorithm for PCSPs of bounded width. CSP templates of bounded width have been fully characterised by Barto and Kozik in~\cite{BK14}. For PCSPs, we can characterise width~1 which corresponds to solvability by (generalised) arc consistency, one of the most prevalent algorithms used in constraint programming. 

\begin{definition}
  A function is \emph{totally symmetric} if its output depends only on the set of input elements.
\end{definition}

Note that being totally symmetric can be described as a~bipartite minor condition: a~function $g$ of arity $n$ is totally symmetric if there are functions $f_1$, \dots, $f_n$ where each $f_i$ has arity $i$ such that for all $i=1,\dots,n$ and all surjective $\pi\colon [n] \to [i]$ we have
\(
  f_i (x_1,\dots,x_i) \equals
  g(x_{\pi(1)},\dots,x_{\pi(n)})
\).

The following theorem generalises the characterisation of width 1 for CSP templates from \cite{FV98,DP99}. We add descriptions using minion homomorphisms and pp-constructibility. Recall the template $\rel H$ (of \textsc{Horn 3-Sat}) as defined in Example~\ref{ex:C(A)} and let $\clo H = \Pol(\rel H)$.

\begin{theorem} \label{thm:width1-characterization}
Let $(\rel A,\rel B)$ be a PCSP template. The following are equivalent.
\begin{enumerate}
  \item\label{it:w1c1} $(\rel A,\rel B)$ has width 1.
  \item\label{it:w1c2} $\Pol(\rel A,\rel B)$ contains totally symmetric functions of all arities.
  \item\label{it:w1c3} There exists a~minion homomorphism from $\clo H$ to $\Pol(\rel A,\rel B)$.
  \item\label{it:w1c4} $(\rel A,\rel B)$ is pp-constructible from $\rel H$.
\end{enumerate}
\end{theorem}

\begin{proof}
Let $\rel F=\rel F_{\clo H}(\rel A)$ be the free structure of $\clo H$ generated by $\rel A$ --- see Example~\ref{ex:C(A)} for detailed information about this structure. 
By Theorem~\ref{thm:minor-homomorphism-is-pp-constructibility}, items (3) and (4) are equivalent between themselves, and also equivalent to the condition $\rel F\rightarrow \rel B$. This last condition can be shown to be equivalent to both (1) and (2) essentially in the same way as the corresponding result for CSP \cite{DP99} (where notation $\mathcal C(\rel A)$ is used for the free structure).
\end{proof}

Let us move towards the general bounded width algorithm. The following lemma shows that the class of PCSP templates of bounded width is closed under minion homomorphisms. The proof builds on an~analogous result for CSPs obtained by Larose and Z\'adori \cite{LZ07}. To establish the hardness part of their then-conjectured characterisation of CSP templates of bounded width, they proved that bounded width is preserved under so-called `pp-interpretations' (pp-power is a~special case of pp-interpretation) of CSP templates. The generalisation to PCSPs suggests that it may be possible to characterise bounded width for PCSPs by minor conditions as well, although we are not able to formulate such a~conjecture yet.

\begin{lemma} \label{lemma:minor-preserves-bw}
Let $(\rel A, \rel B)$ and $(\rel A',\rel B')$ be templates such that there exists a minion homomorphism from $\Pol(\rel A, \rel B)$ to $\Pol(\rel A',\rel B')$. If $(\rel A, \rel B)$ has bounded width, then so does $(\rel A',\rel B')$.
\end{lemma}

\begin{proof}
Using Theorem \ref{thm:minor-homomorphism-is-pp-constructibility}, it is enough to show that bounded width is preserved under homomorphic relaxations and pp-powers. We present a complete proof for relaxations, and sketch a~proof for pp-powers, since the latter proof follows \cite{LZ07}.

Suppose that $(\rel A',\rel B')$ is a~homomorphic relaxation of a~template $(\rel A,\rel B)$, i.e., there exist homomorphisms $h_A\colon\rel A'\to\rel A$ and $h_B\colon\rel B\to\rel B'$, and that $(\rel A,\rel B)$ has width~$(k,l)$. We will prove that in this case, $(\rel A',\rel B')$ has width $(k,l)$ as well. 
Let $\rel I'$ be an instance of $\CSP(\rel A')$ which is $(k,l)$-consistent and let this fact be witnessed by a~nonempty family of partial homomorphisms $\mathcal F'$. It is easy to see that the set $\{h_Af'\mid f'\in\mathcal F'\}$ witnesses that $\rel I'$ is $(k,l)$-consistent as an instance of $\CSP(\rel A)$, and thus there exists a~homomorphism $s\colon\rel I'\to \rel B$. Consequently, $h_Bs: \rel I' \to \rel B'$.

We only sketch a proof for pp-powers: Assume that $(\rel A',\rel B')$ is a~pp-power of $(\rel A,\rel B)$, and the latter has width $(k,l)$. We claim that $(\rel A',\rel B')$ has width $(kM,lM)$ where $M$ is the maximal arity of a~relation of $(\rel A',\rel B')$. The key point is to observe that if we start with an instance $\rel I'$ of $\PCSP(\rel A',\rel B')$ that is $(kM,lM)$-consistent, and replace every constraint of $\rel I'$ by its pp-definition in $(\rel A,\rel B)$, we obtain an instance $\rel I$ of $\PCSP(\rel A,\rel B)$ that is $(k,l)$-consistent. Note that this construction follows the standard reduction from $\CSP(\rel A')$ to $\CSP(\rel A)$, therefore the arguments of \cite[Lemma 3.3]{LZ07} apply. Since $\rel I$ is~$(k,l)$-consistent and $(\rel A,\rel B)$ has width $(k,l)$, we get $\rel I \to \rel B$. This homomorphism witnesses that $\rel I'\to\rel B'$ since $\rel B'$ is defined from  $\rel B$ in the same way as $\rel A'$ from $\rel A$.
\end{proof}

\subsection{Linear programming relaxations}

Every CSP instance can be expressed as a 0-1 integer program in a canonical way. When we allow the variables in this program to attain any values from $[0,1]$ we obtain the so-called basic linear programming relaxation~\cite{KOTYZ12}.

\begin{definition}
Given an instance $\rel I$ of $\CSP(\rel A)$, let $\mathcal C = \{ (\tup v,R) \mid \tup v \in R^{\rel I} \}$. The \emph{basic linear programming relaxation} of $\rel I$ is the following linear program: The variables are $\mu_v(a)$ for every $v\in I$ and $a\in A$, and $\mu_{\tup v,R}(\tup a)$ for every $(\tup v,R)\in \mathcal C$, and every $\mathbf a\in A^{\ar(R)}$. Each of the variables is allowed to have values in the interval $[0,1]$.
The objective is to maximise 
\[
  \frac1{|\mathcal C|} \sum_{(\tup v,R)\in \mathcal C} 
    \sum_{\mathbf a\in R^{\rel A}} \mu_{\tup v,R}(\tup a)
\]
subject to:
\begin{align}
  \sum_{a\in A} \mu_v(a) &= 1 & v\in I, \label{eq:blp1} \\
  \sum_{\tup a \in A^{\ar(R)}, \tup a(i) = a} \mu_{\tup v,R}(\mathbf a) &= \mu_{\tup v(i)}(a) & a\in A, (\tup v,R)\in \mathcal C, i \in [\ar(R)]. \label{eq:blp2}
\end{align}
We denote the maximum possible value of the objective function by $\BLP_{\rel A}(\rel I)$.
\end{definition}

It is clear that the optimum value, $\BLP_{\rel A}(\rel I)$, of this LP is smaller than or equal to $1$, since (\ref{eq:blp1}) and (\ref{eq:blp2}) together imply that $\sum_{\tup a\in A^{\ar(R)}} \mu_{\tup v, R} (\tup a) = 1$, and therefore $\sum_{\tup a\in R} \mu_{\tup v,R} (\tup a) \leq 1$ for each of the constraints $(\tup v, R)$.
Given that the instance $\rel I$ has a~solution as an instance of $\rel A$, say a~homomorphism $s\colon \rel I\to \rel A$, there is an integral solution to the above linear program that achieves this optimum value: $\mu_v(s(v)) = 1$, $\mu_v(a) = 0$ for all $a\neq s(v)$, $\mu_{(v_1,\dots,v_k),R}(s(v_1),\dots,s(v_k)) = 1$, and $\mu_{\tup v,R}(\tup a) = 0$ for all other $\tup a$'s.

\begin{definition} \label{def:solvable-by-BLP}
Let $(\rel A,\rel B)$ be a PCSP template. We say that \emph{\textsc{BLP} solves $\PCSP(\rel A,\rel B)$} if every instance $\rel I$ with $\BLP_\rel A(\rel I)=1$ maps homomorphically to $\rel B$.
\end{definition}

It is easy to see that $\BLP_\rel A(\rel I)=1$ if an only if $\mu_{\tup v,R} (\tup a) = 0$ for each constraint $(\tup v, R)$ in $\rel I$ and $a\notin R^{\rel A}$. Therefore, if we add all such constraints to $\BLP_\rel A(\rel I)$, testing the feasibility of the obtained LP is equivalent to testing whether $\BLP_\rel A(\rel I)=1$.
Therefore, if BLP solves $\PCSP(\rel A,\rel B)$, then $\PCSP(\rel A,\rel B)$ reduces to an LP feasibility problem, where each LP constraint (except the non-negativity inequalities) is a~linear equation with $\pm 1$ coefficients and a~bounded number of variables.  Such a problem is expressible as $\CSP(\rel D)$ for an appropriate structure $\rel D$ with domain $\mathbb Q$ and finitely many relations. Note that the definition of a~CSP extends in a~straightforward way to such structures.
The template $\rel D$ is obtained from the structure $\rel Q_\conv$, whose domain is $\mathbb Q$ and whose (infinitely many) relations are all possible linear inequalities with rational coefficients (see \cite[Definition 4]{BM17}), by simply dropping all but finitely many relations. This is often expressed by saying that $\rel D$ is a~finite reduct of $\rel Q_\conv$.
Note that the relations pp-definable in $\rel Q_\conv$ are the convex polytopes in $\mathbb Q^k$, $k\ge 1$. It is easy to see that the polymorphisms of $\rel Q_\conv$ are exactly convex linear functions, i.e., functions $f\colon \mathbb Q^n \to \mathbb Q$ defined by
\(
  f(x_1,\dots,x_n) = \sum_{i\in[n]} \alpha_i x_i
\)
for some $\alpha_i$'s, $i\in [n]$, such that $\alpha_i \in [0,1]$, $\sum_{i\in[n]} \alpha_i = 1$. We denote the set of all such operations by $\clo Q_\conv$.

\begin{definition}
We say that a function is \emph{symmetric} if the output is independent of the order of the input elements.
\end{definition}

While total symmetry means that the output is dependent only on the set of the input elements, symmetry can be formulated as `the output is dependent only on  the multiset of the input elements'.
It can be also expressed as a~minor condition; a~function $f$ of arity $n$ is symmetric if 
\(
  f(x_1,\dots,x_n) \equals f(x_{\pi(1)},\dots,x_{\pi(n)})
\)
for all bijections $\pi\colon [n] \to [n]$.

Note that $\clo Q_\conv$ contains symmetric operations of all arities: For arity $n$, simply take $f(x_1,\dots,x_n)=\sum_{i\in [n]} x_i/n$.

We provide a~characterisation of the direct applicability of the basic linear programming relaxation that generalises \cite[Theorem 2(5) \& (6)]{KOTYZ12}. Note that the claim of \cite{KOTYZ12} that, for CSPs, solvability by BLP is equivalent to having width 1 (i.e., to items (1)--(3) of their Theorem 2) is false --- see \cite[Example~99]{KS16}.

\begin{theorem} \label{thm:lp-characterization}
Let $(\rel A,\rel B)$ be a PCSP template. The following are equivalent.
\begin{enumerate}
  \item \label{it1:blpc} \textsc{BLP} solves $\PCSP(\rel A,\rel B)$,
  \item \label{it2:blpc} $\Pol(\rel A,\rel B)$ contains symmetric functions of all arities.
  \item \label{it3:blpc} $\Pol(\rel A,\rel B)$ admits a minion homomorphism from $\clo Q_\conv$,
  \item \label{it4:blpc} $(\rel A,\rel B)$ is pp-constructible from a finite reduct of $\rel Q_\conv$. 
\end{enumerate}
\end{theorem}

The structure of the proof is the same as for Theorem~\ref{thm:width1-characterization}. Although, there are several important issues: mostly that Theorem~\ref{thm:minor-homomorphism-is-pp-constructibility} holds for finite structures, but in general it does not hold for infinite structures. Nevertheless, we will show that it holds for the structure $\rel Q_\conv$.
The proof uses an appropriate modification of the notion of a~free structure for $\clo Q_\conv$, which we define next. This modification is similar to the instance $\mathcal M(\Gamma)$ defined in \cite[Definitions 10 and 11]{KOTYZ12}. To keep some consistency in our notation, we will denote this structure by $\LP(\rel A)$.

\begin{definition} \label{def:free-lp}
The structures $\LP(\rel A)$ and $\rel A$ are similar. The universe $\LP(A)$ consists of rational probability distributions on $\rel A$, i.e., functions $\phi\colon \rel A\to \mathbb Q \cap [0,1]$ such that $\sum_{a\in A} \phi(a) = 1$. 
For a~$k$-ary relation $R^{\rel A}$, the corresponding relation $R^{\rel \LP(\rel A)}$ is defined as the set of all $k$-tuples $(\phi_1,\dots,\phi_k)$ of elements of $\LP(A)$ for which there exists a~rational probability distribution $\gamma$ on $R^{\rel A}$ such that
\begin{equation} \tag{$\vardiamondsuit$} \label{eq:lp-free}
  \sum_{\tup a\in R^{\rel A}, \tup a(i) = a} \gamma(\tup a) = \phi_i(a)
\end{equation}
for all $i\in [k]$ and $a\in A$.
\end{definition}

Note that the above structure relates closely to the BLP relaxation: It is easy to observe that $\rel I$ is an~instance such that $\BLP_\rel A(\rel I)=1$ if and only if $\rel I$ maps homomorphically to $\LP(\rel A)$. This has also been observed in \cite{KOTYZ12}.

\begin{remark} \label{rem:lp-is-free}
The structure $\LP(\rel A)$ is isomorphic to the free structure $\rel F_{\clo Q}(\rel A)$ of $\clo Q_\conv$ generated by $\rel A$. For a proof, assume that $A = [n]$. Recall that the elements of the free structure $\rel F_{\clo Q}(\rel A)$ are $n$-ary convex linear functions. An~isomorphism $h\colon \LP(\rel A)\to \rel F_{\clo Q}(\rel A)$ is given by
\(
  h(\phi)= f_\phi
\)
where $f_\phi$ is defined by $f_\phi(x_1,\dots,x_n) = \sum_{i\in [n]} \phi(i)x_i$. Note that $f_\phi$ is always a~convex linear function, therefore an element of $F_{\clo Q}(A)$, and also that $h$ is bijective since every such function is uniquely determined by its coefficients.
To prove that $h$ preserves a~relation $R$, assume $(\phi_1,\dots,\phi_k) \in R^{\LP(\rel A)}$, and let $\gamma$ be the probability distribution on $R^{\rel A}$ witnessing this fact. Further, let $R^{\rel A} = \{ \tup r_1,\dots,\tup r_m \}$, and define $g\colon \mathbb Q^m \to \mathbb Q$ by
\[
  g(x_1,\dots,x_m) = \sum_{i\in [m]} \gamma(\tup r_i) x_i.
\]
We claim that for each $j \in [k]$, we have
\[
  g( x_{\tup r_1(j)},\dots,x_{\tup r_m(j)} ) = h(\phi_j)( x_1,\dots,x_n )
\]
which is easily observed by comparing coefficients of the left- and right-hand side.
The fact that $h^{-1}$ is also a~homomorphism is obtained by reversing this argument, equating a~convex linear function $g\colon \mathbb Q^m \to \mathbb Q$ with the probability distribution $\gamma$ which maps a~tuple $\tup r_i$ to the $i$-th coefficient of $g$.
\end{remark}

In the light of the previous remark, the following can be understood as an~infinite case of Lemma~\ref{lem:free-is-pp}. 

\begin{lemma} \label{lem:lp-free-is-pp}
  Let $\rel A$ be a~finite relational structure, and let $\LP(\rel A)$ be the free structure of $\clo Q_\conv$ generated by $\rel A$. Then $(\rel A,\LP(\rel A))$ is a~relaxation of a~pp-power of $\rel Q_\conv$.
\end{lemma}

\begin{proof} 
In this proof, we assume $A = [n]$ and equate a~probability distribution $\phi$ on $A$ with the tuple $(\phi(1),\dots,\phi(n))$.
Let us define an~$n$-th pp-power $\rel P$ of $\rel Q_\conv$, and its relaxation that will be isomorphic to $(\rel A,\LP(\rel A))$:
A~relation $R^{\rel P}$ is defined to contain all tuples $(\phi_1,\dots,\phi_k)$ of $n$-tuples such that $\phi_j \geq 0$ and $\sum_{i\in [n]} \phi_j(i) = 1$ for all $j$, and so that there exists $\gamma\colon R^{\rel A} \to \mathbf Q$ such that (\ref{eq:lp-free}) is satisfied. This is indeed a~pp-definition since each of the inequalities and identities define a~relation of $\rel Q_\conv$.
Finally, we argue as in Lemma~\ref{lem:free-is-pp}: $\rel A$ maps homomorphically into $\rel P$ by $a\mapsto \chi_a$ where $\chi_a(b) = 1$ for $b = a$ and $\chi_a(b) = 0$ for $b\neq a$, and $\rel P$ maps homomorphically to $\LP(\rel A)$ by $\phi\mapsto f_\phi$ whenever $\sum_{a\in A} \phi(a) = 1$ and $\phi(a) \geq 0$, and extending arbitrarily.
\end{proof}

\begin{remark} \label{rem:compactness}
In the proof below, we will need to find a~homomorphism from a~relational structure with an infinite universe, let us for now call it $\rel I$, to a~similar finite relational structure $\rel B$. It is well-known that in that case it is enough to find a~homomorphism from all finite substructures of $\rel I$. This fact is usually proven by a~standard compactness argument, e.g.\ using Tychonoff's theorem. For completeness, we present one such argument that uses König's lemma. This approach works only for countable structures $\rel I$ which is enough in our case.

Assuming that $I = \{1,2,\dots\}$ and that every finite substructure of $\rel I$ maps to $\rel B$, we construct an infinite, finitely branching tree: The nodes of the tree are partial homomorphisms defined on sets $[n]$ (starting with the empty set, so that the empty mapping is the root of the tree). We set that a~partial homomorphism $r_{n+1} \colon [n+1] \to \rel B$ is a~child of $r_n\colon [n] \to \rel B$ if the map $r_n$ is the restriction of $r_{n+1}$ to $[n]$. Clearly, this tree is finitely branching and infinite. König's lemma states that such a tree has an infinite branch, which in our case gives us a~sequence $r_1, r_2,\dots$ such that $r_{n+1}$ extends $r_n$. We define $r\colon \rel I\to \rel B$ as the union of these maps. It is a~homomorphism since all constraints are local, and therefore included in the domain of some $r_n$ which is a~partial homomorphism.
\end{remark}

\begin{proof}[Proof of Theorem~\ref{thm:lp-characterization}]
  First, we show that item (1) is equivalent to the existence of a~homomorphism from $\LP(\rel A)$ to $\rel B$. The direct implication follows by the same argument as \cite[Proposition 12]{KOTYZ12}. If BLP solves $\CSP(\rel A)$ and $\rel I$ is a~finite subinstance of $\LP(\rel A)$ then it is easy to see that $\BLP_{\rel A}(\rel I)=1$, and therefore $\rel I\to \rel B$. Using a~standard compactness argument (see the previous remark), this implies that $\LP(\rel A)$ homomorphically maps to $\rel B$. The converse is straightforward: indeed any instance $\rel I$ with $\BLP_\rel A(\rel I) = 1$ maps homomorphically to $\LP(\rel A)$, and therefore also $\rel I \to \rel B$ because $\LP(\rel A) \to \rel B$.

  Given a~homomorphism from $\LP(\rel A)$ to $\rel B$, we get item (4), i.e., that $(\rel A,\rel B)$ is pp-constructible from $\rel Q_\conv$, from Lemma~\ref{lem:lp-free-is-pp}. Further, (4)\textto(3) follows from Lemma~\ref{lem:pp-gives-minor} (note that the proof does not require the structures to be finite). We have (3)\textto(2) since $\clo Q_\conv$ contains symmetric operations of all arities and any minion homomorphism preserves this property. (Note that (1)\textto(2) can be also obtained by argument similar to \cite[Proposition 12]{KOTYZ12}).
  
  Let us prove that (2) implies that $\LP(\rel A)\to \rel B$, and hence item (1). To do that, we define homomorphisms from certain finite substructures of the structure $\LP(\rel A)$:
  We define $\LP_\ell(\rel A)$ to be a~structure similar to $\rel A$ whose universe consists of rational probability distributions on $\rel A$ with denominators dividing $\ell$, i.e., functions $\phi\colon A \to \mathbb Q$ where, for each $a\in A$, $\phi(a) = q/\ell$ for some $q\in \{0,1,\dots,\ell\}$ and $\sum_{a\in A}\phi(a) = 1$. The relations are defined the same way as in $\LP(\rel A)$ where we restrict to probability distributions $\gamma$ with denominators dividing $\ell$. (Note that, unlike instances $\mathcal M_\ell(\Gamma)$ defined in \cite[Definition 10]{KOTYZ12}, the structure $\LP_\ell(\rel A)$ is not an induced substructure of $\LP(\rel A)$.)
 
  We define a~homomorphism $h_\ell\colon \LP_\ell(\rel A)\to \rel B$ by fixing a~symmetric function $s_\ell\in \Pol(\rel A,\rel B)$ of arity $\ell$, and setting
  \[
    h_\ell(\phi) = s_\ell( a_1,\dots,a_\ell )
  \]
  where $a_1,\dots,a_\ell \in A$ are chosen so that each $a$ appears exactly $\phi(a)\ell$ times. Note that the order of $a_i$'s does not matter, since $s_\ell$ is symmetric.
  To show that $h_\ell$ is a~homomorphism, consider a~tuple $(\phi_1,\dots,\phi_k) \in R^{\LP_\ell{\rel A}}$ and let $\gamma\colon R^{\rel A} \to \mathbb Q$ be the witnessing probability distribution. We pick tuples $\tup r_1,\dots,\tup r_\ell\in R^{\rel A}$ such that each $\tup r\in R^{\rel A}$ appears exactly $\gamma(\tup r)\ell$ times. Now,
  \[
  (h_\ell(\phi_1),\dots,h_\ell(\phi_k))
  = s_\ell( \tup r_1,\dots,\tup r_\ell )
  \in R^{\rel B}
  \]
  where the equality follows since, for each $i\in [k]$, each $a\in A$ appears exactly $\phi_i(a)\ell = \sum_{r\in R^{\rel A}, \tup r(i) = a} \gamma(\tup r)\ell$ times among $\tup r_1(i),\dots,\tup r_\ell(i)$ (again we use symmetry of $s_\ell$).
  This establishes that every $\LP_\ell(\rel A)$ maps homomorphically to $\rel B$. Note that every finite substructure of $\LP(\rel A)$ is a~substructure of $\LP_\ell(\rel A)$ for a big enough $\ell$, and 
  hence every finite substructure of $\LP(\rel A)$ maps homomorphically to $\rel B$. Thus we have a~homomorphism from $\LP(\rel A)$ to $\rel B$ by a~standard compactness argument.
\end{proof}

\subsection{Affine Diophantine relaxations}

Another way to relax the natural 0-1 integer program expressing a~CSP instance is to allow the variables to attain any integer values, and relax constraints to linear equations. We get a~relaxation with very similar properties to the basic LP relaxation described in the previous section.

\begin{definition}
  Given an instance $\rel I$ of $\CSP(\rel A)$, let $\mathcal C = \{(\tup v, R) \mid v\in R^{\rel I} \}$. The \emph{basic affine integer relaxation} of $\rel I$, denoted $\AIP_{\rel A}(\rel I)$, is the following affine program.
  The variables are $\mu_v(a)$ for every $v\in I$ and $a\in A$, and $\mu_{\tup v,R}(\tup a)$ for every $(\tup v, R) \in \mathcal C$ and $\tup a \in R^{\rel A}$. The objective is to solve the following system over $\mathbb Z$:
  \begin{align}
    \sum_{a\in A} \mu_v(a) &= 1 & v\in I, \label{eq:lip1} \\
    \sum_{\tup a \in R^{\rel A}, \tup a(i) = a} \mu_{\tup v,R}(\mathbf a) &= \mu_{\tup v(i)}(a) & a\in A, (\tup v,R)\in \mathcal C, i \in [\ar(R)]. \label{eq:lip2}
  \end{align}
\end{definition}

Note that this system is an instance of the infinite CSP with template $\rel Z_\aff$ with domain $\mathbb Z$ whose relations are all affine equations over $\mathbb Z$. Let $\clo Z_\aff = \Pol(\rel Z_\aff)$.  
We claim that $\clo Z_\aff$ consists of all affine functions over $\mathbb Z$, i.e., all functions $g\colon \mathbb Z^n \to \mathbb Z$ described as $g(x_1,\dots,x_n) = \sum_{i\in [n]} \gamma(i) x_i$ where $\gamma$ is such that $\sum_{i\in [n]} \gamma(i) = 1$. Clearly, any such function is a~polymorphism. For the other inclusion assume that $f\colon \mathbb Z^n \to \mathbb Z$ is a~polymorphism. Then in particular, it preserves the relation $\{ (x,y,z) \mid x+y = z \}$ which implies that it is linear. Also, it preserves the unary singleton relation $\{1\}$ (the solution to the equation $x=1$), which implies that the sum of its coefficients is 1.

This situation is similar to the one with convex linear functions in the previous subsection. And in fact, using similar methods, we obtain similar results. Let us first show an example of the use of this relaxation.

\begin{example}
  \label{example:1in3-nae2} \label{ex:1in3-nae-aip}
Let us describe the $\AIP$ relaxation of 1-in-3- vs.\ Not-All-Equal-\Sat, and compare it with an algorithm for this PCSP described in \cite{BG18}. Recall that the PCSP template of this problem is $(\rel T,\rel H_2)$ as defined in Example~\ref{example:1in3-nae}, we denote the single ternary relation of these structures by $R$.

The basic affine integer relaxation of an instance $\rel I$ of $\PCSP(\rel T,\rel H_2)$ is a~system of equations using variables $\mu_v(0)$, $\mu_v(1)$ for $v\in I$ bound by $\mu_v(0) + \mu_v(1) = 1$, and $\mu_{(v_1,v_2,v_3)}(\tup a)$ for $(v_1,v_2,v_3) \in R^{\rel I}$ and $\tup a \in R^{\rel T}$:
\begin{align*}
    \mu_{v_1,v_2,v_3}(0,0,1) + \mu_{v_1,v_2,v_3}(0,1,0) &= \mu_{v_1}(0) \\
    \mu_{v_1,v_2,v_3}(1,0,0) &= \mu_{v_1}(1) \\
    \mu_{v_1,v_2,v_3}(0,0,1) + \mu_{v_1,v_2,v_3}(1,0,0) &= \mu_{v_2}(0) \\
    \mu_{v_1,v_2,v_3}(0,1,0) &= \mu_{v_2}(1) \\
    \mu_{v_1,v_2,v_3}(0,1,0) + \mu_{v_1,v_2,v_3}(1,0,0) &= \mu_{v_3}(0) \\
    \mu_{v_1,v_2,v_3}(0,0,1) &= \mu_{v_3}(1)
\end{align*}
for each $(v_1,v_2,v_3) \in R^{\rel I}$. Since the value of $\mu_v(0)$ is determined by the value of $\mu_v(1)$, and moreover $\sum_{\tup a\in R^{\rel T}} \mu_{v_1,v_2,v_3}(\tup a) = \mu_{v_1}(0) + \mu_{v_1}(1) = 1$, we can simplify this system by dropping variables $\mu_v(0)$ and $\mu_{v_1,v_2,v_3}(i)$ and replacing the six equations above with
\[
    \mu_{v_1}(1) + \mu_{v_2}(1) + \mu_{v_3}(1) = 1
.\]
Note that each satisfying tuple of $R^{\rel T}$ satisfies this constraint. The resulting system is the same as suggested by \cite[Remark 3.3]{BG18}.
\end{example}

\begin{definition}
  A~function $a$ of arity $2n+1$ is called \emph{alternating}, if
  \[
    a( x_{1},\dots,x_{2n+1} ) \equals
    a( x_{\pi(1)},\dots,x_{\pi(2n+1)} )
  \]
  for all permutations $\pi$ that preserve parity, and
  \[
    a( x_{1},\dots,x_{2n-1},y,y ) \equals 
    a( x_{1},\dots,x_{2n-1},z,z ).
  \]
\end{definition}

The property of being alternating can be also expressed as a~variant of symmetry: the value is independent of the order of its inputs on odd positions, and also of the order of its inputs on even positions. Putting this together with the second identity which expresses some form of cancellation, we get that the output is dependent only on the multiset (where we allow negative coefficients) of inputs where the odd inputs are counted positively, and even inputs negatively.

\begin{example}
  An~example of an alternating function is the function $a\colon \mathbb Z^{2n+1} \to \mathbb Z$ defined as the alternating sum, i.e.,
  \[
    a(x_1,\dots,x_n) = x_1 - x_2 + x_3 - \dots + x_{2n+1}.
  \]
  Clearly permuting $x_i$'s with odd indices as well as permuting those with even indices does not change the value. Also
  \[
    a(x_1,\dots,x_{2n-1},y,y) = x_1 - \dots + x_{2n-1} - y + y = x_1 - \dots + x_{2n-1}
  \]
  which concludes that this value does not depend on $y$. Also note that the function defined by
  \[
    a'(x_1,\dots,x_{2n-1}) = a(x_1,\dots,x_{2n-1},x_1,x_1)
  \]
  is an~alternating function of arity $2n-1$. This is always the case as can be easily derived from the two defining identities.

  Note that such alternating sum can be generalised from $\mathbb Z$ to any abelian group. One such example, though degenerate, would be the Boolean \emph{parity function} (see \cite[p.\ 9]{BG16a}) of odd arity; it is defined as $p(x_1,\dots,x_{2n+1}) = \sum_{i\in [2n+1]} x_i \bmod 2$.
\end{example}

\begin{example}
  \label{ex:alternating-threshold}
  Another important example is the \emph{alternating threshold} defined in \cite[p.\ 9]{BG16a} as the~Boolean function $t$ of arity $2n+1$ satisfying:
  \[
    t(x_1,\dots,x_n) = \begin{cases}
      1 & \text{if $x_1 - x_2 + x_3 - \dots + x_{2n+1} > 0$, and} \\
      0 & \text{otherwise.}
    \end{cases}
  \]
  This function is obtained from the alternating sum by reflection:
  \[
    t(x_1,\dots,x_{2n-1}) = r(a(e(x_1),\dots,e(x_{2n-1})))
  \]
  where $e\colon \{0,1\} \to \mathbb Z$ is the natural inclusion, and $r\colon \mathbb Z\to \{0,1\}$ maps positive integers to $1$ and non-positive to $0$.
\end{example}

Let us now formulate the main result of this section.

\begin{theorem} \label{thm:ip-characterization}
Let $(\rel A,\rel B)$ be a PCSP template. The following are equivalent.
\begin{enumerate}
  \item \label{it1:lipc} \textsc{AIP} solves $\PCSP(\rel A,\rel B)$,
  \item \label{it2:lipc} $\Pol(\rel A,\rel B)$ contains alternating functions of all odd arities.
  \item \label{it3:lipc} $\Pol(\rel A,\rel B)$ admits a minion homomorphism from $\clo Z_\aff$,
  \item \label{it4:lipc} $(\rel A,\rel B)$ is pp-constructible from (a finite reduct of) $\rel Z_\aff$.
\end{enumerate}
\end{theorem}

The proof of this theorem follows closely the proof of Theorem~\ref{thm:lp-characterization}.
As in the mentioned proof, we rely on an appropriate modification of the free structure of $\clo Z_\aff$. In this case, there is no straightforward interpretation as probability distributions, but the core idea remains the same.

\begin{definition} For a~structure $\rel A$, we define an infinite structure $\IP(\rel A)$ similar to $\rel A$ in the following way. The universe $\IP(A)$ is the set of all mappings (tuples) $\phi\colon A \to \mathbb Z$ such that $\sum_{a\in A} \phi(a) = 1$. For a~$k$-ary relation $R^{\rel A}$, we define the corresponding relation $R^{\IP(\rel A)}$ as the set of all $k$-tuples $(\phi_1,\dots,\phi_k)$ for which there exists a~mapping $\gamma\colon R^{\rel A} \to \mathbb Z$ such that $\sum_{\tup a\in R^{\rel A}} \gamma(\tup a) = 1$ and
\[
  \sum_{\tup a\in R^{\rel A}, \tup a(i) = a} \gamma(\tup a) = \phi_i(a)
\]
for all $a\in A$ and $i\in [k]$.
\end{definition}

\begin{remark}
  It can be proven in very similar way as in Remark~\ref{rem:lp-is-free} that the structure $\IP(\rel A)$ is isomorphic to the free structure of $\clo Z_\aff$ generated by $\rel A$.
\end{remark}

Similarly to Lemma~\ref{lem:lp-free-is-pp}, the following can be viewed as an infinite case of Lemma~\ref{lem:free-is-pp}.

\begin{lemma} \label{lem:lip-free-is-pp}
  Let $\rel A$ be a~finite relational structure, and let $\IP(\rel A)$ denote the free structure of $\clo Z_\aff$ generated by $\rel A$. Then $(\rel A,\IP(\rel A))$ is a~relaxation of a~pp-power of $\rel Z_\aff$.
\end{lemma}

\begin{proof}
  The lemma is proven the same way as Lemma~\ref{lem:lp-free-is-pp}.
\end{proof}

\begin{proof}[Proof of Theorem \ref{thm:ip-characterization}]
   We claim that (1) is equivalent to the existence of a~homomorphism from $\IP(\rel A)$ to $\rel B$. The key is that an instance $\rel I$ of $\PCSP(\rel A,\rel B)$ maps to $\IP(\rel A)$ if and only if the AIP relaxation of $\rel I$ is a~solvable system of equations --- a~solution $\mu$ defines such a~homomorphism by $v\mapsto f_{\mu_v}$ where $f_{\mu_v}(x_1,\dots,x_n) = \sum_{a\in A} \mu_v(a)x_a$ (here, we assume $A = [n]$), and vice-versa. This immediately implies that if $\IP(\rel A)$ maps homomorphically to $\rel B$ then $\AIP$ solves $\PCSP(\rel A,\rel B)$. In the other direction, we know that every finite substructure of $\IP(\rel A)$ maps to $\rel B$, a~global homomorphism then follows by the standard compactness argument.

  As in the proof of Theorem~\ref{thm:lp-characterization}, the combination of the previous lemma with the proof of Theorem~\ref{thm:minor-homomorphism-is-pp-constructibility} gives us that the existence of a homomorphism from $\IP(\rel A)$ to $\rel B$ implies item (4), and (4)\textto(3). Further, (3)\textto(2) since $\clo Z_\aff$ has alternating functions of all odd arities (the alternating sums). 

  We finish the proof by showing that item (2) implies the existence of a~homomorphism from $\IP(\rel A)$ to $\rel B$. This again follows the proof of Theorem~\ref{thm:lp-characterization}. We first define some substructures of $\IP(\rel A)$: $\IP_\ell(\rel A)$ is a~structure similar to $\rel A$ whose universe is the set of all functions $\phi\colon A \to \mathbb Z$ such that $\sum_{a\in A} |\phi(a)| \leq 2\ell+1$ and $\sum_{a\in A} \phi(a) = 1$. The relations of $\IP_\ell(\rel A)$ are defined in the same way as those of $\IP(\rel A)$ with the only difference that we require that the witnessing function $\gamma\colon R^{\rel A}\to\mathbb Z$ also satisfies $\sum_{\tup r\in R^{\rel A}} |\gamma(\tup r)| \leq 2\ell+1$. Now, we define a~homomorphism $h_\ell\colon \IP_\ell(\rel A) \to \rel B$ by fixing a~$(2\ell+1)$-ary alternating function $a_{2\ell+1}\in \Pol(\rel A,\rel B)$ and setting
  \[
    h_\ell( \phi ) =
      a_{2\ell+1} (a_1,\dots,a_{2\ell+1})
  \]
  where $a_1,\dots,a_{2\ell+1}$ are chosen in such a~way that for all $a\in A$ the difference of the number of times $a$ appears among $a_i$ with odd and even indices is exactly $\phi(a)$. This is possible thanks to $\sum |\phi(a)| \leq 2\ell+1$. For $h$ to be well-defined, we rely on the fact that $a_{2\ell+1}$ is alternating. To prove that $h$ is a~homomorphism, suppose that $(\phi_1,\dots,\phi_k) \in R^{\IP_\ell(\rel A)}$ and this fact is witnessed by $\gamma\colon R^{\rel A} \to \mathbb Z$. Again, we pick tuples $\tup r_1,\dots,\tup r_{2\ell +1}$ in such a~way that for all $\tup r\in R^{\rel A}$ the difference of the number of times $\tup r$ appears among $\tup r_i$ with odd and even indices is exactly $\phi(\tup r)$. Now,
  \[
    (h_\ell(\phi_1),\dots,h_\ell(\phi_k)) =
      a_{2\ell+1}(\tup r_1,\dots,\tup r_{2\ell+1}) \in R^{\rel B}
  \]
  where the equality follows from the fact that, for each $i\in [k]$, the difference between the number of times some $a\in A$ appears among $\tup r_j(i)$ with odd and even $j$ is
  \(
    \phi_i(a) = \sum_{\tup r\in R^{\rel A}, \tup r(i) = a} \gamma(\tup r)
  \)
  (again we use the fact that $a_{2\ell+1}$ is alternating). We get that each $\IP_\ell(\rel A)$ maps homomorphically to $\rel B$, and since every finite substructure of $\IP(\rel A)$ is included in $\IP_\ell(\rel A)$ for some $\ell$, we get that it also maps to $\rel B$. The homomorphism from $\IP(\rel A)$ to $\rel B$ is then given by a~standard compactness argument.
\end{proof}
\section{More on tractability of 1-in-3- vs.\ NAE-Sat}
  \label{sec:fin_intract}
  
All known tractability results for PCSPs, such as those in the previous section and those in \cite{BG18,BG18b}, are obtained by following the same scheme --- namely by showing how a PCSP template is pp-constructed from a tractable CSP template (recall Definition~\ref{def:pp-constructible}), possibly with an infinite domain. 
In this section we show that using infinite domains in this scheme can be necessary. Namely, we show that this is the case for the 1-in-3- vs.\ Not-All-Equal-\Sat\ problem (see Example \ref{example:1in3-nae}).
This problem is in {\Ptime} \cite{BG18}, since its template $(\rel T,\rel H_2)$ can be pp-constructed from $\rel Z_\aff$ (see Example~\ref{example:1in3-nae2} and Theorem~\ref{thm:ip-characterization}). Specifically, it is easy to check (or see \cite{BG18}) that this template is a relaxation of the CSP template $(\mathbb Z; x+y+z=1)$, which is a finite reduct of $\rel Z_\aff$.
The template $(\rel T,\rel H_2)$ is also a relaxation of other tractable CSP templates with an infinite domain (see \cite{BG18,BG18b,Bar19}).  

\begin{theorem} \label{thm:infinity-necessary}
Let $\rel D$ be a~finite relational structure  such that $(\OneInThree, \NAE_2)$ is pp-constructible from $\rel D$.
Then $\CSP(\rel D)$ is \NP-complete.
\end{theorem}

The rest of this section is devoted to the proof of this theorem. 

\subsection{Proof outline}
Striving for a contradiction, assume that $\CSP(\rel D)$ is not \NP-complete and $(\OneInThree, \NAE_2)$ is pp-constructible from $\rel D$. We start by simplifying the latter assumption.

From (\ref{it2:h1pp}) $\rightarrow$ (\ref{it1:h1pp}) in Theorem~\ref{thm:minionhom-and-pp} we know that $(\OneInThree, \NAE_2)$ is a homomorphic relaxation of a pp-power of $\rel D$. Since a pp-power of a~finite tractable CSP template is a~finite tractable CSP template, we may assume that $(\OneInThree, \NAE_2)$ is a homomorphic relaxation of $\rel D$. Let $\rel D = (D; R)$, where $R \subseteq D^3$, and let  $f\colon \OneInThree \to \rel D$ and $g\colon \rel D \to \NAE_2$ be homomorphisms from the definition of homomorphic relaxation, Definition~\ref{def:relaxation}.

We simplify the situation a bit further. Since $gf$ is a homomorphism, this mapping applied component-wise to the 1-in-3 tuple $(0,0,1)$ is a not-all-equal tuple. In particular $f(0) \neq f(1)$. We rename the elements of $D$ so that $\{0,1\} \subseteq D$ and $f(0)=0$, $f(1)=1$. As $f$ and $g$ are homomorphisms, we get
$$
  \{0,1\} \subseteq D, \quad \{(1,0,0), (0,1,0), (0,0,1)\} \subseteq R
$$
and 
$$
  |\{g(a),g(b),g(c)\}|>1 \mbox{ whenever } (a,b,c) \in R.
$$

Now we employ the assumption that $\CSP(\rel D)$ is not \NP-complete. We use a~sufficient condition for \NP-completeness from~\cite{BK12}.

\begin{definition}
An operation $s \colon D^n \rightarrow D$ is called \emph{cyclic} if, for all $(a_1, \dots, a_n) \in D^n$, we have
\[
  s(a_1,a_2, \dots, a_n) = s(a_2, \dots, a_n, a_1)
.\]
\end{definition}

\begin{theorem}[\cite{BK12}] \label{thm:cyclic}
Let $\rel D$ be a~finite CSP template. If $\CSP(\rel D)$ is not NP-complete, then $\rel D$ has a~cyclic polymorphism of arity $p$ for every prime number $p > |D|$. 
\end{theorem}

\begin{remark}
Cyclic polymorphisms in fact characterise the borderline between \NP-complete and tractable CSPs conjectured in~\cite{BJK05} and proved in~\cite{Bul17,Zhu17}: $\CSP(\rel D)$ is tractable if and only if (assuming $\Ptime \neq \NP$) $\rel D$ has a cyclic polymorphism of arity at least 2 (if and only if $\rel D$ has a cyclic polymorphism of arity $p$ for every prime number $p > |D|$).
\end{remark}

By Theorem~\ref{thm:cyclic}, $\rel D$ has a cyclic polymorphism of any prime arity $p > |D|$. We fix a cyclic polymorphism $s$ of prime arity $p > 60 |D|$.

Recall that the polymorphisms of CSP template can be composed (to produce new polymorphisms).
Next we define an operation $t$ on $D$ of arity $p^2$ by
\begin{align*}
t&(x_{11}, x_{12}, \ldots, x_{1p}, x_{21}, x_{22}, \ldots x_{2p}, x_{31}, \ldots, \ldots, x_{pp}) \\
 &= s(s(x_{11},x_{21}, \ldots, x_{p1}), 
     s(x_{12},x_{22}, \ldots, x_{p2}), 
		 \dots 
		 s(x_{1p},x_{2p}, \ldots, x_{pp})).
\end{align*}

It will be convenient to organise the arguments of $t$ into a $p \times p$ matrix $X$ whose entry in the $i$-th row and $j$-th column is $x_{ij}$, so the value
$$
t 
\begin{pmatrix}
x_{11} & x_{12} & \cdots & x_{1p} \\
x_{21} & x_{22} & \cdots & x_{2p} \\
\vdots & \vdots & \ddots & \vdots \\
x_{p1} & x_{p2} & \cdots & x_{pp}
\end{pmatrix}
$$
is obtained by applying $s$ to the column vectors and then $s$ to the resulting row vector.

We introduce several  concepts for zero-one matrices, since only 0-1 values for the variables in $t$ will play a role in the proof. 

\begin{definition}
Let $X=(x_{ij}), Y$ be $p \times p$ zero-one matrices. The \emph{area of $X$} is the fraction of ones and is denoted
$$
\lambda(X) = \left(\sum_{i,j} x_{ij}\right)/p^2.
$$
The matrices $X,Y$ are called \emph{$g$-equivalent}, denoted $X \sim Y$, if $g(t(X)) = g(t(Y))$.
The matrix $X$ is called \emph{tame} if 
\begin{align*}
\text{either}\quad &X \sim 0_{p \times p}   \quad\text{and}\quad  \lambda(X) < 1/3, \\
\text{or}\quad &X \sim 1_{p \times p}  \quad\text{and}\quad  \lambda(X) > 1/3, 
\end{align*}
where $0_{p \times p}$ stands for the zero matrix and $1_{p \times p}$ for the all-ones matrix.
\end{definition}

Observe that the equivalence $\sim$ has two blocks, so, e.g., $X \not\sim Y \not\sim Z$ implies $X \sim Z$. Also recall that $p>3$ is a prime number, so the area of $X$ is never equal to $1/3$.

The proof now proceeds as follows. We show that certain matrices, called ``almost rectangles'', are tame. The proof is by induction (although the proof logic, as presented, is a bit different). Subsection~\ref{subsec:lines} provides the base case and Subsection~\ref{subsec:rect} handles the induction step.
In Subsection~\ref{subsec:contra}, we construct two tame matrices $X_1$, $X_2$ such that $\lambda(X_1)<1/3$ and $\lambda(X_2) > 1/3$, but $t(X_1) = t(X_2)$ (because the corresponding columns of $X_1$ and $X_2$ will be evaluated by $s$ to the same elements). This  gives us a contradiction since $0_{p \times p} \not\sim 1_{p \times p}$ as we shall see.

Before launching into the technicalities, we introduce an additional concept and state a consequence of the fact that $s$ is a polymorphism.

\begin{definition}
A triple $X,Y,Z$ of $p \times p$ zero-one matrices is called a \emph{cover} if, for every $1 \leq i,j \leq p$, exactly one of $x_{ij},y_{ij},z_{ij}$ is equal to one.
\end{definition}

\begin{lemma} \label{lem:nae}
If $X,Y,Z$ is a cover, then $X,Y,Z$ are not all $g$-equivalent. 
\end{lemma}

\begin{proof}
By the definition of a cover, the triple $(x_{ij},y_{ij},z_{ij})$ is in $\{(0,0,1)$, $(0,1,0)$, $(1,0,0)\} \subseteq R$ for each $i,j$. 
Since $t$ preserves $R$ (because $s$ does), the triple $(t(X)$, $t(Y)$, $t(Z))$ is in $R$ as well. 
Finally, $g$ is a homomorphism from $\rel D$ to $\NAE_2$, therefore $g(t(X))$, $g(t(Y))$, $g(t(Z))$ are not all equal. 
In other words, $X$, $Y$, $Z$ are not all $g$-equivalent, as claimed.
\end{proof}

\subsection{Line segments are tame} \label{subsec:lines}

In this subsection it will be more convenient to regard the arguments of $t$ as a tuple $\vc{x} = (x_{11},x_{12}, \ldots)$ of length $p^2$ rather than a matrix. The concepts of the area, $g$-equivalence, tameness, and cover is extended to tuples in the obvious way.
Since $p>3$ is a prime number, $p^2$ is 1 modulo 3. Let $q$ be such that
$$
 p^2 = 3q+1.
$$
Moreover, let $\ttt{i}$ denote the following tuple of length $p^2$.
$$
\ttt{i} = (\underbrace{1,1, \cdots, 1}_{i \mbox{ positions} }, 0,0, \cdots 0)
$$
We prove in this subsection that all such tuples are tame.
We first recall a well-known fact.

\begin{lemma}
The operation $t$ is cyclic.
\end{lemma}

\begin{proof}
By cyclically shifting the arguments we get the same result:
\begin{multline*}
t(x_{12}, \cdots, x_{pp},x_{11}) =
t 
\begin{pmatrix}
x_{12} & x_{13} & \cdots & x_{1p} & x_{21} \\
x_{22} & x_{23} & \cdots & x_{2p} & x_{31} \\
\vdots & \vdots & \ddots & \vdots & \vdots\\
x_{p2} & x_{p3} & \cdots & x_{pp} & x_{11} 
\end{pmatrix}
= \\
t 
\begin{pmatrix}
x_{21} & x_{12} & x_{13} & \cdots & x_{1p}  \\
x_{31} & x_{22} & x_{23} & \cdots & x_{2p}  \\
\vdots & \vdots & \ddots & \vdots & \vdots \\
x_{11} & x_{p2} & x_{p3} & \cdots & x_{pp}   
\end{pmatrix}
=
t 
\begin{pmatrix}
x_{11} & x_{12} & \cdots & x_{1p} \\
x_{21} & x_{22} & \cdots & x_{2p}  \\
\vdots & \vdots & \ddots & \vdots \\
x_{p1} & x_{p2} & \cdots & x_{pp}  
\end{pmatrix}
=
t(x_{11}, x_{12}, \cdots, x_{pp}),
\end{multline*}
where the second equality uses the cyclicity of the outer ``$s$'' in the definition of $t$, while the third one the cyclicity of the first inner ``$s$''.
\end{proof}

\begin{lemma}
 $\ttt{0} \sim \ttt{1} \sim \dots \sim \ttt{q} \not\sim \ttt{q+1} \sim \dots \sim \ttt{2q} \sim \ttt{2q+1}$.
\end{lemma}

\begin{proof}
By induction on $i=0,1, \dots, q$, we prove
\[
\ttt{q-i} \sim \ttt{q-i+1} \sim \cdots \sim \ttt{q} 
\not\sim \ttt{q+1} \sim \cdots \sim \ttt{q+i} \sim \ttt{q+i+1}.
\]

For the induction base, $i=0$,
let $\vc{x} = \ttt{q}$, let $\vc{y}$ be $\ttt{q}$ (cyclically) shifted $q$ times to the right (so the first 1 is at the $(q+1)$-st position), and let $\vc{z}$ be $\ttt{q+1}$ shifted $2q$ times to the right. 
The tuples $\vc{x},\vc{y},\vc{z}$ form a cover, therefore they are not all $g$-equivalent by Lemma~\ref{lem:nae}.
But $t$ is cyclic, thus $t(\vc{x})=t(\vc{y}) = t(\ttt{q})$ and $t(\vc{z}) = t(\ttt{q+1})$. It follows that $\ttt{q},\ttt{q},\ttt{q+1}$ are not all $g$-equivalent and we get $\ttt{q} \not\sim \ttt{q+1}$.

Now we prove the claim for $i>0$ assuming it holds for $i-1$. To verify $\ttt{q-i} \sim \ttt{q-i+1}$ consider
$\ttt{q-i}$, $\ttt{q+1}$, $\ttt{q+i}$. Since $(q-i)+(q+1)+(q+i) = 3q+1=p^2$, these tuples can be cyclically shifted to form a cover and then the same argument as above gives us that $\ttt{q-i}$, $\ttt{q+1}$, $\ttt{q+i}$ are not all $g$-equivalent. But $\ttt{q+1} \sim \ttt{q+i}$ by the induction hypothesis, therefore $\ttt{q-i} \not\sim \ttt{q+1}$. Since $\ttt{q+1} \not\sim \ttt{q-i+1}$ (again by the induction hypothesis), we get $\ttt{q-i} \sim \ttt{q-i+1}$, as required.

It remains to check $\ttt{q+i} \sim \ttt{q+i+1}$. This is done in a similar way, using the tuples $\ttt{q-i}$, $\ttt{q}$, $\ttt{q+i+1}$. 
\end{proof}

We have proved that $\ttt{0} \sim \dots \sim \ttt{q} \not\sim \ttt{q+1} \sim \dots \sim \ttt{2q+1}$.
Using the same argument as in the previous lemma once more for $\ttt{0},\ttt{p^2-i},\ttt{i}$ with $p^2 \geq i > 2q+1$ we get $\ttt{i} \not\sim \ttt{0}$. In summary,
$\ttt{i} \sim \ttt{0}$ whenever $i \leq q$ and $\ttt{i} \sim \ttt{p^2} \not\sim \ttt{0}$ when $i \geq q+1$. 
Observing that $\lambda(\ttt{i})<1/3$ if and only if $i \leq q$, we obtain the following lemma.

\begin{lemma} \label{lem:lines_tame}
Each $\ttt{i}$, $i \in \{0, 1, \cdots, p^2\}$, is tame and $\ttt{0} \not\sim \ttt{p^2}$. 
\end{lemma}

\subsection{Almost rectangles are tame} \label{subsec:rect}

We start by introducing a special type of zero-one matrices. 

\begin{definition}
Let $1 \leq k_1, \dots, k_p \leq p$. By
$\rrr{k_1,k_2,\dots,k_p}$
we denote the matrix whose $i$-th column begins with $k_i$ ones followed by ($p-k_i$) zeros, for each $i \in \{1, \dots, p\}$.

An \emph{almost rectangle} is a matrix of the form $\rrr{k,k, \dots, k, l, l, \dots, l}$ (the number of $k$'s can be arbitrary, including 0 or $p$) where $0 \leq k-l \leq 5|D|$. The quantity $k-l$ is referred to as the \emph{size of the step}. 
\end{definition}

In the remainder of this subsection we prove the following proposition.

\begin{proposition} \label{prop:tame}
Each almost rectangle is tame.
\end{proposition}

Let 
$$
X=\rrr{\underbrace{k, k, \cdots, k}_{m \text{ positions}}, l,l, \dots, l}
$$
be a minimal counterexample in the following sense.
\begin{itemize}
\item $X$ has the minimum size of the step and,
\item among such counterexamples, $\abs{\lambda(X)-1/3}$ is maximal.
\end{itemize}

\begin{lemma}
The size of the step of $X$ is at least 2.
\end{lemma}

\begin{proof}
This lemma is just a different formulation of Lemma~\ref{lem:lines_tame}, since an almost rectangle with step of size 0 or 1 represents the same choice of arguments as $\ttt{i}$ for some $i$. 
\end{proof}

We handle two cases $\lambda(X) \geq 5/12$ and $\lambda(X) \leq 5/12$ separately, but the basic idea for both of them is the same as in the proof of Lemma~\ref{lem:lines_tame}. To avoid puzzles, let us remark that any number strictly between $1/3$ and $1/2$ (instead of $5/12$) would work with a sufficiently large $p$. 

\begin{lemma} \label{lem:c}
The area of $X$ is less than $5/12$.
\end{lemma}

\begin{proof}
Assume that $\lambda(X) \geq 5/12$. Let $k_1$, $k_2$, $l_1$, and $l_2$ be the non-negative integers such that 
\begin{eqnarray}
l_1+l_2+k = p = k_1+k_2+l, \label{eq:a} \\
1 \geq k_1 - k_2 \geq 0, \mbox{ and }
1 \geq l_1-l_2 \geq 0 . \label{eq:b}
\end{eqnarray} 
We have $k_1 \geq l_1$ and $k_2 \geq l_2$. Moreover,
since $k-l \geq 2$ by the previous lemma, it follows that both $k_1-l_1$ and $k_2-l_2$ are strictly smaller than $k-l$.

Consider the matrices
$$
Y_i = \rrr{\underbrace{l_i,l_i, \dots, l_i}_{m \text{ positions}}, k_i, k_i, \dots, k_i}, \quad i=1,2.
$$
By shifting all the rows of $Y_i$, $i \in \{1,2\}$, $m$ times to the left we obtain an almost rectangle with a smaller step size than $X$, which is thus tame by the minimality assumption on $X$. 
Since such a shift changes neither the value of $t$ (as the outer ``$s$'' in the definition of $t$ is cyclic) nor the area, both $Y_1$ and $Y_2$ are tame matrices.

Let $Y_1'$ ($Y_2'$, resp.) be the matrices obtained from $Y_1$ ($Y_2$, resp.) by shifting the first $m$ columns $k$ times ($k+l_1$ times, resp.) down and the remaining columns $l$ times ($l+k_1$ times, resp.) down. Since $X,Y_1',Y_2'$ is a cover (by~(\ref{eq:a})) and
cyclically shifting columns does not change the value of $t$ (as the inner occurrences of ``$s$'' in the definition of $t$ are cyclic), Lemma~\ref{lem:nae} implies that $X$, $Y_1$, $Y_2$ are not all $g$-equivalent. 

From $X, Y_1', Y_2'$ being a cover, it also follows that
$$
\lambda(X) + \lambda(Y_1') + \lambda(Y_2') = \lambda(X) + \lambda(Y_1) + \lambda(Y_2) = 1.
$$
Moreover, by~(\ref{eq:b}), we have $\lambda(Y_2) \leq \lambda(Y_1)$ and these areas differ by at most $p/p^2=1/p$. Therefore
$$
\lambda(Y_1) = 1 - \lambda(X) - \lambda(Y_2) \leq 1 - 5/12 - \lambda(Y_1) + 1/p
$$
and, since $p > 12$ by the choice of $p$, we obtain
$$
\lambda(Y_2) \leq \lambda(Y_1) < 1/3.
$$

The tameness of $Y_i$ now gives us $Y_1 \sim Y_2 \sim 0_{p \times p}$ and then, since $Y_1,Y_2,X$ are not all $g$-equivalent and $0_{p \times p} \not\sim 1_{p \times p}$ (by the second part of Lemma~\ref{lem:lines_tame}), 
we get $X \sim 1_{p \times p}$. But $\lambda(X) \geq 5/12 > 1/3$, hence $X$ is tame, a contradiction with the choice of $X$. 
\end{proof}

It remains to handle the case $\lambda(X) < 5/12$. 

We first claim that $2k$ (and thus $k+l$ and $2l$) is less than $p$. Indeed, since the step size of $X$ is at most $5|D|$ (by the definition of an almost rectangle) and $p > 60|D|$, we get
\begin{align*}
5/12 > \lambda(X) &\geq \frac{p(k - 5|D|)}{p^2}, \mbox{ and hence } \\
k & \leq 5p/12 + 5|D| < 5p/12 + p/12 = p/2.
\end{align*}

We now again need to distinguish two cases.
Assume first that $m < p/2$. 

Let 
\begin{align*}
Y &= \rrr{\underbrace{l, \cdots,l}_{m \text{ positions} },\underbrace{k, \cdots, k}_{m \text{ positions} }, l, \cdots, l},  \\
Z &= \rrr{\underbrace{p-k-l, \cdots, p-k-l}_{2m \text{ positions}}, p-2l, \cdots, p-2l} .
\end{align*}
The definition of $Z$ makes sense since $p-k-l, p-2l \geq 0$ by the inequality $2k < p$ derived above. 

The triple $X,Y,Z$ (similarly to $X,Y_1,Y_2$ in the proof of Lemma~\ref{lem:c}) is such that we can obtain a cover by shifting the columns down. Therefore $X$, $Y$, $Z$ are not all $g$-equivalent and $\lambda(X)+\lambda(Y) + \lambda(Z)=1$.

On the other hand, by shifting all the rows of $Y$ $m$ times to the left we obtain $X$.
We get $\lambda(X) = \lambda(Y)$ and $t(X)=t(Y)$, therefore  $Z \not\sim X$ by the previous paragraph. 

Moreover,
 by shifting all the rows of $Z$ $2m$ times to the left we obtain an almost rectangle $Z'$ with $t(Z)=t(Z')$ and $\lambda(Z)=\lambda(Z')$. The step size of $Z'$ is $(p-2l) - (p-k-l) = k-l$, which is the same as the step size of $X$. However, the distance of its area from $1/3$ is strictly greater as shown by the following calculation. 
\[
\frac{\abs{\lambda(Z)-1/3}}{\abs{\lambda(X)-1/3}} =
\frac{\abs{(1 - 2\lambda(X)) - 1/3}}{\abs{\lambda(X)-1/3}} =
\frac{\abs{2(1/3 - \lambda(X))}}{\abs{\lambda(X)-1/3}} = 2 > 1.
\]
By the minimality of $X$, the almost rectangle $Z'$ is tame and so is $Z$. It is also apparent from the calculation that the signs of $\lambda(X)-1/3$ and $\lambda(Z)-1/3$ are opposite. Combining these two facts with $Z \not\sim X$ derived above, we obtain that $X$ is tame, a contradiction.

In the other case, when $m > p/2$, the proof is similar using the tuples
\begin{align*}
Y &= (l, \cdots, l, \underbrace{k, \cdots, k}_{m \text{ positions}}), \\
Z &= (\underbrace{p-k-l, \cdots, p-k-l}_{(p-m) \text{ positions}}, p-2k, \cdots, p-2k, 
  \underbrace{p-k-l, \cdots, p-k-l}_{(p-m) \text{ positions}}) .
\end{align*}
The proof of Proposition~\ref{prop:tame} is concluded.

\subsection{Contradiction} \label{subsec:contra}

Let 
$$
m = (p-1)/2
$$
and choose natural numbers $l_1$ and $l_2$ so that
$$
p/3 - 2|D| < l_1 < l_2 < p/3 
$$
and
$$
s(\underbrace{1, \cdots, 1}_{l_1 \text{ positions}}, 0, \cdots, 0)
=
s(\underbrace{1, \cdots, 1}_{l_2 \text{ positions}}, 0, \cdots, 0).
$$
This is possible by the pigeonhole principle, since there are $2|D| > D$ integers in the interval and $p/3 - 2|D| > 0$ by the choice of  $p$.

The required contradiction will be obtained by considering the two matrices
$$
X_i = \rrr{ \underbrace{k, \dots, k}_{m \text{ positions}}, l_i, \dots, l_i}, \ i=1,2 ,
$$ 
where $k$ will be specified soon.

Before choosing $k$, we observe that $t(X_1) = t(X_2)$. Indeed, the first $m$ columns of these matrices are the same (and thus so are their images under $s$) and the remaining columns have the same image under $s$ by the choice of $l_1$ and $l_2$. The claim thus follows from the definition of $t$. 

Next, note that for $k \leq p/3$ the area of both matrices is less than $1/3$ since $l_i < p/3$. On the other hand, for $k \geq p/3 + 3|D|$ the area is greater:
\begin{multline*}
\lambda(X_i) 
= \frac{mk+(p-m)l_i}{p^2} 
\geq \frac{\frac{p-1}{2} (p/3 + 3|D|) + \frac{p+1}{2} (p/3-2|D|)}{p^2} \\
= \frac{p^2/3 + |D| (p-5)/2}{p^2} > 1/3.
\end{multline*}   

Choose the maximum $k$ so that $\lambda(X_1) < 1/3$. The derived inequalities and the choice of $l_i$ implies
$$
l_1 < l_2 \leq k < p/3 + 3|D| \leq l_1 + 5|D| < l_2 + 5|D| ,
$$
therefore both $X_1$ and $X_2$ are almost rectangles. By Proposition~\ref{prop:tame}, $X_1$ and $X_2$ are tame.

Since the area of $X_1$ is less than $1/3$, we get $X_1 \sim 0_{p \times p}$.   
We chose $k$ so that increasing $k$ by 1 makes the area of $X_1$ greater than $1/3$. 
From $m < p/2$ it follows that increasing $l_1$ by 1 makes the area even greater, hence $\lambda(X_2) > 1/3$ (recall that $l_2>l_1$)
and we obtain $X_2 \sim 1_{p \times p}$. 

Recall that $0_{p \times p} \not\sim 1_{p \times p}$ by the second part of Lemma~\ref{lem:lines_tame}.
Therefore $X_1 \not\sim X_2$, contradicting $t(X_1) = t(X_2)$ and concluding the proof of Theorem~\ref{thm:infinity-necessary}.
\medskip

\begin{remark}
The proof of Theorem~\ref{thm:infinity-necessary} could be simplified if we had stronger or more suitable polymorphisms than cyclic operations. Alternative versions of Theorem~\ref{thm:cyclic} could also help in simplifying the proof of the CSP dichotomy conjecture. 
In particular, the following question seems open.
Does every finite $\rel D$ with a cyclic polymorphism of arity at least 2 necessarily have a polymorphism $s$ of arity $n > 1$ such that, for any $a,b \in D$ and $(x_1, \dots, x_n) \in \{a,b\}^n$, the value $s(x_1, \dots, x_n)$ depends only on the number of occurrences of $a$ in $(x_1, \dots, x_n)$?
Note that a more optimistic version involving evaluations with $|\{x_1, \dots, x_n\}|=3$ 
is false, a counterexample is the disjoint union of a~directed 2-cycle and a~directed 3-cycle.
\end{remark}

\section{Algebraic constructions}
  \label{sec:wonderland}

The main goal of this section is to provide a more wholesome picture of the algebraic theory of minions and to align the algebraic approach to PCSP more closely with the standard theory as presented in \cite{BKW17}. The theory that we present is a~natural generalisation of \cite{BOP18} from clones to minions. We note that, as in \cite{BOP18}, many of these results can be generalised to minions over infinite sets. Nevertheless, we keep our focus on the case of finite domains.

Let us start with describing algebraic counterparts of pp-constructions.

\begin{definition}
Let $\clo M_1$ and $\clo M_2$ be two minions on $(A_1,B_1)$ and $(A_2,B_2)$, respectively. We use the following terminology.
\begin{itemize}
  \item $\clo M_2$ is a~\emph{reflection} of $\clo M_1$ if there are maps $h_A\colon A_2 \to A_1$ and $h_B\colon B_1 \to B_2$ such that
\[
      \clo M_2 = \{ (x_1,\dots,x_{\ar f}) \mapsto h_B(f(h_A(x_1),\dots,h_A(x_{\ar f}))) \mid f\in \clo M_1 \},
    \]
  \item $\clo M_2$ is an~\emph{expansion} of $\clo M_1$ if $A_1 = A_2$, $B_1 = B_2$, and $\clo M_1 \subseteq \clo M_2$, and
  \item $\clo M_2$ is the~\emph{$n$-th power} of $\clo M_1$ (sometimes denoted $\clo M_2=\clo M_1^n$) if $A_2 = A_1^n$, $B_2 = B_1^n$, and $\clo M_2$ is the set of coordinate-wise actions of all functions in $\clo M_1$. More precisely, for each $t\ge 1$, $\clo M_2^{(t)}$ consist of all functions of the form 
  \[
    ((x_{11},\ldots,x_{1n}),\ldots,(x_{t1},\ldots,x_{tn})) \mapsto 
      (f(x_{11},\ldots,x_{t1}),\ldots,f(x_{1n},\ldots,x_{tn}))
  \]  
where $f$ is a~$t$-ary function from $\clo M_1$.
\end{itemize}
Following the standard notation, we denote the class of all reflections of $\clo M$ by $\Ref \clo M$, the class of all expansions by $\E \clo M$, the class of all powers by $\Pow \clo M$, and the class of all finite powers by $\Pfin \clo M$.
\end{definition}

We remark that a reflection of a minion is always a minion. This is in contrast with a reflection of a~clone (as defined in \cite[Definition 4.3]{BOP18}), which is not necessarily a clone (see \cite[p.\ 379]{BOP18}).

The following lemma relates the above to the relational constructions
from Sections~\ref{sec:pp-formulas} and~\ref{sec:pp-constr}. 

\begin{lemma} \label{lem:alg-rel-connection}
  Let $(\rel A,\rel B)$ and $(\rel A',\rel B')$ be two templates. Denote their polymorphism minions by $\clo M$ and $\clo M'$, respectively.
  \begin{enumerate}
    \item $(\rel A',\rel B')$ is ppp-definable from $(\rel A,\rel B)$ if and only if $\clo M' \in \E \clo M$.
    \item $(\rel A',\rel B')$ is a~strict relaxation of a~pp-power of $(\rel A,\rel B)$ if and only if $\clo M' \in \EPfin \clo M$.
    \item $(\rel A',\rel B')$ is a~relaxation of a~structure that is
      pp-definable from  $(\rel A,\rel B)$ if and only if $\clo M' \in \ER \clo M$.
  \end{enumerate}
\end{lemma}

\begin{proof}
Item (1) follows directly from Theorem~\ref{thm:galois-corespondence} (see \cite{Pip02,BG16a} for a~proof).

The proof of item (2) is similar to \cite[Proposition 3.1]{Bar13}.
Let us first define templates $(\rel A_n,\rel B_n)$ that will have the property that all $n$-th pp-powers of $(\rel A,\rel B)$ are ppp-definable in $(\rel A_n,\rel B_n)$. We set
$A_n = A^n$, $B_n = B^n$, and for each relation $R$ of $(\rel A,\rel B)$ of arity $k$ and each tuple $(i_1,\dots,i_k) \in [n]^k$ we define a~new $k$-ary relation $R_{i_1\dots i_k}$ by
\[
  ((a_{11},\dots,a_{1n}),\dots,(a_{k1},\dots,a_{kn})) \in R_{i_1\dots i_k}^{\rel A_n}
    \mathrel{\text{if}}
  (a_{1i_1},\dots,a_{ki_k}) \in R^{\rel A},
\]
and similarly for $R_{i_1 \dots i_k}^{\rel B_n}$. Further, for all $i,j\in [n]$, we add a~new relation symbol $E_{i,j}$ interpreted as \( ((a_1,\dots,a_n),(b_1,\dots,b_n)) \in E_{i,j}^{\rel A_n} \) if $a_i = b_j$, and similarly for $E_{i,j}^{\rel B_n}$.  It follows directly from definitions that any structure that is ppp-definable in $(\rel A_n,\rel B_n)$ is an~$n$-th pp-power of $(\rel A,\rel B)$ and vice-versa. Further, we claim that $\clo M^n = \Pol(\rel A_n,\rel B_n)$.
It is straightforward to verify that $\clo M^n \subseteq \Pol(\rel A_n,\rel B_n)$. For the reverse inclusion, first observe that every polymorphism of $(\rel A_n,\rel B_n)$ acts independently on each of the coordinates; however the actions on different coordinates are identical, as follows from the compatibility with the the relations $E_{ij}$. From the compatibility with the relations $R_{i_1\dots i_k}$ we obtain that the action on each coordinate is a~polymorphism of $(\rel A,\rel B)$.   
We conclude that $\Pol(\rel A_n,\rel B_n) = \clo M^n$.

Now, assume that $(\rel A',\rel B')$ is an~$n$-th ppp-power of $(\rel A,\rel B)$. Since it is ppp-definable in $(\rel A_n,\rel B_n)$, we have $\clo M' \in \E \clo M^n \subseteq \EPfin \clo M$.
The converse is also proved using the template $(\rel A_n,\rel B_n)$. Given that $\clo M' \in \EPfin \clo M$, there is $n$ such that $\clo M' \in \E \clo M^n$. The rest again follows from the first item.

To prove item (3), assume that $(\rel A',\rel B')$ is a~relaxation of a~template $(\rel A_0,\rel B_0)$ that is pp-definable from $(\rel A,\rel B)$, and let $h_A\colon \rel A' \to \rel A_0$ and $h_B\colon \rel B_0 \to \rel B'$ be the relaxing homomorphisms. We claim that all mappings of the form $(x_1,\dots,x_n) \to h_B(f(h_A(x_1),\dots,h_A(x_n)))$, where $f\in \clo M$, are polymorphisms from $\rel A'$ to $\rel B'$. Since $(\rel A_0,\rel B_0)$ is pp-definable in $(\rel A,\rel B)$, each $f\in \clo M$ is also a~polymorphism from  $\rel A_0$ to $\rel B_0$. Consequently, we can prove that the above composition is a polymorphism from $\rel A'$ to $\rel B'$ in the same way as proving that a~composition of homomorphisms of relational structures is a~homomorphism. This concludes that $\clo M' \in \ER \clo M$.

For the other implication, suppose that $\clo M'\in \ER \clo M$, and let the reflection be defined by mappings $h_A$ and $h_B$. We define an~intermediate template $(\rel A_0,\rel B_0)$ by putting $A_0 = A$, $B_0 = B$, $R^{\rel A_0} = h_A(R^{\rel A'})$, and $R^{\rel B_0} = h_B^{-1}(R^{\rel B'})$ for all $R$.  This definition ensures that $(\rel A',\rel B')$ is a~relaxation of $(\rel A_0,\rel B_0)$.  To prove that $(\rel A_0,\rel B_0)$ is ppp-definable in~$(\rel A,\rel B)$, it is enough to show that $\clo M \subseteq \Pol(\rel A_0,\rel B_0)$; the rest follows from item (1).  

Let $f\in \clo M$ be of arity $n$ and $\tup c_1,\dots,\tup c_n \in R^{\rel A_0}$. From the definition of $R^{\rel A_0}$, we get $\tup b_1,\dots,\tup b_n \in R^{\rel A'}$ such that $\tup c_i = h_A(\tup b_i)$ for each $i = 1,\dots,n$.  Now,
\[
  h_B(f(h_A(\tup b_1),\dots,h_A(\tup b_n))) \in R^{\rel B'}
\]
since the composition $h_B(f(h_A,\dots,h_A))\in \clo M'$, and consequently 
\[
 f(\tup c_1,\dots,\tup c_n) = f(h_A(\tup b_1),\dots,h_A(\tup b_n)) \in h_B^{-1}(R^{\rel B'}) = R^{\rel B_0}
.\]
This shows that $(\rel A',\rel B')$ is a~relaxation of a~structure that is ppp-definable from $(\rel A,\rel B)$, therefore it is pp-constructible from $(\rel A,\rel B)$ which is the same as being a~relaxation of a~structure that is pp-definable from $(\rel A,\rel B)$ by Theorem~\ref{thm:minor-homomorphism-is-pp-constructibility}.
\end{proof}

The following theorem is an algebraic version of Theorem~\ref{thm:minor-homomorphism-is-pp-constructibility}. Its proof is therefore very similar, and a~connection to constructions from Section~\ref{sec:pp-constructions} should be immediate.

\begin{theorem} \label{thm:birkhoff}
  Let $\clo M$ and $\clo M'$ be two minions, respectively. Then $\clo M' \in \ERPfin \clo M$ if and only if there exists a~minion homomorphism from $\clo M$ to $\clo M'$.
\end{theorem}

\begin{proof}
Assume that $\clo M \subseteq\clo O(A,B)$ and $\clo M' \subseteq \clo O(A',B')$.
If $\clo M' \in \ERPfin \clo M$ then there is a~natural mapping $\xi\colon \clo M \to \clo M'$ that is clearly a~minion homomorphism. 

For the converse, suppose that $\xi\colon \clo M \to \clo M'$ is a~minion homomorphism.  We will use this mapping to find a~reflection of a~suitable power of $\clo M$ that gives a~subset of $\clo M'$. A~suitable exponent of the power is $N = A^{A'}$. For the reflection, we need to define two mappings, $h_A\colon A'\to A^N$ and $h_B\colon B^N \to B'$. We will choose them in such a~way that the following holds
\begin{equation}
  h_B\bigl(f^{\clo M^N}(h_A(x_1),\dots,h_A(x_n))\bigr) = \xi(f)(x_1,\dots,x_n)
    \label{eq:birkhoff}\tag{$\diamondsuit$}
\end{equation}
for all $f\in \clo M$ and $x_1,\dots,x_n\in A'$.
First, we define $h_A\colon A' \to A^N$ as the long code encoding of an element of $A'$ over the alphabet $A$. More formally, we put $h_A(a') = p_{a'}$ where $p_{a'}\colon A^{A'} \to A$ is the $a'$-th projection (i.e., $p_{a'}(f) = f(a')$).

To define the mapping $h_B\colon B^N \to B'$, we view an element of $B^N$ of the form $f^{\clo M^N}(h_A(a'_1),\dots,h_A(a'_n))$ as a~function from $N =A^{A'}$ to $B$ that is obtained as the~composition of the projection functions $p_{a'_i}\colon A^{A'} \to A$, for $i\in[n]$, with $f\colon A^n \to B$. Note that this function has at most $n$ essential variables because every $p_{a'_i}$ depends only on the coordinate $a'_i$. If $b\in B^N$ can be expressed as $b = f(p_{a'_1},\dots,p_{a'_n})$, then we define 
\(
  h_B(b) = \xi(f)(a'_1,\dots,a'_n),
\)
and for $b$ for which such $f$ and $a'_1,\dots,a'_n$ do not exist we define $h_B$ arbitrarily. 

We need to show that $h_B$ is well-defined, i.e., if $f(a'_1,\dots,a'_n) = g(b'_1,\dots,b'_m)$, then $\xi(f)(a'_1,\dots,a'_n) = \xi(g)(b'_1,\dots,b'_m)$. But this is true since the premise gives a~minor identity between $f$ and $g$ and $\xi$ preserves such identities. It is also easy to see that $h_A$ and $h_B$ satisfy (\ref{eq:birkhoff}). Thus we obtain that $\clo M' \in \ER \clo M^N$ which gives $\clo M' \in \ERPfin \clo M$.
\end{proof}

\begin{remark}
  The finiteness of sets $B$ and $B'$ is not used in the proof above, and therefore the theorem is true as stated as long as both $A$ and $A'$ are finite. Dropping even that assumption, following the same proof, one can show that the statement is true for minions over infinite sets if $\Pfin$ is replaced by $\Pow$.
\end{remark}

Finally, returning to polymorphism minions of PCSP templates, we can relate the above to Theorem~\ref{thm:minor-homomorphism-is-pp-constructibility}.

\begin{corollary} \label{cor:pp-constructible-minor-homomorphism}
  Let $(\rel A,\rel B)$ and $(\rel A',\rel B')$ be two templates and let $\clo M$ and $\clo M'$, respectively, be their polymorphism minions. The following are equivalent.
  \begin{enumerate}
    \item $(\rel A',\rel B')$ is pp-constructible from $(\rel A,\rel B)$ (consequently, $\PCSP(\rel A',\rel B')$ reduces in log-space to $\PCSP(\rel A,\rel B)$).
    \item $\clo M' \in \ERPfin \clo M$.
    \item There exists a~minion homomorphism $\xi\colon \clo M \to \clo M'$.
  \end{enumerate}
\end{corollary}

\begin{proof} The equivalence of (2) and (3) is given by the previous theorem. The equivalence of (1) and (3) follows from Theorem~\ref{thm:minor-homomorphism-is-pp-constructibility}. Alternatively we can argue that (1) is equivalent with the rest of the items in the following way: (1)\textto(3) is the easier direction, it follows easily from Lemma~\ref{lem:pp-gives-minor} (this mirrors the proof of Theorem~\ref{thm:minor-homomorphism-is-pp-constructibility}). We prove the other implication by proving that (2) implies (1):
Since $\clo M' \in \ERPfin \clo M$ there is a~minion $\clo M''\in \Pfin \clo M$, say that $\clo M'' = \clo M^n$, such that $\clo M' \in \ER \clo M''$.
Furthermore, this minion is the polymorphism minion of the template $(\rel A_n,\rel B_n)$ introduced in the proof of Lemma~\ref{lem:alg-rel-connection}(2).
Finally, Lemma~\ref{lem:alg-rel-connection}(3) concludes that $(\rel A',\rel B')$ is a~relaxation of a~structure pp-definable from $(\rel A_n,\rel B_n)$, and therefore it is pp-constructible from $(\rel A,\rel B)$. (In fact, this directly shows that $(\rel A',\rel B')$ is a~relaxation of a~pp-power of $(\rel A,\rel B)$.)
\end{proof}
\section{Bipartite minor conditions satisfied in \texorpdfstring{$\Pol(\rel K_k,\rel K_c)$}{Pol(K\_k,K\_c)}} \label{sec:graph_coloring}

The key in our approach to the complexity of PCSPs is the analysis
of bipartite minor conditions satisfied in their polymorphism minions.
While we cannot yet fully resolve the complexity of approximate graph colouring, in this section we provide some analysis and comparison of 
bipartite minor conditions satisfied in $\Pol(\rel K_k,\rel K_c)$ for various $k$ and $c$.
We remark that all hard CSPs satisfy the same (i.e., only trivial)
bipartite minor conditions, but this is not the case for PCSPs, as seen from our results in this section.

The hardness of $\PCSP(\rel K_k,\rel K_{c})$ with $c\le 2k-2$ was proved in \cite{BG16}, essentially by showing that, for such $c$, $\Pol(\rel K_k,\rel K_{c})$ has a minion homomorphism to the minion $\clo P_2$ of projections (though they did not use this terminology). In other words, all bipartite minor conditions satisfied in $\Pol(\rel K_k,\rel K_{c})$ for such $k$ and $c$ are trivial.  It was noted in \cite{BG16} that their approach cannot be extended to $c\ge 2k-1$ because then the polymorphisms satisfy some non-trivial minor conditions (we gave a specific example above, see Example~\ref{ex:cond-K3K5}).  We settled the case $c=2k-1$ by proving that the bipartite minor conditions satisfied in $\Pol(\rel K_k,\rel K_{c})$, though possibly non-trivial, are limited.  Specifically, all such conditions are also satisfied in $\Pol(\rel H_2,\rel H_K)$ for some $K$ because $\Pol(\rel K_k,\rel K_{2k-1})$ does not contain an Olšák function. We now show that this does not hold for $c\ge 2k$.

\begin{proposition} \label{prop:no-olsak-in-d-2d}
  $\Pol(\rel K_k,\rel K_{2k})$ contains an~Olšák function.
\end{proposition}

\begin{proof}
Define a~6-ary operation $o$ from $K_k$ to $K_{2k}$ in the following way:
  \[
    o(x_1,\dots,x_6) = \begin{cases}
      x_1 & \text{ if $x_1 = x_2$, or $x_1 = x_3$, }  \\
      x_2 & \text{ if $x_2 = x_3$, and } \\
      x_1 + k& \text{ otherwise. }
    \end{cases}
  \]
Note that $o$ depends essentially only on its first three variables.
It is straightforward to check that $o$ is an Olšák function. To see that it is a~polymorphism suppose that $x_1\neq y_1$, \dots, $x_6 \neq y_6$ while $o(x_1,\dots,x_6) = o(y_1,\dots,y_6) = a$. If $a$ is obtained using the last case of the definition of $o$, then $a = x_1 + k = y_1 + k$, and consequently $x_1 = y_1$, a~contradiction. If $a \in K_k$ then only the first two rows of the definition had to be used to compute the value for both $o(x_1,\dots,x_6)$ and $o(y_1,\dots,y_6)$ but since any two of the sets $\{1,2\}$, $\{1,3\}$, and $\{2,3\}$ intersect, we get that $a = x_i = y_i$ for any $i$ in the intersection, a contradiction again.
\end{proof}

We now show that Theorem~\ref{thm:no-epsilon-robust-is-hard} cannot be directly applied to to prove hardness of $\PCSP(\rel K_k,\rel K_{2k})$.

\begin{proposition}\label{prop:epsilon-robust-k,2k}
  For any $\epsilon > 0$, $\Pol(\rel K_k,\rel K_{2k})$ satisfies some $\epsilon$-robust bipartite minor condition.
\end{proposition}

\begin{proof}
Let $m$ be the ternary Boolean majority function, i.e., $m(x,y,z)$ outputs the repeated value among $x,y,z$. Further, let $m_1=m$ and, for $n\ge 1$, define 
\[
  m_n(x_1,\ldots,x_{3^n})=m_{n-1}(m(x_1,x_2,x_3), \ldots, m(x_{3^n-2},x_{3^n-1},x_{3^n})).
\]
For example, $m_2(x_1,\ldots,x_9)=m(m(x_1,x_2,x_3),m(x_4,x_5,x_6),m(x_7,x_8,x_9))$.

For a subset $I\subseteq \{1,\ldots,3^n\}$, let $\mathbf{x}_I$ be tuple in $\{0,1\}^{3n}$ such that $x_i=1$ if and only if $i\in I$.  Let $S_n=\{I\subseteq \{1,\ldots,3^n\}\mid m_n(\mathbf{x}_I)=1\}$, and let $MS_n$ be the set of all minimal (under inclusion) sets in $S_n$.  
Let $\Sigma_n$ be the bipartite minor condition which contains, for each $I\in MS_n$, the identity
\[
  f_n(x_0,y_1,\ldots,y_\ell)\equals g_n(x_{\pi_I(1)},\ldots,x_{\pi_I(3^n)})
\]
where 
\begin{itemize}
  \item $y_1,\ldots,y_\ell$ is the list of all variables of the form $x_{(J,i)}$ with $J\in MS_n$ and $1\le i\le 3^n$ such that $i\not\in J$, and
  \item $\pi_I(i)=0$ if $i\in I$ and $\pi_I(i)=(I,i)$ otherwise.
\end{itemize}
Note that the function symbols $f_n$ and $g_n$ are the same in all identities in $\Sigma_n$.

Fix any $\epsilon>0$ and choose $n$ so that $\epsilon>(2/3)^n$. We will show that $\Sigma_n$ is $\epsilon$-robust. Note that each variable of the form $x_{(I,i)}$ appears on the right-hand side of only one identity in $\Sigma_n$ (namely, the one corresponding to $I$).
Assume first that $f_n$ is assigned a projection not on the first variable (i.e., not on $x_0$). Then this projection corresponds to some variable of the form $x_{(I,i)}$ on the left-hand side, and any assignment of a projection to $g_n$ would satisfy at most one identity in $\Sigma_n$ (again, the one corresponding to $I$), which is certainly less than an $\epsilon$-fraction of all identities in $\Sigma_n$.  So assume that $f_n$ is assigned a projection on the first variable $x_0$. It is not hard to check that, for a fixed $1\le i\le 3^n$, the probability that $\pi_I(i)=0$ (i.e., that $i\in I$) for a randomly chosen set $I\in MS_n$ is $(2/3)^n$.
Thus any assignment of a projection to $g_n$ satisfies less than an $\epsilon$-fraction of identities in $\Sigma_n$ and we conclude that $\Sigma_n$ is $\epsilon$-robust.

It remains to show that $\Sigma_n$ is satisfied in $\Pol(\rel K_k,\rel K_{2k})$.
Define the following functions on $\rel K_k$:
\[
  g_n(x_1,\dots,x_{3^n}) = \begin{cases}
    a & \text{ if $\{i\mid x_i=a\}\in S_n$ for some $a$,} \\
    x_1+k & \text{ otherwise. } \\
  \end{cases}
\]
Note that, by the choice of $S_n$, the condition in the first case of the definition of $g_n$ cannot hold for more than one value $a$, so this function is well-defined. Let $f_n(x_0,y_1,\ldots,y_\ell)=x_0$.  It is straightforward to check that $f_n,g_n\in\Pol(\rel K_k,\rel K_{2k})$ and that these functions satisfy all identities in $\Sigma_n$.
\end{proof}

We now show that, for any $k\ge 3$, there is a bipartite minor condition that is satisfied in $\Pol(\rel K_k,\rel K_c)$ for some $c$, but not in $\Pol(\rel K_{k'},\rel K_{c'})$ for any $k<k'\le c'$.  This means that even if one proves \NP-hardness of $\PCSP(\rel K_{k'},\rel K_{c'})$ for some \emph{fixed} $k'>3$ and \emph{all} $c'\ge k'$, Theorem~\ref{thm:main} would not immediately imply the same for any fixed $k$ with $3\le k<k'$. 

\begin{proposition}\label{prop:no-hom-k,c-to-k',c'}
  For each $k\geq 3$, there exists $c\geq k$ such that $\Pol(\rel K_k,\rel K_c)$ has no~minion homomorphism to $\Pol(\rel K_{k'},\rel K_{c'})$ for any $k<k'\le c'$.
\end{proposition}

\begin{proof}
  A~simple bipartite minor condition that is not satisfied in $\Pol (\rel K_{k'},\rel K_{c'})$ is the $\rel K_{k+1}$-loop condition --- this follows from the fact that $\rel K_{k'}$ contains $\rel K_{k+1}$, but $\rel K_{c'}$ has no loop. The $\rel K_{k+1}$-loop condition is described by the following two identities
  \begin{align*}
    t(x_1,\dots,x_{k+1}) &\equals s( x_1, x_2, \dots, x_k,x_{k+1} )   \\
    t(x_1,\dots,x_{k+1}) &\equals s( x_2, x_1, \dots, x_{k+1},x_{k} ).
  \end{align*}
  The variables on the right-hand side correspond to the edges of $\rel K_{k+1}$ (viewed as a digraph), we will index them accordingly: the first one is $x_{1,2}$, the second $x_{2,1}$, and so on, so the arity of $s$ is $k(k+1)$. We claim that this condition is satisfied in $\Pol(\rel K_k,\rel K_c)$ for some $c\geq k$. To show that, we define functions $t$ and $s$ without regard to the size of their range by picking a~new colour whenever we are not forced by any of the identities to do otherwise. More precisely, we set
  \[
    t(a_1,\dots,a_{k+1}) = (a_1,\dots,a_{k+1}).
  \]
  In other words, if we view $t$ as a colouring of $\rel K_k^{k+1}$ then each vertex gets its own colour. Further, let
  \[
    s(a_{1,2},\dots,a_{k+1,k}) = \begin{cases}
      t(a_1,\dots,a_{k+1}) & \text{ if $a_{i,j} = a_i$ for all $i\neq j$, } \\
      t(a_1,\dots,a_{k+1}) & \text{ if $a_{i,j} = a_j$ for all $i\neq j$, and } \\
      (a_{1,2},\dots,a_{k+1,k}) & \text{ otherwise. } \\
    \end{cases}
  \]
  If we view $s$ as a colouring of $\rel K_k^{k(k+1)}$, then the last case in the definition of $s$ assigns each vertex (satisfying the condition of the case) its own colour, and that colour is not used in $t$.
  Note that $s$ is well defined, because if both $a_{i,j} = a_i$ and $a_{i,j} = a'_j$ for all $i\neq j$, then $a_{i,j} = a_{i',j'}$ for all $i\neq j$ and $i'\neq j'$, and in that case both rows give the same value.  It is straightforward to check that these functions satisfy the $\rel K_{k+1}$-loop condition, since the tuples that require checking are exactly those in the first two rows of the definition of $s$.
  We also claim that both $t$ and $s$ are proper colourings. This is obvious for $t$. For $s$, we have to show that if
  \(
    s(a_{1,2},\dots,a_{k+1,k}) = s(b_{1,2},\dots,b_{k+1,k}),
  \)
  then $a_{i,j} = b_{i,j}$ for some $i\neq j$. Clearly, if the two values of $s$ agree, then the resulting colours must have been obtained using one, or both of the first two rows of definition of $s$. If the same row was used for both, there is nothing to prove as the resulting colour is uniquely determined by the tuples $(a_1,\dots,a_{k+1})$ and $(b_1,\dots,b_{k+1})$, respectively. The remaining case is when both rows were used. Since the situation is symmetric, we may assume that $a_{i,j} = a_i$ and $b_{i,j} = b_j$ for all $i\neq j$. Further since the result is $t(a_1,\dots,a_{k+1}) = t(b_1,\dots,b_{k+1})$ we have that $a_i = b_i$ for all $i$. Finally, since there are only $k$ possible values for $a_i$, we have that $a_i = a_j$ for some $i\neq j$ which implies that $a_{i,j} = a_i = a_j = b_j = b_{i,j}$.
  The total number $c$ of colours used in the range of $t$ and $s$ is bounded by $k^{k+1} + k^{k(k+1)}$.
\end{proof}

The above proposition does not give the best possible value of $c$. In fact, it was shown by Olšák \cite{Ols18} that for $k=3$, a~possible value is $c=6$. We include a~proof for completeness.

\begin{proposition}[\cite{Ols18}]
  $\Pol(\rel K_3,\rel K_6)$ does not map by a~minion homomorphism to $\Pol(\rel K_{k'},\rel K_{c'})$ for any $3<k'\leq c'$.
\end{proposition}

\begin{proof}
  As in the case of the previous proposition, the distinguishing bipartite minor condition is the $\rel K_4$-loop condition. Again, it is clear to see that it is not satisfied in $\Pol(\rel K_{k'},\rel K_{c'})$, so it is enough to define two polymorphisms from $\rel K_3$ to $\rel K_6$ that satisfy this condition. We consider:
  \begin{align*}
    t (x_1,x_2,x_3,x_4) &= \begin{cases}
      x_1 & \text{if $x_1 = x_2 = x_3$, and} \\
      x_4 + 3 & \text{otherwise;}
    \end{cases} \\
    s (x_{1,2},x_{2,1},x_{1,3},\dots,x_{4,3}) &= \begin{cases}
      t(x'_1,\dots,x'_4) & \text{if $x_{i,j} = x'_i$ for all $i\neq j$,} \\
      t(x'_1,\dots,x'_4) & \text{if $x_{i,j} = x'_j$ for all $i\neq j$, and} \\
      x_{1,2} & \text{otherwise}
    \end{cases}
  \end{align*}
  (to ease readability we label variables of $s$ by edges of $\rel K_4$).  Clearly, $s$ and $t$ satisfy the required identities. It is straightforward to check that $t$ is a~polymorphism; let us then prove that $s$ is. As before, it is enough to show that if
  \[
    s (x_{1,2},\dots,x_{3,4}) =
    s (y_{1,2},\dots,y_{3,4}).
  \]
  for some $x_{i,j}$'s and $y_{i,j}$'s then there are $i\neq j$ such that $x_{i,j} = y_{i,j}$. Let us assume that the above common value of $s$ is $a$, and consider two cases:
  First, $a\in K_3$. Such value can be obtained for $x$'s if $x_{1,2} = a$, or $x_{i,j} = a$ for all $i\neq j$ and $i \in \{1,2,3\}$, or $x_{i,j} = a$ for all $i\neq j$ and $j\in \{1,2,3\}$. Either way, $x_{1,2} = a$. By the same argument, we get $y_{1,2} = x_{1,2} = a$.
  Second, $a \in K_6\setminus K_3$. This implies that one of the two rows of the definition of $s$ have been used for both values. 
  If the first row have been used for both, we have $x_{4,i}=y_{4,i}=a-3$ for all $i\neq 4$; similarly, if the second row have been used for both, then $x_{i,4}=y_{i,4}=a-3$ for all $i\neq 4$.
  Hence, the first row was used for one value and the second for the other. Using the symmetry, we may assume that the first row was used for the $x$'s and the second row for the $y$'s, that is
  \begin{align*}
    a=x_4 + 3 &= t(x_1,x_2,x_3,x_4)= s( x_1, x_2, x_1, x_3, x_1, x_4, x_2, x_3, x_2, x_4, x_3 , x_4 ) \\
    a=y_4 + 3 &= t(y_1,y_2,y_3,y_4)= s ( y_2, y_1, y_3, y_1, y_4, y_1, y_3, y_2, y_4, y_2, y_4 , y_3 ).
  \end{align*}
Since the second row of the definition of $t$ was used, we know that $y_1,y_2,y_3$ are not all equal, so assume by symmetry that $y_2 \neq y_3$. Note that, we now need to show that $x_i = y_j$ for some $i \neq j$. If it is not the case, then $y_4 = x_4 \neq y_2$ and also $y_4 = x_4 \neq y_3$ which means that $y_2$, $y_3$ and $y_4$ take three different values from $K_3$, and consequently $x_1$ must be equal to one of them, a contradiction.
\end{proof}

Note that, for any fixed $k\ge 3$, the case $c=2k$ is the smallest one for which \NP-hardness of $\PCSP(\rel K_k,\rel K_c)$ is still open, unless $k$ is large enough for Huang's result~\cite{Hua13} to apply.  The smallest (in terms of $k,c$) open case is $\PCSP(\rel K_3,\rel K_6)$, and the above result shows that even if one proves \NP-hardness of $\PCSP(\rel K_{k'},\rel K_{c'})$ for all $4\le k'\le c'$, Theorem~\ref{thm:main} would not immediately imply \NP-hardness of $\PCSP(\rel K_3,\rel K_6)$. 

Finally, we prove is that if we fix $k$ and increase $c$ then the family of bipartite minor conditions satisfied in $\Pol(\rel K_k,\rel K_c)$ also grows.

\begin{proposition}
  For any $c\geq k>2$, there exists $C>c$ such that $\Pol(\rel K_k,\rel K_C)$ has no~minion homomorphism to $\Pol(\rel K_k,\rel K_c)$.
\end{proposition}

\begin{proof}
  To simplify the notation, let $\clo K_n = \Pol(\rel K_k,\rel K_n)$. We need to find a~bipartite minor condition which is not satisfied in $\clo K_c$, but is satisfied in $\clo K_C$ for some $C>c$.
  A~simple condition that is not satisfied in $\clo K_c$ is the $\rel K_k$-loop condition. However, this condition is also not satisfied in $\clo K_C$ for any $C$ either. Nevertheless, Theorem~\ref{thm:minionhom-and-pp}(3)$\leftrightarrow$(1) implies that there is a~condition of the form $\Sigma(\rel K_k,\rel F)$ where $\rel F$ is some graph, that witnesses that there is no~minion homomorphism. We consider the condition $\Sigma(\rel K_k,\rel K_{c+1})$. Recall that $\Sigma(\rel K_k,\rel K_{c+1})$ is constructed as follows: Introduce a $k$-ary symbol $f_u$ for each $u\in [{c+1}]$, and for each $\{u,v\} \subseteq [{c+1}]$, $u < v$ a~$k(k-1)$-ary symbol $g_{\{u,v\}}$ and all identities of the following form:
  \begin{align*}
    f_u(x_1,\dots,x_k) &\equals g_{\{u,v\}}( x_{1}, x_2, \dots, x_{k-1} , x_{k} )   \\
    f_v(x_1,\dots,x_k) &\equals g_{\{u,v\}}( x_{2}, x_1, \dots, x_{k}, x_{k-1} ) 
  \end{align*}
  where the variables on the right-hand side are distributed in such a~way that $x_i$ and $x_j$ appear together in some column for all $i\neq j$.
  First, we claim that this condition is not satisfied in $\clo K_c$. This follows directly from Lemma~\ref{lem:pcsp-to-lc}(2): if it was, then there would be a~homomorphism from $\rel K_{c+1}$ to $\rel K_c$.

We now show that this bipartite minor condition is satisfied in $\clo K_C$ for some $C$. As in the proof of Proposition~\ref{prop:no-hom-k,c-to-k',c'}, we define functions first and determine $C$ later.
Define functions $f_u$ and $g_{\{u,v\}}$ without regard to the size of the range by setting $f_u(a,\dots,a) = a$, $g_{\{u,v\}}(a,\dots,a) = a$, and then picking a~fresh colour whenever we are not forced by any of the identities to do otherwise. More precisely, we set
\[
  f_u(a_1,\dots,a_k) = \begin{cases}
    a & \text{ if $a_1 = \dots = a_k = a$, and } \\
    (u;a_1,\dots,a_k) & \text{ otherwise, } \\
  \end{cases}
\]
and similarly,
\[
  g_{u,v}(a_{1,2},\dots,a_{k,k-1}) = \begin{cases}
    f_u(a_1,\dots,a_k) & \text{ if $a_{i,j} = a_i$ for all $i\neq j$, } \\
    f_v(a_1,\dots,a_k) & \text{ if $a_{i,j} = a_j$ for all $i\neq j$, and } \\
    (u,v;a_{1,2},\dots,a_{k,k-1}) & \text{ otherwise. } \\
  \end{cases}
\]
Note that $g_{u,v}$ is well defined, because if both $a_{i,j} = a_i$ and $a_{i,j} = a'_j$ for all $i\neq j$, then $a_{i,j} = a_{i',j'}$ for all $i\neq j$ and $i'\neq j'$, and in that case both rows give the same value. It is straightforward to check that these functions satisfy the minor condition.
We also claim that both $f_u$ and $g_{u,v}$ are proper colourings, but that is easy to see since for a~fixed $u$ ($\{u,v\}$, respectively) no colour is used twice for a~value of $f_u$ ($g_{\{u,v\}}$, respectively). The total number $C$ of colours used is bounded by $k + (c+1)k^k + (c+1)^2k^{k(k-1)}$.
\end{proof}

We remark that $\Sigma(\rel K_3,\rel K_4)$ is not satisfied in $\Pol(\rel K_3,\rel K_4)$ since it is a~non-trivial condition (follows from Lemma~\ref{lem:pcsp-to-lc}(2)) and $\Pol(\rel K_3,\rel K_4)$ maps to $\clo P_2$ by a~minion homomorphism (see Example~\ref{ex:pol-3,4-maps-to-projections}). So in the previous proof, the obvious lower bound on $C$ cannot be always met.

\section{Conclusion}
  \label{sec:conclusion}
  
This paper deals with the Promise CSP framework, which is a~significant generalisation of the (finite-domain) CSP. The PCSP provides a~nice interplay between the study of approximability and universal-algebraic methods in computational complexity. We presented a~general abstract algebraic theory that captures the complexity of PCSPs with a~fixed template $(\rel A,\rel B)$. The key element in our approach is the bipartite minor conditions satisfied in the polymorphism minion $\Pol(\rel A,\rel B)$ of a~template. We have shown that such conditions determine the complexity $\PCSP(\rel A,\rel B)$. We gave some applications of our general theory, in particular, in approximate graph colouring. 

The complexity landscape of PCSP (beyond CSP) is largely unknown, even in the Boolean case (despite some progress in~\cite{BG18b,FKOS19}), and includes many specific problems of interest. We hope that our theory will provide the basis for a~fruitful research programme of charting this landscape. Below we discuss some of the possible directions within this programme.

Let us first discuss how the complexity classification quest for PCSPs compares with that for CSPs. As we said above, the gist of the algebraic approach is that lack or presence of (high-dimensional) symmetries determines the complexity. For CSPs, there is a~sharp algebraic dichotomy: having only trivial symmetries (i.e., satisfying only those systems of minor identities that are satisfied in polymorphisms of every CSP) leads to \NP-hardness, while  any non-trivial symmetry  implies rather strong symmetry and thus leads to tractability. Moreover, the algorithms for tractable cases are (rather involved) combinations of only two basic algorithms --- one is based on local propagation \cite{BK14} and the other can be seen as a~very general form of Gaussian elimination \cite{IMMVW10}.
It is already clear that the situation is more complicated for PCSPs: there are hard PCSPs with non-trivial (but limited in some sense) symmetries, and tractable cases are more varied \cite{AGH17,BG18,BG18b,DRS05}.  This calls for more advanced methods, and we hope that our paper will provide the basis for such methods.
There is an obvious question whether PCSPs exhibit a~dichotomy as CSPs do, but there is not enough evidence yet to conjecture an answer. More specifically, it is not clear whether there is any PCSP whose polymorphisms are not limited enough (in terms of satisfying systems of minor identities) to give \NP-hardness, but also not strong enough to ensure tractability. Classifications for special cases such as Boolean PCSPs and graph homomorphisms would help to obtain more intuition about the general complexity landscape of PCSPs, but these special cases are currently open. 

The sources of hardness in PCSP appear to be much more varied than in CSP (that has a~unique such source), and much remains to be understood there. What limitations on the bipartite minor conditions satisfied in polymorphisms lead to \NP-hardness?  We gave some general and some specific results in this direction, but it is clear that our general theory needs to be further developed.
Currently, variants of the Gap Label Cover provide the source of hardness, but it is possible that new versions of \GLC{} may need to be used. It is not clear in advance what these versions would be --- their form would be dictated by the analysis of polymorphisms and minor conditions. However, it would be interesting to eventually go even further and avoid dependency on deep approximation results such as the PCP theorem and parallel repetition, which are more appropriate for quantitative approximation concerning the number of (un)satisfied constraints. Instead, one would aim to provide a~self-contained theory that includes a~new type of reduction between PCSPs so that one could bypass the PCP Theorem and parallel repetition altogether. In particular, can one come up with purely algebraic reductions (e.g.\ by extending pp-constructions) that create and amplify the algebraic gap in problems $\PMC_{\clo M}(N)$, which is what the PCP theorem and parallel repetition do for the quantitative gap in Gap Label Cover? 

The analysis of polymorphisms of approximate graph colouring problems (and their relatives) may provide further intuition as to what limitations on minor conditions  can be used for hardness proofs.  
To give another specific problem, it may be interesting to understand for what structures $\rel B$ the problem $\PCSP(\rel T,\rel B)$ is hard.  Here $\rel T$ is the Boolean ``1-in-3'' structure from Example~\ref{example:1in3-nae} and $\rel B$ is a~(not necessarily Boolean) structure with a~single ternary relation. We know that the problem is \NP-hard for $\rel B=\rel T$ and tractable for $\rel B=\rel H_2$. 

What algorithmic techniques are needed to solve tractable PCSPs?  One general approach to proving tractability of $\PCSP(\rel A,\rel B)$ is presented in~\cite{BG18b}. The main idea is to find a~structure $\rel D$ such that $\rel A\rightarrow \rel D\rightarrow \rel B$ and $\CSP(\rel D)$ is tractable.  Such structures are called ``homomorphic sandwiches'' in~\cite{BG18b}. It is clear that an algorithm for $\CSP(\rel D)$ solves $\PCSP(\rel A,\rel B)$. In general, $\rel D$ may have infinite domain (but instances of $\CSP(\rel D)$ are still finite). Indeed, for $\PCSP(\rel T,\rel H_2)$ from Example~\ref{example:1in3-nae} no such finite $\rel D$ exists (Theorem~\ref{thm:infinity-necessary}), but there are several infinite ones \cite{BG18,BG18b,Bar19}, such as $\rel D=(\mathbb{Z}; x+y+z=1)$. 
It would be interesting to find out what other tractable PCSPs (e.g., those considered in\cite{BG18}, \cite{FKOS19}) require infinite-domain CSPs. Another natural question is ``how infinite'' such a template $\rel D$ needs to be --- while the class of all infinite-domain CSPs essentially covers all decision problems~\cite{BG08}, some parts of the theory can be generalised from finite structures to the rich class of so called $\omega$-categorical structures~\cite{Bod08,Pin15}. Can a structure $\rel D$ such that $\rel A \rightarrow \rel D \rightarrow \rel B$ be $\omega$-categorical, say, for $(\rel A, \rel B) = (\rel T,\rel H_2)$?
It would be also very interesting to develop a~general theory of how such a~structure $\rel D$ can be constructed from bipartite minor conditions satisfied in $\Pol(\rel A,\rel B)$ and what properties of such conditions guarantee tractability of $\CSP(\rel D)$.  
Overall, it is another interesting feature of the (finite-domain) PCSP framework that it seems to require research into infinite-domain CSP. 

The polynomial-time algorithms of Bulatov and Zhuk \cite{Bul17,Zhu17} that solve all tractable finite-domain CSPs are quite involved and both use deep structural analysis of finite algebras. Can algorithms for tractable PCSPs that involve infinite-domain CSPs (e.g. such as those in \cite{BG18b}) lead to a new simpler algorithm that solves all tractable CSPs?

Every tractable finite-domain CSP can be solved \cite{Bul17,Zhu17} by a (rather involved)~combination of local consistency checking~\cite{BK14} and the ``few subpowers'' algorithm that handles compact representation of constraints~\cite{IMMVW10}.  The study of local consistency checking algorithms generally played a~very important role in the algebraic theory of CSP. Apart from being one of the two main algorithmic approach, it connected, via homomorphism dualities~\cite{BKL08}, the algebraic theory of CSP with the combinatorial~\cite{HN04} and logical~\cite{KV08} approaches.  As discussed in Section~\ref{sec:tractable}, it makes sense to use local consistency checking for $\PCSP(\rel A,\rel B)$.  Although it is obviously an interesting problem to characterise PCSPs solvable by local consistency, it is not clear yet whether this method still plays a~key role for PCSP or must be superseded by more powerful algorithms in this context.  The study of duality for CSP is intimately connected with investigating CSPs inside $\Ptime$, e.g.\ those in \NL, \Lspace, or $\AC^0$ --- it is natural to extend this line of research to PCSP.  There are many other (both open and resolved) questions about local consistency for CSP that can be investigated for PCSP.

Unlike with the local consistency method, it is not even clear how to transfer the ``few subpowers'' algorithm to the realm of PCSP. It may be that a~different algorithmic approach is required that will handle all that ``few subpowers'' could and that can be used for PCSP.

There are many results in CSP exactly characterising the power of a~given algorithm. We generalised some of such results to PCSP in Section~\ref{sec:tractable}, but this line of research will need to continue, especially as new algorithmic techniques are being developed for PCSP.

Finally, the PCSP framework itself can be extended by incorporating approximation based on counting (un)satisfied constraints, see~\cite{AH13} for an example --- this would give another range of interesting open questions, but now we will not discuss this direction further.

  \subsection*{Acknowledgements}
  We would like to thank Josh Brakensiek, Venkat Guruswami, Alex Kazda, Mirek Olšák, Erkko Lehtonen, Manuel Bodirsky, and Marcello Mamino for valuable discussions. We also thank referees of the conference version for useful suggestions.

  \bibliographystyle{alpha}
\newcommand{\etalchar}[1]{$^{#1}$}

\end{document}